\documentclass[11pt]{article}
\usepackage[dvips]{graphicx}
\graphicspath{{images/}}
\usepackage{verbatim}
\usepackage{amsmath,amsthm, amsfonts}
\usepackage{xspace}
\usepackage{pifont}
\usepackage{graphicx}
\usepackage{amssymb}
\usepackage{epic, eepic}
\usepackage{dsfont}
\usepackage{amssymb}
\usepackage{anyfontsize}
\usepackage{makeidx}
\usepackage{mathrsfs}
\usepackage{arydshln}
\usepackage{exscale}
\usepackage[dvipsnames]{xcolor}
\usepackage{overpic} 
\usepackage{bm}
\usepackage{bbm}
\usepackage{booktabs} 
\usepackage{color, colortbl}
\usepackage{subcaption}
\usepackage{float}
\usepackage[numbers]{natbib}
\usepackage[title]{appendix}
\usepackage{booktabs}
\usepackage{makecell}
\usepackage{enumerate}
\usepackage{soul}
\usepackage{algorithm2e}

\RequirePackage[colorlinks,citecolor=blue,urlcolor=blue]{hyperref}

\definecolor{Gray}{gray}{0.9}

\usepackage{amsmath,afterpage}
\usepackage{epsf}
\usepackage{graphics,color}
\RequirePackage[colorlinks,citecolor=blue,urlcolor=blue]{hyperref}
\usepackage[export]{adjustbox}

\usepackage{tcolorbox}

\def\0{\mathbf{0}}

\def\lam{\lambda}

\def \< {\langle}
\def \> {\rangle}

\def\beqa{\begin{eqnarray}}
\def\eeqa{\end{eqnarray}}
\def\beqas{\begin{eqnarray*}}
\def\eeqas{\end{eqnarray*}}

\def\red#1{\textcolor{red}{#1}}

\newtheorem{theorem}{Theorem}[section]
\newtheorem{lemma}[theorem]{Lemma}

\newtheorem{proposition}[theorem]{Proposition}

\newtheorem{remark}[theorem]{Remark}

\newtheorem{definition}[theorem]{Definition}
\newtheorem{assumption}[theorem]{Assumption}
\numberwithin{equation}{section}
\newcommand{\hatd}[1]{{}}

\newcommand{\bd}{\begin{displaymath}}
\newcommand{\ed}{\end{displaymath}}
\newcommand{\be}{\begin{equation}}
\newcommand{\ee}{\end{equation}}
\newcommand{\bq}{\begin{eqnarray}}
\newcommand{\eq}{\end{eqnarray}}
\newcommand{\bn}{\begin{eqnarray*}}
\newcommand{\en}{\end{eqnarray*}}
\newcommand{\dl}{\delta}

\def\wt{\widetilde}

\def\E{{\mathbb{E}}}

\usepackage[normalem]{ulem}

\usepackage{mathtools}
\DeclarePairedDelimiter{\norm}{\lVert}{\rVert}
\mathtoolsset{showonlyrefs}

\usepackage{authblk}

\title{Competition in Dealer Markets with Internalisation and Externalisation
\footnotetext{\scriptsize{\hskip-0.53cm Robert Boyce is supported by the EPSRC Centre for Doctoral Training in Mathematics of Random \mbox{Systems}: Analysis, Modelling and Simulation (EP/S023925/1).}}} 

\usepackage{geometry}
\DeclareMathAlphabet{\altcal}{OMS}{zplm}{m}{n}

\usepackage[bottom]{footmisc}

\author[1]{Robert Boyce}
\author[1]{Eyal Neuman}
\affil[1]{Department of Mathematics, Imperial College London}

\begin{document}

\vspace{-0.5cm}
\maketitle

\begin{abstract}
\noindent
We model a market with multiple dealers who compete for client order flow by dynamically updating their bid and ask quotes for a risky asset. Dealers aim to maximise expected profits while controlling inventory risk by skewing their quotes to attract offsetting order flow (internalisation) or by directly offloading positions in the market (externalisation). Using a variational approach, we derive a closed-form equilibrium for the resulting Nash competition, shedding light on key features of dealer market dynamics. We show that dealers relying on internalisation are compelled to increase their externalisation activity when competing with externalising dealers. This strategic shift in equilibrium leads to significantly higher hedging costs for all dealers and substantially wider spreads for clients.
\end{abstract} 

\begin{description}
{\small \item[Mathematics Subject Classification (2010):] 
91G10, 49N10, 49N90, 93E20 
\item[JEL Classification:] C61, C73, G11, G24, G32
\item[Keywords:] market making, Nash equilibrium, optimal liquidation, price impact
}
\end{description}


\section{Introduction}

In modern financial markets, the distinction between liquidity providers and liquidity takers is becoming increasingly blurred. Liquidity providers, also known as dealers, set prices at which they are willing to trade with clients but are typically averse to accumulating large positions due to the risk of adverse price movements, known as inventory risk. The practice of holding inventory in anticipation of a future offsetting trade is referred to as \emph{internalising}. Conversely, dealers may act as liquidity takers, incurring transaction costs to rapidly unwind positions; this is known as \emph{externalising}. A dealer’s risk preferences play a key role in determining whether to internalise or externalise a trade, but this decision may also depend on the behaviour of other dealers trading the same asset.

The behaviour of dealers who exclusively internalise their trades, i.e. who warehouse client orders and wait for offsetting trades to passively unwind their inventory, has been extensively studied in the literature. The seminal work of \citet{avellaneda2008high} introduced a market-making model in which a dealer controls the spread to optimise the trade-off between revenues and risk. An explicit optimal strategy for this model was derived in \citet{gueant2013solution}, and numerous extensions have since been proposed (see, e.g., Chapter 10 of \cite{cartea2015algorithmic} and \citep{cartea2014buy, gueant2017optimal, cartea2017algorithmic, boyce2025market}). \citet{butz2019internalisation} developed a model of internalisation using queueing theory, deriving closed-form expressions for the internalisation horizon and analysing the associated costs. \citet{guo2026macroscopic} study a multi-agent market-making Nash competition in which dealers exclusively internalise order flow. They characterise the solution in terms of a system of FBSDEs and provide sufficient conditions for the existence of a solution to the system.

Dealers also have the option to externalise their accumulated inventory in the market in order to reduce their exposure to risk. This process may create additional transaction costs due to the price impact of their orders. The trade-off between internalisation and externalisation has been studied in the context of interactions between a broker and a client \citep{bergault2025mean,cartea2025brokers,cartea2026nash,donnelly2025liquidity}. The problem of a central risk book in this context was studied by \citep{barzykin2024unwinding,nutz2023unwinding}. Both papers adopt an optimal liquidation framework (see \citep{lehalle2019signals,neuman2022optimal,neuman2023trading,obizhaeva2013dynamics}), where the dealer controls only the externalisation rate. \citet{barzykin2023algorithmic}, proposed an extension to the Avellaneda-Stoikov model, in which the dealer not only controls the spread but can also externalise inventory. 

When multiple liquidity providers compete for clients’ order flow, the externalisation behaviour of some providers can reduce the portfolio value of others through the adverse effects of their aggregated price impact on the asset price. \citet{oomen2017aggregator} proposed a preliminary model that captures key features of this phenomenon by considering a static game between two internalisers and one externaliser competing bilaterally for a single client’s order flow (see Section 5 therein). In this model, each dealer estimates the fundamental price and posts a constant, non-controlled spread, with the option to immediately externalise, thereby inducing a permanent price impact. One of the main results of \citet{oomen2017aggregator} is that the presence of a single externaliser leads other dealers to externalise as well. This effect was subsequently documented empirically in Section 3.3 of \citet{oomen2019price}, in particular in Figure 13 and Table 3 therein.

Some limitations of the model in \cite{oomen2017aggregator} are that it does not account for inventory or inventory risk, dealers cannot control the bid–ask spread, and the framework is static. As a result, important features of the system are not captured. For example, when internalisers hold large inventories, the price impact generated by externalisers’ trades may move prices in an adverse direction, reducing the value of internalisers’ positions. This mechanism is one of the main reasons why internalisers may choose to externalise in such settings. Moreover, allowing all dealers to dynamically control their spreads can shed light on how the presence of externalisers affects transaction costs for clients.

The main objective of this work is to propose a dynamic model that quantifies how a given dealer is affected by the risk management decisions of competitors when the system is at equilibrium. Furthermore, the model illustrates how these externalisation decisions adversely affect the prices offered to clients. To this end, we formulate a continuous-time model of a market with multiple dealers who provide liquidity in a risky asset and compete for client order flow by quoting bid and ask prices. Each dealer controls both their bid–ask quotes and their rate of externalisation. Differences across dealers primarily stem from their sensitivity to inventory risk, which influences their urgency to externalise. Dealers aim to maximise profits from trading with clients via the bid–ask spread while managing inventory risk, either by skewing their quotes or by externalising in the broader market.

Interactions among dealers arise both through competition for client order flow via pricing strategies and through the aggregate price impact of their trades. We consider a dynamic Nash competition framework and, using a variational approach, derive the first-order conditions in the form of a system of forward–backward stochastic differential equations (FBSDEs). We solve this system explicitly and obtain the unique open-loop Nash equilibrium. Finally, we conduct numerical experiments to analyse market behaviour under different compositions of internalisers and externalisers.

\paragraph{Main contribution.} 
Previous works, as mentioned earlier, either consider only single-agent internalisation and externalisation problems, as in \citet{barzykin2023algorithmic}, or focus on purely internalisation, as in \citet{guo2026macroscopic}. Our work integrates a multiplayer framework with internalisation–externalisation strategies for each agent, where agents simultaneously compete for client order flow and influence the market through their collective price impact. These additional interactions significantly alter the results, as they create a feedback loop between quoting behaviour and inventory risk management across dealers.
A central contribution of our work is Theorem \ref{thm:solution}, in which we derive a closed-form solution for the Nash equilibrium in this dynamic, competitive environment. This result not only advances the theoretical understanding of multi-dealer liquidity provision under inventory risk but also provides a rigorous explanation for the empirical findings in Section 3.3 of \citet{oomen2019price}, in particular regarding how externalisation by other liquidity providers can adversely affect prices. Furthermore, our analysis offers insights into how the presence and relative proportions of internalisers and externalisers influence price dynamics, execution costs, and overall market liquidity, as detailed below.

\begin{enumerate}[\textbf{(i)}]

    \item \textbf{More externalisers make internalisers behave like externalisers.}
When native externalisers, who have a higher priority for inventory liquidation, are present in the market, native internalisers with greater tolerance for inventory risk are incentivised to unwind positions more aggressively and to externalise more. This is because the execution order flow of the externalisers acts as a signal that internalisers can exploit in their trading strategies to generate additional profits. This effect leads to higher price impact for all agents and, consequently, increases inventory hedging costs. In Table \ref{tab:keymetrics-unconstrained}, we show that the average hedging volume of the internalisers increases by a factor of $14$ when a single internaliser in the liquidity pool is replaced by an externaliser, while their average hedging costs increase by a factor of $8$.

    \item [\textbf{(ii)}] \textbf{More externalisers increase internalisers’ effective spread capture.} 
    As the ratio of externalisers to internalisers in the pool increases, the remaining internalisers can increase their effective spread capture. Indeed, a substantial inflow on one side of the book will cause externalisers to aggressively skew their spreads in order to maintain low inventory, while internalisers, who face less restrictive inventory costs, will be able to offer more competitive spreads in the direction of the inflow. See Tables \ref{tab:keymetrics-unconstrained} and \ref{tab:keymetrics-constrained}.

    \item [\textbf{(iii)}] \textbf{More externalisers can worsen execution prices for the client.} While the execution prices of the median client can initially improve slightly by diversifying dealer types, adding more externalisers quickly negates this effect and increases costs for the typical (median) client. A net client inflow on one side of the book will cause externalisers to unwind positions aggressively due to inventory risk, leading to price distortion in a direction adverse to the client’s position. See Table \ref{tab:keymetrics-unconstrained}, where the median client's costs defined in \eqref{def:median_client_best_spread} gradually increase by over $20\%$ as internalisers are replaced by externalisers.  
    
    \item [\textbf{(iv)}] \textbf{If dealers are order-flow agnostic, the above results are emphasised, and internalisers’ P\&L decreases as the number of externalisers increases.} In the case where dealers are agnostic to the order flow generated by others’ executions due to regulation or informational inaccuracy, and therefore do not account for aggregated price impact, their hedging costs rise much more sharply in its presence, while median client spread costs strictly increase as externalisers enter the market. Moreover, the increased hedging costs outweigh the benefits of improved effective spread capture, leading to a decline in internalisers’ P\&L.

Specifically, when a single internaliser in a pool of six is replaced by an externaliser, the average hedging volume and costs of the remaining internalisers increase by factors of $4$ and $15$, respectively, and client spreads rise with the number of externalisers, eventually by over $20\%$. Since internalisers are individually incentivised to react to the externalisers’ actions in this way, yet are collectively worse off, this can be interpreted as a form of prisoner’s dilemma. This is a similar effect to the prisoner's dilemma demonstrated in the static model in Section 5 of \cite{oomen2017aggregator}.
\end{enumerate}

\paragraph{Organisation of this paper:}
In Section \ref{sec:model}, we introduce the multi-dealer stochastic game and derive the unique Nash equilibrium in Theorem \ref{thm:solution}. In Section \ref{sec:numerical_results}, we present numerical illustrations across a range of market scenarios, along with the methodology and results of our main experiments. Finally, Appendix \ref{appendix:proof} contains the proofs of the main results, while auxiliary definitions and notation are provided in Appendices \ref{appendix:matrices} and \ref{appendix:matrices2}.

\section{Model Setup and Main Results} \label{sec:model}

We consider a game between $N \geq 2$ dealers who all simultaneously and continuously in time choose their ask prices, bid prices, and externalisation rates, in order to `make the spread' while minimising trading costs and inventory risk. 

Let $T>0$ be a finite time horizon, and let $(\Omega, \mathcal{F}, \{\mathcal{F}_{t}\}_{t\in[0, T]}, \mathbb{P})$ be a filtered probability space satisfying the usual conditions of right continuity and completeness. We consider a single risky asset which follows a $\{\mathcal{F}_{t}\}_{t\in[0, T]}$-martingale unaffected price process $P^{\ast}=(P^{\ast}_{t})_{ t \in [0,T]}$ and satisfies $\mathbb{E}\left[\langle P^{\ast} \rangle_{T}\right] < \infty$ with $P^{\ast}_0=p$. The dealers aim to optimally decide between unwinding or internalising the order flow by controlling the externalisation rate $q^i=(q_t^i)_{ t \in [0,T]}$, $i=1,...,N$, chosen from a set of admissible strategies  \begin{equation} \label{adms} 
    \altcal{A} = \left\{q: q \text{ progressively measurable s.t. } \mathbb{E}\left[\int_{0}^{T}q^{2}_{t}dt\right]< \infty\right\}.
\end{equation}
We assume that the dealers' aggregated trading activity causes price impact such that the traded price is given by
\begin{equation} \label{def:P}
    P_{t} = P^{\ast}_{t} +  I_{t}, \quad  0\leq t \leq T, 
\end{equation}
where
\begin{equation} \label{def:Y}
I_{t}:= \lambda \int_{0}^{t}e^{-\beta(t-s)}\left(\sum_{i=1}^{N} q^{i}_{s}\right)ds,
\end{equation}
is the transient price impact created by the aggregated trades according to the Obizhaeva and Wang model   \citep{obizhaeva2013dynamics}, and $\lambda, \beta$ are non-negative constants. 

The market's cumulative ask and bid order flows are modelled by two Poisson processes $Z^{a}= (Z_t^a)_{t\geq0}$ and $Z^{b}= (Z_t^b)_{t\geq0}$ with intensities $\rho^a$ and $\rho^b$, respectively, which are initialised at $0$. In addition, each dealer has an idiosyncratic cumulative ask flow $J^{a, i}=(J^{a, i}_{t})_{t\geq0}$ and cumulative bid flow $J^{b, i}=(J^{b, i}_{t})_{t\geq0}$ which they do not have to compete for. These are also modelled by independent Poisson processes. For simplicity we assume the intensity of all dealers' idiosyncratic ask and bid flow is equal to $\mu$. 
In the following we often use the notation $[N] = \{1,...,N\}$. 
Each dealer $i \in [N]$ continuously controls their ask and bid depths $\delta^{a, i}=(\delta_t^{a, i})_{ t \in [0,T]} $ and $\delta^{b, i}=(\delta_t^{b, i})_{ t \in [0,T]} $ which are chosen from the set of strategies
\begin{equation} \label{adms2} 
    \altcal{D} = \left\{ \delta: \delta \text{ progressively measurable s.t. } \mathbb{E}\left[\int_{0}^{T}\delta^{2}_{t}dt\right]<\infty \right\}.
\end{equation}
The ask and bid quotes of each dealer $i$ at any time $0\leq t \leq T$ are therefore given by,  
\begin{equation} \label{def:S^i}
    S^{a, i}_{t} = P_{t} + \delta^{a, i}_{t},
    \qquad\text{and}\qquad
    S^{b, i}_{t} = P_{t} - \delta^{b, i}_{t}.  
\end{equation}
 The inflow received by each dealer $i$ is assumed to be a proportion of the clients' order flow as follows, 
\begin{equation} \label{def:Z^i}
\begin{split} 
    dZ^{a, i}_{t} &= \frac{1}{N}\left(1-\kappa\left(\delta^{a, i}_{t}-\frac{1}{N}\sum_{j=1}^{N}\delta_{t}^{a, j}\right)\right)dZ^{a}_{t} + {dJ^{a, i}_{t}}, \\ 
    dZ^{b, i}_{t} &= \frac{1}{N}\left(1-\kappa\left(\delta^{b, i}_{t}-\frac{1}{N}\sum_{j=1}^{N}\delta_{t}^{b, j}\right)\right)dZ^{b}_{t} +{ dJ^{b, i}_{t}}.
\end{split}   
\end{equation}
We notice from \eqref{def:Z^i} that the dealers face a trade-off when choosing $\delta^{a, i}$ and $\delta^{b, i}$ between trading at a wider spread or getting  higher proportion of the total order flow. The dealer's inflows collectively sum to the total market inflow on both the ask and bid side and the total idiosyncratic flow, that is,
\begin{equation}
\sum_{i=1}^N Z^{k, i}_{t} = Z^k_t + \sum_{i=1}^{N}J^{k, i}_{t} , \quad \textrm{for all } 0\leq t \leq T , \  k\in \{a,b\}.
\end{equation}
In particular if $\mu=0$ then $\sum_{i=1}^N Z^{k, i}_{t} = Z^k_t$.
\begin{remark}
The dynamics of \( Z^{a, i} \) and \( Z^{b, i} \) in \eqref{def:Z^i} are motivated by a linearisation of exponential fill probabilities. In the single-dealer model of Avellaneda and Stoikov \cite{avellaneda2008high}, the fill probability of a client's market order on the ask side is given by \( e^{-\kappa \delta^a_t} \), where \( \kappa > 0 \) is a constant and \( \delta^a_t \) is the quoted depth. Higher depths therefore attract less order flow. A natural extension of the fill probability for dealer \( i \) in the multi-dealer setting is
\begin{equation}
    \exp\left(-\kappa\left(\delta^{a, i} - \frac{1}{N} \sum_{j=1}^N \delta^{a, j} \right)\right),
\end{equation}
so that quotes better than the average receive more order flow, and vice versa. A first-order Taylor expansion of the deviation from the average depth around zero yields the integrands of \eqref{def:Z^i}. This approximation is justified by the simulation shown in Figures \ref{fig:illustrative-unconstrained} and \ref{fig:illustrative-constrained}, where the differences between competitors' depths are small. Note that linear-quadratic approximations of fill probabilities have previously been used in \citep{adrian2020intraday,bergault2021closed} and in Section 2 of \cite{guo2026macroscopic}.
\end{remark}
Dealer $i$'s inventory is given by the difference between its individual inflow from clients and the amount executed in the market, including an initial inventory  $x_0^{i} \in \mathbb{R}$,
\begin{equation} \label{def:X}
    X^{i}_{t} = x_0^{i} + Z^{b, i}_{t}-Z^{a, i}_{t} + \int_{0}^{t}q^{i}_{s}ds,  \quad t\geq0, \ i \in [N]. 
\end{equation}
Let $\phi^i$ and $\alpha^i$ be positive constants representing inventory risk aversion and terminal inventory penalty (respectively) for each dealer $i \in [N]$.  
The heterogeneity in the risk aversion and the terminal penalty parameters gives us the flexibility to define two types of dealers: internalisers who are comfortable with holding inventory and have low $\phi^{i}$ and $\alpha^{i}$, and native externalisers who are averse to holding inventory and have high $\phi^{i}$ and $\alpha^{i}$ (see Table \ref{tab:params} for the specific values). 

Define $u^i= (\dl^{a,i},\dl^{b,i},q^{i})$ and $u^{-i}=(\dl^{a,j},\dl^{b,j},q^{j})_{ j \in [N] \setminus \{i\} }$. We formulate the objective of dealer $i\in [N]$ as the following revenue-risk functional, 
\begin{equation} \label{def:objective}
\begin{split}
    &\altcal{J}^{i}(u^{i}, u^{-i}) \\
    &= \mathbb{E}\Bigg[\int_{0}^{T}\left(P^{\ast}_{t}+I_{t}+\delta^{a, i}_{t}\right)dZ^{a, i}_{t} -
    \int_{0}^{T}\big(P^{\ast}_{t}+I_{t}-\delta^{b, i}_{t}\big)dZ^{b, i}_{t}  - \int_{0}^{T}\big(P^{\ast}_{t}+I_{t}+\frac{1}{2}\epsilon^{i} q_t^i\big)q^{i}_{t}dt 
    \  \\
    &\qquad - \phi^{i}\int_{0}^{T}\left(X^{i}_{t}\right)^{2}dt - \alpha^{i} \left(X^{i}_{T}\right)^{2} + (P^{\ast}_{T}X^{i}_{T} - p x_0^{i}) \Bigg]. 
\end{split}
\end{equation}  
The first two terms on the right-hand side of \eqref{def:objective} represent the revenue from making the spread, that is from selling and buying the asset from the clients. The third term represents the externalisation costs, where in particular the term $\frac{\epsilon^{i}}{2}(q^{i}_{t})^2$ represents dealer $i$'s slippage costs. 
The last three terms represent the inventory risk aversion, terminal inventory penalisation and the difference in book value of the terminal inventory and initial inventory, respectively.

The goal of each dealer $i \in [N]$ is to maximise the objective functional \eqref{def:objective} over $(\dl^{a,i},\dl^{b,i},q^i)$ in a class of admissible controls determined by \eqref{adms} and \eqref{adms2}. We  prove that if all $N$ control problems can be solved simultaneously, it follows that the system has a Nash equilibrium in the following sense. 
\begin{definition} \label{def:nash}
    A set of admissible strategies $(u^{i})_{i \in [N] }$ where $u^{i}=(\delta^{a, i}, \delta^{b, i}, q^{i})$, is called an open-loop Nash equilibrium if for every admissible strategy $v$ it holds that
    \begin{equation} \label{eq:nash_condition}
        \altcal{J}^{i}(u^{i}, u^{-i}) \geq \altcal{J}^{i}(v, u^{-i}), \quad \textrm{for all } i\in [N]. 
    \end{equation}
\end{definition}
Before we can state our theoretical main result, we introduce some notation and definitions. Let 
\begin{equation}\label{v-i-def}
    V^{i}_{t} = \big(
         \delta^{a, i}_{t}, 
        \delta^{b, i}_{t},
        q^{i}_{t}, 
        \Gamma^{i}_{t} 
     \big)^{\top}, \quad 0\leq t \leq T,\  i \in [N], 
\end{equation}
where each $\Gamma^{i} =(\Gamma^{i}_t)_{ t \in [0,T]}$ is the square-integrable auxiliary process given by 
\begin{equation} \label{def:Gamma}
     \Gamma^{i}_{t} = \lambda e^{\beta t}\Bigg(\int_{0}^{t}e^{-\beta s}\big(dZ^{b, i}_{s}- dZ^{a, i}_{s}\big)
    + \int_{0}^{t}e^{-\beta s}q^{i}_{s}ds  + \wt{N}^{i}_{t}\Bigg)
 \end{equation}
with the square-integrable martingale $(\wt{N}_t)_{ t \in [0,T]}$, 
\begin{equation} \label{def:tildeN} 
\begin{split}
    \wt{N}^{i}_{t} :=
    &- \mathbb{E}\Bigg[\int_{0}^{T}e^{-\beta s}\big(dZ^{b, i}_{s}- dZ^{a, i}_{s}\big)
    + \int_{0}^{T}e^{-\beta s}q^{i}_{s}ds\Big\vert\mathcal{F}_{t}\Bigg].
\end{split}
\end{equation}
We further define ${\mathbf{V}}$ as the following stochastic process, 
\begin{equation} \label{def:V_vector}
    {\mathbf{V}}_{t} = \big(
       (V^{1}_{t})^\top,
         \dots,
        (V^{N}_{t})^\top, 
        1
 \big)^{\top}  , \quad 0\leq t \leq T, 
\end{equation}
Note that for any $t\in[0,T]$, ${\mathbf{V}}_{t}$ is a vector in $\mathbb{R}^{4N+1}$ that contains all the controls $(
         \delta^{a, i}_{t}, 
        \delta^{b, i}_{t},
        q^{i}_{t}, 
        \Gamma^{i}_{t})_{i \in [N]}$ in the first $4N$ entries and the scalar $1$ in $4N+1$ entry. 
        
Recall the agent's inventory and the transient price impact processes \eqref{def:X} and \eqref{def:Y}, respectively. We further define $\mathbf{X}$ as the following $ \mathbb{R}^{N+1}$-valued stochastic process,  
\begin{equation} \label{def:X_vector}
    \mathbf{X}_{t} = \left( 
         X^{1}_{t}, \dots, 
        X^{N}_{t}, I_{t}
    \right)^{\top} , \quad 0\leq t \leq T. 
\end{equation}
We introduce two matrices $\mathbf{K}\in\mathbb{R}^{(5N+2)\times(5N+2)}$ and $\mathbf{T}\in\mathbb{R}^{(4N+1)\times(5N+2)}$ which contain combinations of the model's parameters $N, \lambda, \beta, \kappa, (\epsilon^{i}, \phi^{i}, \alpha^{i})_{i \in [N]}$ and the intensities of the clients' flow. Due to their involved presentation we refer to their explicit forms \eqref{def:K} and \eqref{def:T} in Appendix \ref{appendix:matrices}. We further define,  
\begin{equation}
    \mathbf{G}(t) := \exp\left(\mathbf{K}t\right) , \quad 0\leq t \leq T. 
\end{equation}
and let $\mathbf{H}(t) := \mathbf{T}\mathbf{G}(t)\in\mathbb{R}^{(4N+1)\times(5N+2)}$. Also denote by $\mathbf{H}^{V}$ the first $4N+1$ columns of $\mathbf{H}$ and by $\mathbf{H}^{X}$ the last $N+1$ columns. Finally, we define the following vector in $\mathbb{R}^{4N+1}$,
\begin{equation}
    \mathbf{c} := (0, ..., 0, 1)^{\top}. 
\end{equation}
\paragraph{Convention.} We denote by $\norm*{\cdot}$ the Euclidean norm and by $\norm*{\cdot}_{2}$ the matrix spectral norm. 
The following assumption is crucial for the derivation of the Nash equilibrium. 
\begin{assumption} \label{assum:bounded}
    We assume the model parameters are chosen such that,
    \begin{itemize}
    \item[\textbf{(i)}] 
    The inverse matrix $(\mathbf{H}^{V}( t))^{-1}$ exists for all $t\in(0, T)$  and that 
    \begin{equation} 
        \sup_{t\in [0,T]} \left(\norm*{\left(\mathbf{H}^{V}(t)\right)^{-1}\mathbf{c}} +\norm*{\left(\mathbf{H}^{V}(t)\right)^{-1}\mathbf{H}^{X}(t)}_{2} \right)<\infty, 
    \end{equation}
        \item[\textbf{(ii)}] 
 $$ \frac{\beta}{\lambda} \geq \frac{\kappa(N-1)}{4N^{2}}(\rho^a+\rho^b).$$
 \end{itemize} 
 \end{assumption}
We are now ready to state our main theoretical result, which derives a closed-form solution to the Nash equilibrium of the model. Recall that the process $\mathbf{V}$ in \eqref{def:V_vector} contains all dealers' controls and the process $\mathbf{X}$ includes the states of the model. 
\begin{theorem} \label{thm:solution}
Under Assumption \ref{assum:bounded} there exists a unique Nash equilibrium $u^{*}=(u^{i,*})_{i\in [N]}$ to the game, in the sense of Definition \ref{def:nash} which is given by the following feedback form, 
    \begin{equation} \label{eq:feedback}
    \mathbf{V}^{}_{t} = \Big(\mathbf{H}^{V}(T-t)\Big)^{-1}\Big(\mathbf{c}-\mathbf{H}^{X}(T-t)\mathbf{X}_{t}\Big), \quad 0\leq t \leq T.  
\end{equation}
\end{theorem}
The proof of Theorem \ref{thm:solution} is given in Appendix \ref{appendix:proof}.
\begin{remark}
In Proposition \ref{prop:FBSDE}, we use a variational approach to derive a system of FBSDEs that characterises the Nash equilibrium defined in Definition \ref{def:nash}, under a slightly more relaxed assumption that $Z^{a}$ and $Z^{b}$ are non-decreasing, non-negative, progressively measurable processes. To obtain closed-form solutions, we focus on the specific case of Poisson inflow in Theorem \ref{thm:solution}.
 \end{remark}
\begin{remark}
Assumption \ref{assum:bounded}(i) is essential for deriving the closed-form solution in Theorem \ref{thm:solution} and for establishing its admissibility. Assumption \ref{assum:bounded}(ii) is required to prove the concavity of the individual agents' cost functionals \eqref{def:objective} (see Proposition \ref{prop:uniqueness}). This concavity, in turn, is necessary to establish the uniqueness of the Nash equilibrium.
\end{remark}

\subsection{Evaluation metrics} \label{sec:metrics}
We introduce several metrics we use to investigate our experiments. 

\begin{itemize}
\item \textbf{Nominal spread capture}
\begin{equation}
{S}^i_T = \int_{0}^{T}\delta^{a, i}_{t} dZ^{a, i}_{t} + \int_{0}^{T}\delta^{b, i}_{t} dZ^{b, i}_{t} 
\end{equation}
Nominal spread capture measures the profit from `making the spread'.
\item\textbf{Effective spread capture}
\begin{equation}
\altcal{S}^i_T = \int_{0}^{T}\left(I_{t}+\delta^{a, i}_{t}\right)dZ^{a, i}_{t} - \int_{0}^{T}\big(I_{t}-\delta^{b, i}_{t}\big)dZ^{b, i}_{t} 
\end{equation}
Effective spread capture measures the profit from `making the spread' after price impact.
\item\textbf{Total hedging costs}
\newline
\noindent
Instantaneous hedging costs measure the cost of externalisation due to the immediate effects, for example approximating the primary spread and are given by
\begin{equation}
\tilde{H}^i_T = \frac{\epsilon^{i}}{2}\int_{0}^{T}(q^{i}_{t})^{2}dt.
\end{equation}
Impact hedging costs measure the cost of externalisation due to transient price impact. Note that unlike instantaneous hedging costs this depends indirectly on the trading of other dealers in the market. Impact hedging costs are given by 
\begin{equation}
    \hat{H}^i_T = \int_0^T I_{t} q_{t}^{i} dt.
\end{equation}
Total hedging costs are the sum of instantaneous hedging costs and impact hedging costs and are thus given by
\begin{equation} \label{def:total_hedging_costs}
    \altcal{H}_{T}^{i} = \tilde{H}^{i}_{T} + \hat{H}^{i}_{T}.
\end{equation}
\item\textbf{P\&L}
\begin{equation} \label{def:pnl}
\begin{split}
{P\&L}^i_T 
&= \altcal{S}^i_T -\altcal{H}^i_T + \int_{0}^{T} P^{\ast}_{t} (dZ^{a, i}_{t} - dZ^{b, i}_{t} - q^{i}_{t}dt) + (P^{\ast}_{T} X_T^i - p x_0^i)
\end{split}
\end{equation}
P\&L is the revenue minus costs of all activity of dealer $i$. This includes effective spread capture, total hedging costs, which is exactly the interior of the objective functional \eqref{def:objective} without the penalty terms.
\item\textbf{Hedging volume}
\begin{equation}
    v^i_T = \int_{0}^{T}|q^{i}_{t}|dt
\end{equation}
Hedging volume is the volume externalised by dealer $i$ on the interbank market.
\item\textbf{Median client's spread}
\begin{equation} \label{def:median_client_spread}
    \left(-I_{t} + \underset{i\in[N]}{\min}\delta^{b, i}_{t} \right)\mathbbm{1}_{\{Z^{b}_{T}>Z^{a}_{T}\}} + \left(I_{t} + \underset{i\in\{1, ..., N\}}{\min}\delta^{a, i}_{t}\right)\mathbbm{1}_{\{Z^{b}_{T}<Z^{a}_{T}\}}. 
\end{equation}
Although clients are not defined as agents in this model, we can consider the trading costs faced by a client. These costs are defined in \eqref{def:median_client_spread} and the reasoning is as follows. In a market where the terminal client order imbalance is positive ($Z^{b}_{T}-Z^{a}_{T}>0$), we say that a selling client is the median client. A selling client's spread cost at time $t$ is the difference between the best price they can sell at, and the unaffected price  $P^{\ast}$, that is 
\begin{equation} \label{eq:selling_median_spread}
    \underset{i\in\{1, ..., N\}}{\max}\left( P^{\ast}_{t} + I_{t} - \delta^{b, i}_{t} \right) - P^{\ast}_{t}   =  I_{t} - \underset{i\in\{1, ..., N\}}{\min}\delta^{b, i}_{t}  .
\end{equation}
In a market where the terminal order imbalance is negative ($Z^{b}_{T}-Z^{a}_{T}<0$), we say that a buying client is the median client. A buying client's spread cost at time $t$ is the difference between the best price they can buy at, and the unaffected price
\begin{equation} \label{eq:buying_median_spread}
    \underset{i\in\{1, ..., N\}}{\min} \left( P^{\ast}_{t} + I_{t} + \delta^{a, i}_{t} \right)  -  P^{\ast}_{t}  =  I_{t} + \underset{i\in\{1, ..., N\}}{\min}\delta^{a, i}_{t}.
\end{equation}
Taking into account the negative sign of the price impact for the sale, the median client's spread at time $t$ is given by \eqref{def:median_client_spread}. High median client's spreads indicate that the majority of the dealers' clients are in an unfavourable scenario.  
\item\textbf{Median client's mean spread} 
\begin{equation} \label{def:median_client_best_spread}
    \left(\frac{1}{T}\int_{0}^{T}-I_{t} + \underset{i\in\{1, ..., N\}}{\min}\delta^{b, i}_{t}dt \right)\mathbbm{1}_{\{Z^{b}_{T}>Z^{a}_{T}\}} + \left(\frac{1}{T}\int_{0}^{T}I_{t} + \underset{i\in\{1, ..., N\}}{\min}\delta^{a, i}_{t}dt\right)\mathbbm{1}_{\{Z^{b}_{T}<Z^{a}_{T}\}}
\end{equation}
Similar to the median client's spread, the median client's mean spread is the average median client's spread per unit time.
\end{itemize}

\section{Main Model Outputs} \label{sec:numerical_results}

In this section we present the main insights and conclusions arising from Theorem \ref{thm:solution}. We investigate how the presence of native externalisers in a pool of  internalisers affects P\&L of the dealers as well as clients' executed prices. We model internalisers as dealers who have lower running and terminal inventory penalties, and are therefore more likely to be able to offset positions with opposing client orders. We set their running inventory costs and terminal inventory penalty coefficients to $\phi^{i}=0.001, \alpha^{i}=0.001$ (see \eqref{def:objective}). Recall that internalisers also have the option to externalise in the market. On the other hand, native externalisers are uncomfortable with holding inventory for long periods of time and therefore are prepared to aggressively unwind their positions, while creating substantial price impact. Their running inventory costs and terminal inventory penalty coefficients are set much higher to $\phi^{i}=0.1, \alpha^{i}=0.1$. These parameters  are summarised in Table \ref{tab:params}. 

Moreover, we wish to display how these effects last over a prolonged time, such as a trading day or week. However, Nash equilibria in our model are unstable over such time horizons. We therefore use the solution at time $t=0$ across a 24-hour simulation period. We set all dealers' initial inventories to zero, that is $x^{i}_{0}=0$, $i\in[N]$. Other model parameters are also given in Table \ref{tab:params}. Note that in Table \ref{tab:params} the intensities of order flow on both the ask and bid side are equal and thus the order imbalance, defined by 
\begin{equation}\label{def:order-imbalance}
    \text{Imb}_{t} = Z^{b}_{t} - Z^{a}_{t}, \quad 0\leq t\leq T, 
\end{equation}
yields a mean-zero net order-flow.   

\begin{table}[H]
\begin{center}
\caption{Model parameter values}\label{tab:params}
\vskip-0.2cm
\small
\setlength{\tabcolsep}{4.5pt}
\begin{tabular}{lclllcl}
\toprule
\multicolumn{2}{l}{\textbf{Parameter}} & \textbf{Description} && \multicolumn{2}{l}{\textbf{Parameter}} & \textbf{Description}\\
\midrule
$N$              & 6     & number of dealers
&& $p$         & 0     & initial price level \\
$\sigma$         & 5.0   & price volatility
&& $\rho^{a,b}$  & 25    & client flow rate \\
$T$              & 24    & time period
&& $\mu$         & 1     & idiosyncratic flow rate \\
$\kappa$         & 0.5   & price sensitivity of flow
&& $\phi^{i}$    & 0.001 & internaliser inventory risk-aversion \\
$\lambda$        & 1.0   & transient price impact
&& $\phi^{i}$    & 0.1   & externaliser inventory risk-aversion \\
$\beta$          & 1.0   & impact decay-rate
&& $\alpha^{i}$  & 0.001 & internaliser terminal inventory penalty \\
$\epsilon_i$     & 0.1   & externalisation cost $\forall i\in [N]$
&& $\alpha^{i}$  & 0.1   & externaliser terminal inventory penalty \\
\bottomrule
\end{tabular}
\end{center}
\end{table}

\subsection{Illustrative simulation} \label{sec:unconstrained_scenario}

In this section we show a simulation of the model with the parameters summarised in Table \ref{tab:params}. We wish to simulate over a long period, such as 24 hours, so to ensure numerical stability we run the strategy evaluated at $t=0$ over the entire period, and do the same for the remainder of this section. Moreover, this is justified by an interpretation of the terminal inventory penalty as future running penalty discounted to a time horizon that should not actually be reached, owing to the 24-hour nature of FX markets.

\begin{figure}[H]
  \centering

  \begin{subfigure}[b]{0.32\textwidth}
    \centering
    \includegraphics[width=\linewidth]{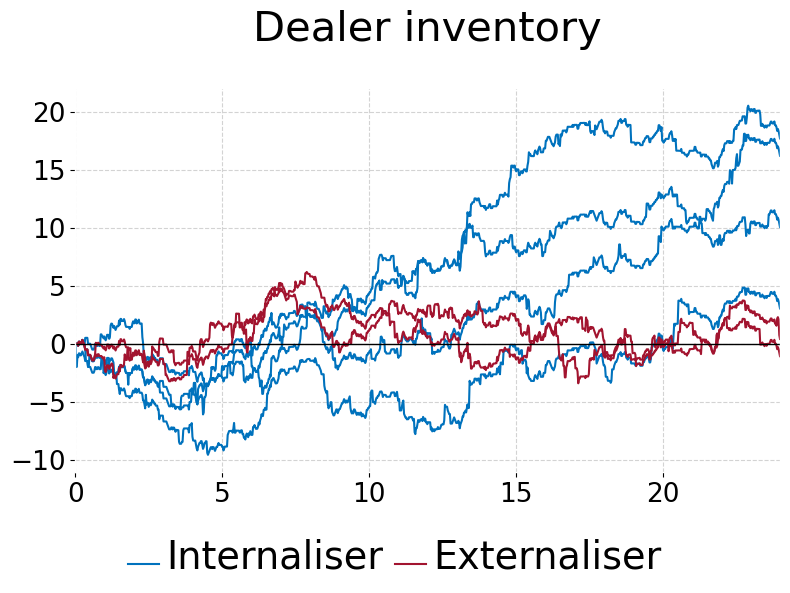}
    \label{fig:X-unconstrained}
  \end{subfigure}\hfill
  \begin{subfigure}[b]{0.32\textwidth}
    \centering
    \includegraphics[width=\linewidth]{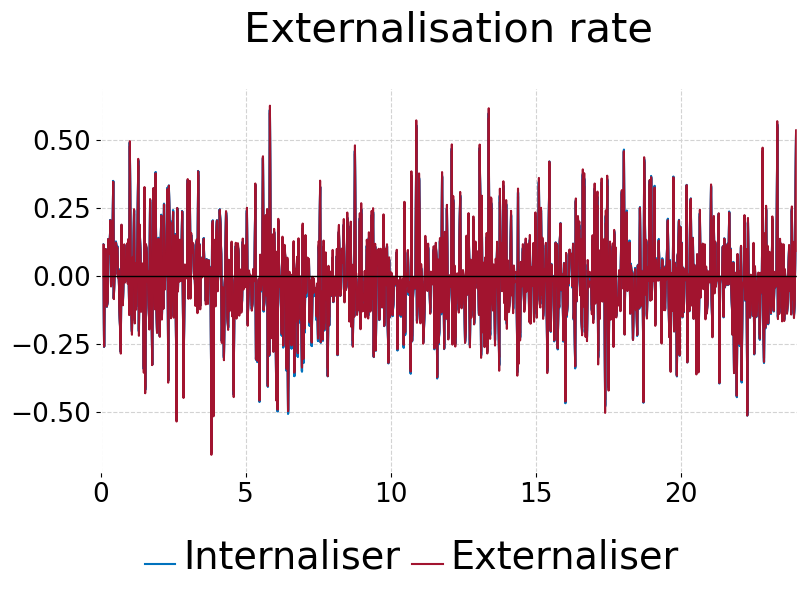}
    \label{fig:q-unconstrained}
  \end{subfigure}\hfill
  \begin{subfigure}[b]{0.32\textwidth}
    \centering
    \includegraphics[width=\linewidth]{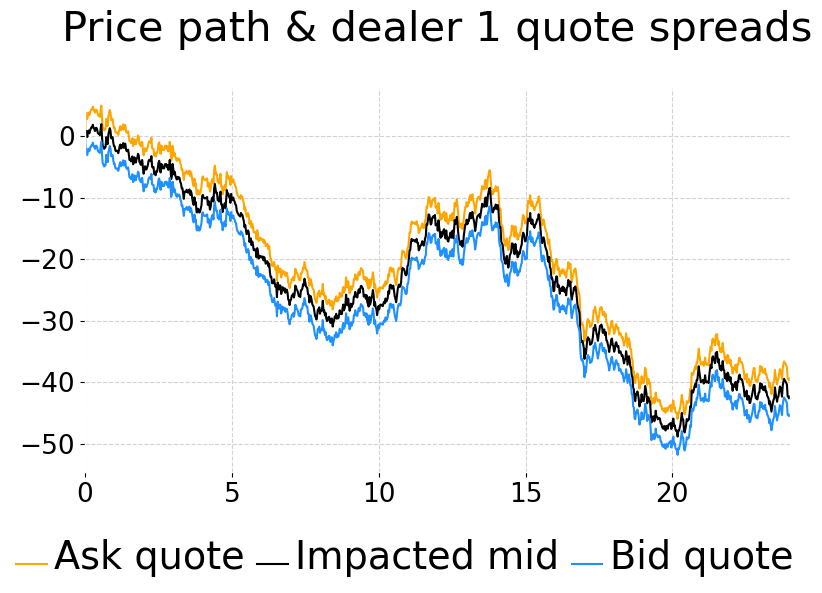}
    \label{fig:prices-unconstrained}
  \end{subfigure}

  \vspace{0.6em}

  \begin{subfigure}[b]{0.32\textwidth}
    \centering
    \includegraphics[width=\linewidth]{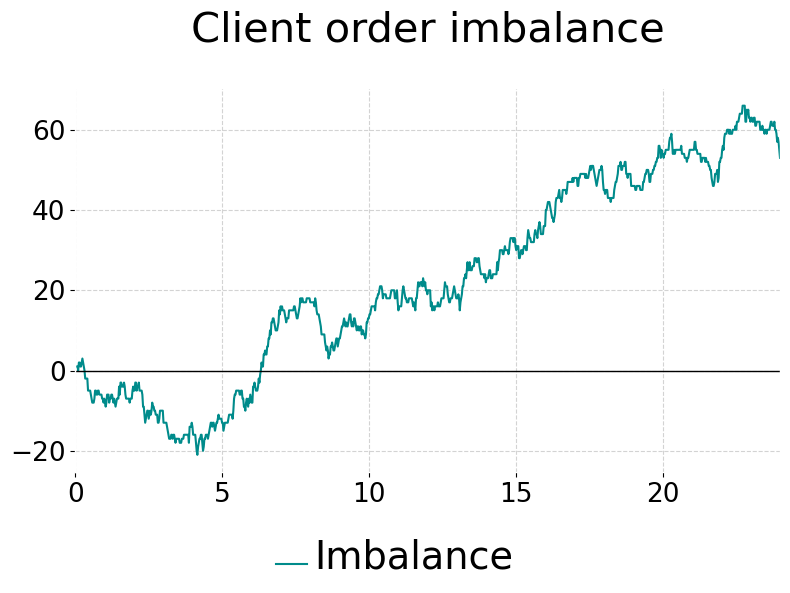}
    \label{fig:imbalance-unconstrained}
  \end{subfigure}\hfill
  \begin{subfigure}[b]{0.32\textwidth}
    \centering
    \includegraphics[width=\linewidth]{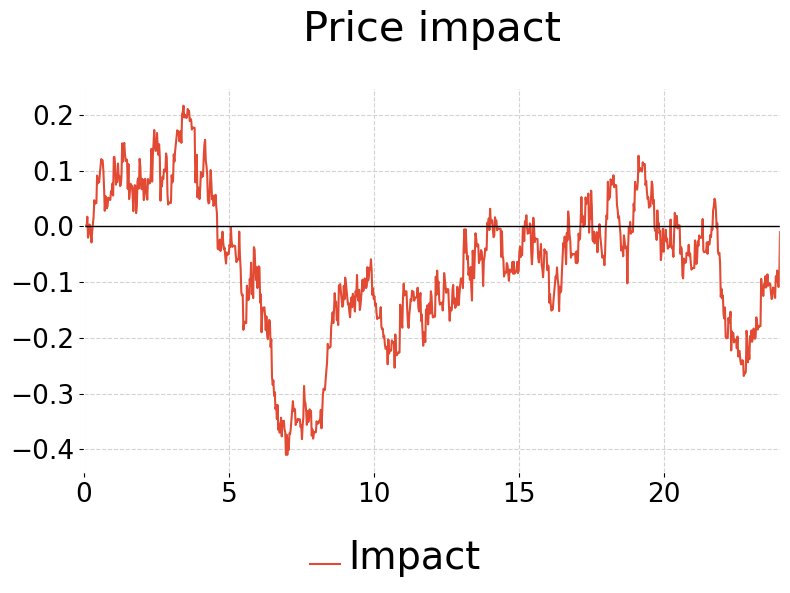}
    \label{fig:I-unconstrained}
  \end{subfigure}\hfill
  \begin{subfigure}[b]{0.32\textwidth}
    \centering
    \includegraphics[width=\linewidth]{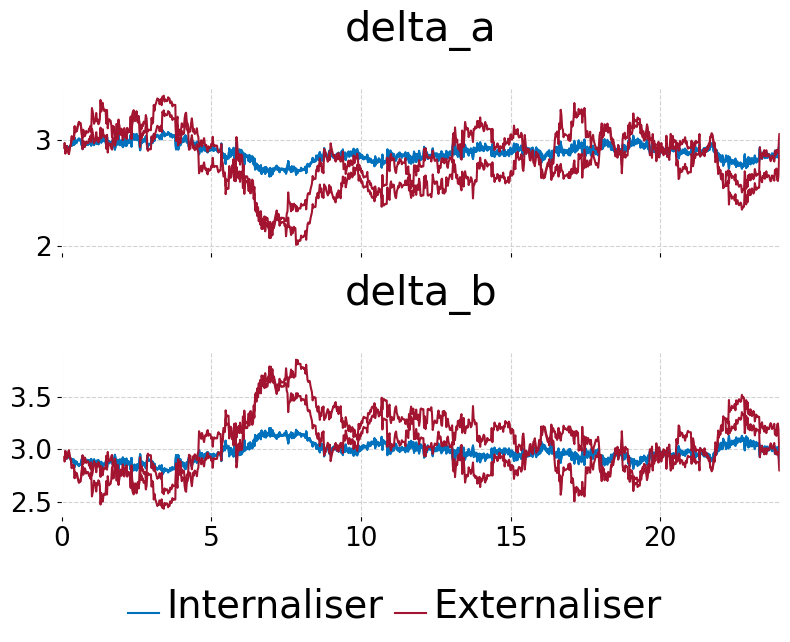}
    \label{fig:deltas-unconstrained}
  \end{subfigure}

  \caption{An example simulation with the parameters from Table \ref{tab:params}. Four dealers are internalisers and two are externalisers.}
  \label{fig:illustrative-unconstrained}
\end{figure}

In Figure \ref{fig:illustrative-unconstrained} we present a realisation of the evolution over time for the dealers' inventories $X^{i}$ (blue for internalisers and red for externalisers), externalisation rates $q^{i}$ (also blue for internalisers and red for externalisers), price paths with dealer 1's spreads, where dealer 1 is an internaliser, OTC order imbalance \eqref{def:order-imbalance}, transient price impact $I$ and dealers' half-spreads $\delta^{a, i}, \delta^{b, i}$ (blue for internalisers and red for externalisers).  Note that the blue line in the upper-middle panel cannot be seen because it is almost exactly hidden by the red line. We investigate this phenomenon further later in this section.

We observe that the externalisers, as required by their strict inventory penalties, maintain positions much smaller in magnitude than the internalisers who are comfortable with deviations from zero. This is achieved by high magnitude externalisation rates, but in fact the internaliser's externalisation rates are very close to the externalisers' and the difference in inventories between internalisers and externalisers is explained mostly by their quoting. We return to this observation later in the section. We also note that the externalisation rates are very noisy and erratic, and appear more reminiscent of high-frequency opportunistic behaviour than of prudent management of inventory risk.

\subsection{Comparison of dealer combinations} \label{sec:comparison-unconstrained}

We consider the same 1000 simulated paths of unaffected price and order flow movements across various scenarios, which vary the composition of dealers in the market, holding the total number of dealers constant at six, $N=6$. In particular, there are six native internalisers in Scenario A, five native internalisers and one native externaliser in Scenario B, four native internalisers and two native externalisers in Scenario C, and so on, until there are six native externalisers in Scenario G. The compositions are summarised in Table \ref{tab:dealer_types}.
\begin{table}[H]
    \centering
    \caption{\textbf{Scenario dealer compositions}}
    \label{tab:dealer_types}
    \renewcommand{\arraystretch}{1.2}
    \setlength{\tabcolsep}{10pt}
    \begin{tabular}{llcccccc}
        \hline
        \textbf{Scenario} & \textbf{A} & \textbf{B} & \textbf{C} & \textbf{D} & \textbf{E} & \textbf{F} & \textbf{G} \\
        \hline
        \textbf{Internalisers} &  6 & 5 & 4 & 3 & 2 & 1 & 0 \\
        \textbf{Externalisers} &  0 & 1 & 2 & 3 & 4 & 5 & 6 \\
        \hline
    \end{tabular}
\end{table}
For each scenario, the means and 95\% confidence intervals of the dealers' P\&Ls, nominal and effective spread captures, hedging costs, hedging volumes, and median client mean spread (all defined in Section \ref{sec:metrics}) are given in Table \ref{tab:keymetrics-unconstrained} averaged across the two dealer types. This allows us to draw some immediate conclusions.

\textbf{Internalisers become more like externalisers as more externalisers are present.}
We find that internalisers behave more like externalisers when the proportion of externalisers in the pool increases. Internalisers externalise much more aggressively in this case, which can be seen in the increased hedging costs \eqref{def:total_hedging_costs}. This is illustrated in Figure \ref{fig:q-comparison-unconstrained}, which shows a single example path of dealer 1's externalisation rate $q^{1}$ from the same random draws of unaffected price and order flow across four scenarios (note dealer 1 is a native internaliser in scenarios A, C, E and a native externaliser in G). The magnitudes of the externalisation rates rise dramatically when there are more native externalisers present even as the dealer type of dealer 1 remains the same. A consequence of this is higher instantaneous trading costs, and that much more transient price impact accumulates (see Figure \ref{fig:I-percentiles-unconstrained} which shows bands within which 5\%-95\% of paths of transient impact fall), which results in higher hedging costs (see Table \ref{tab:keymetrics-unconstrained}). 

\begin{figure}[H]
    \centering
    \begin{subfigure}[t]{0.48\textwidth}
        \centering
        \includegraphics[width=\linewidth]{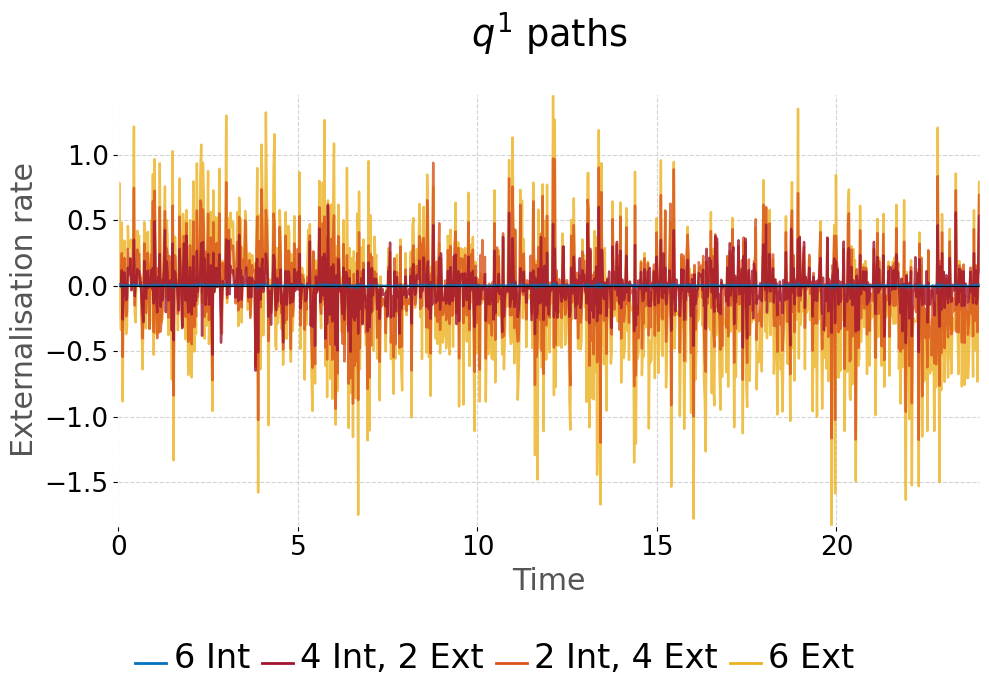}
        \caption{}
        \label{fig:q-comparison-unconstrained}
    \end{subfigure}
    \hfill
    \begin{subfigure}[t]{0.48\textwidth}
        \centering
        \includegraphics[width=\linewidth]{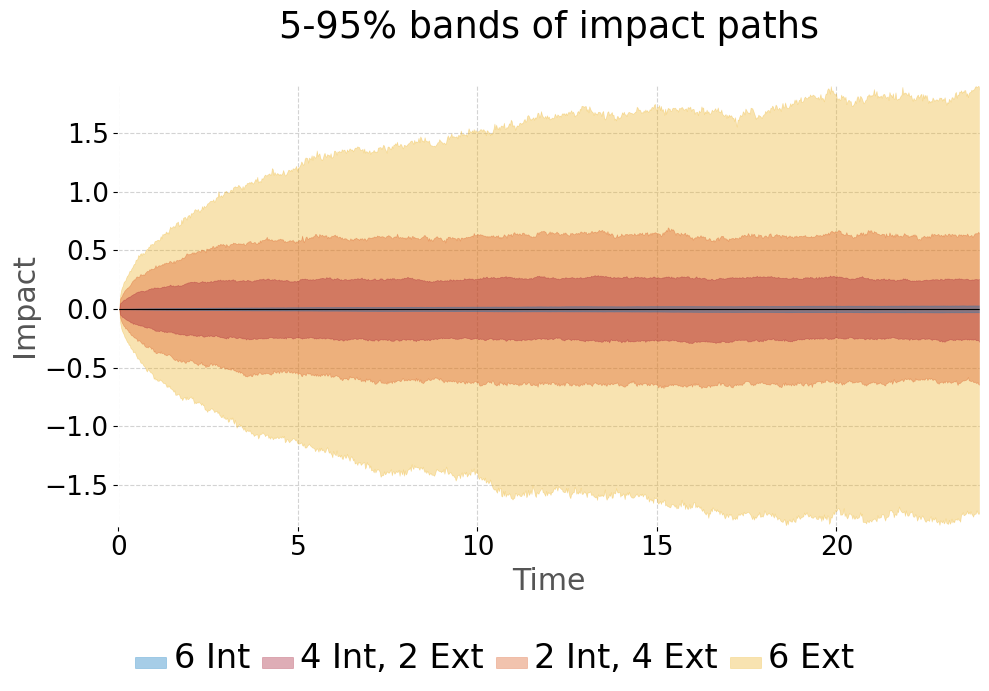}
        \caption{}
        \label{fig:I-percentiles-unconstrained}
    \end{subfigure}
    \caption{(a) A single path of internaliser 1's externalisation rate in Scenarios A (blue), C (red), E (orange), and G (yellow). (b) 5--95\% percentiles for the paths of transient price impact $I$ across time in the same scenarios. Both use the impact-aware strategy profile.}
    \label{fig:unconstrained-side-by-side}
\end{figure}

\textbf{Internalisers have increased effective spread capture as more externalisers are present.} 
In Table \ref{tab:keymetrics-unconstrained}, the internalisers' effective spread capture rises as externalisers are added to the pool. This occurs because both the price impact and the externalisers' individual mids are anticorrelated with the externalisers' inventories. This can be understood most easily through a hypothetical example. Suppose the median externaliser position is short due to large client inflow from the ask side of the book. Then the externalisers are buying and creating price impact which pushes the price up. At the same time, the externalisers skew aggressively, quoting very wide ask spreads and very narrow bid spreads in order to balance their inventory. The internalisers, however, do not replicate this skew, and sell more through clients' flow when the price is temporarily elevated due to impact. The converse logic applies when the median externaliser position is long: in that case, the internalisers buy at temporarily depressed prices.

This mechanism increases the internalisers' effective spread capture and therefore their P\&L, while also increasing their risk. A visual illustration can be seen in Figure \ref{fig:illustrative-zoomed}, which depicts the same simulation as in Figure \ref{fig:illustrative-unconstrained} zoomed to focus on the first ten hours of the day. In particular, between the sixth and eighth hour a sudden increase in client selling (dealer buying) causes dealers to sell via externalisation while adjusting their skews so that impact drives the price downward. The externalisers adjust their skews more aggressively than the internalisers, leaving the internalisers buying at temporarily depressed prices.

\begin{figure}[H]
  \centering
  \begin{subfigure}[b]{0.32\textwidth}
    \centering
    \includegraphics[width=\linewidth]{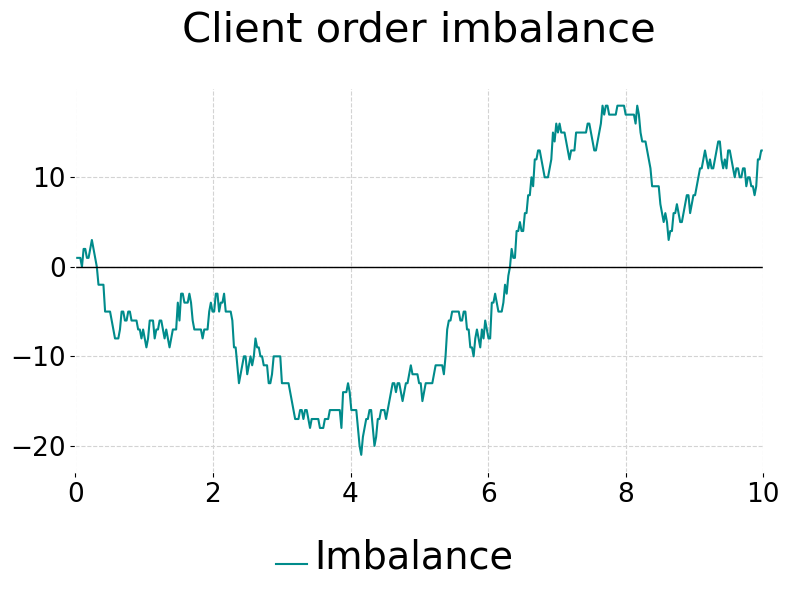}
    \label{fig:imbalance-zoomed}
  \end{subfigure}\hfill
  \begin{subfigure}[b]{0.32\textwidth}
    \centering
    \includegraphics[width=\linewidth]{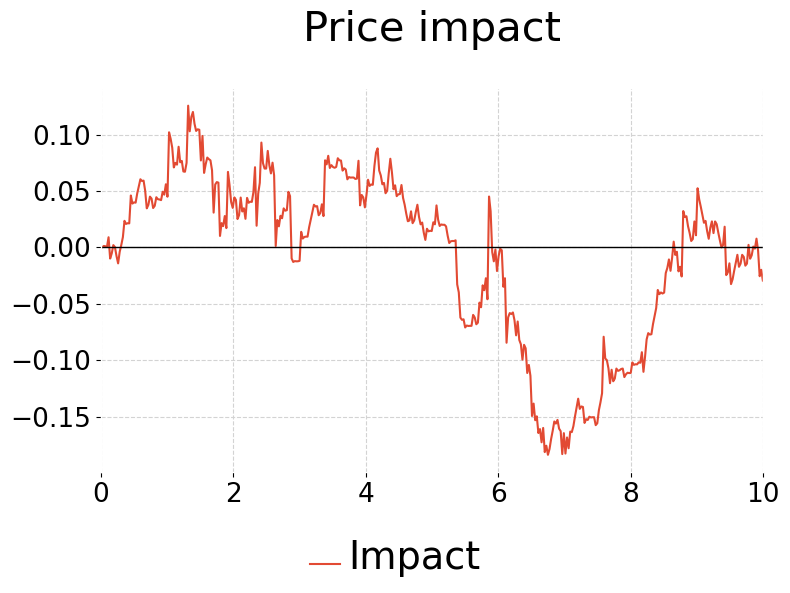}
    \label{fig:I-zoomed}
  \end{subfigure}\hfill
  \begin{subfigure}[b]{0.32\textwidth}
    \centering
    \includegraphics[width=\linewidth]{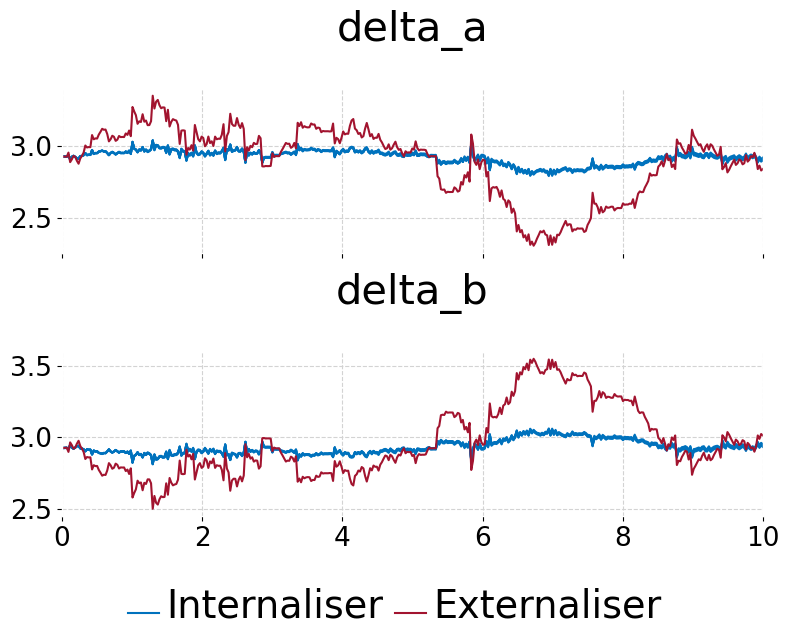}
    \label{fig:deltas-zoomed}
  \end{subfigure}

  \caption{The example simulation from Figure \ref{fig:illustrative-unconstrained}, zoomed to focus on the first ten hours. The effect that leads to a difference in effective spread capture between dealer types is most prominent between the sixth and eighth hours. Four dealers are internalisers and two are externalisers.}
\label{fig:illustrative-zoomed}
\end{figure}

\begin{table}[H]
\begin{center}
\caption{Client and dealer performance metrics: impact-aware strategy profile}
\label{tab:keymetrics-unconstrained}
\vskip-0.2cm
\small
\setlength{\tabcolsep}{4.5pt}
\fontsize{9.2pt}{11.5pt}\selectfont
\begin{tabular*}{\textwidth}{@{\extracolsep{\fill}}lcccccc@{}}
\toprule
            &   & \multicolumn{2}{c}{\textbf{Spread capture}} & \textbf{Hedging} & \textbf{Hedging} & \textbf{Median client} \\
\cline{3-4}
            & \textbf{P\&L} & Nominal $S$ & Effective $\altcal{S}$ & \textbf{costs} $\altcal{H}$ & \textbf{volume $v$} & \textbf{spread} \\
\midrule
\multicolumn{6}{l}{\emph{Scenario A : 6 internalisers}} & 2.930 $\pm$ 0.001 \\
Internaliser    & 726.916 $\pm$ 6.632 & 2.925 $\pm$ 0.000 & 3.629 $\pm$ 0.003 & 0.011 $\pm$ 0.048 &  0.109 ± 0.001 &  \\
\midrule
\multicolumn{6}{l}{\emph{Scenario B : 5 internalisers \& 1 externaliser}} & 2.894 $\pm$ 0.006 \\
Internaliser    & 727.881 $\pm$ 7.591 & 2.925 $\pm$ 0.000 & 3.632 $\pm$ 0.003 & 0.096 $\pm$ 0.168 & 1.561 ± 0.004 &  \\
\cdashline{1-7}[0.4pt/1.6pt]
Externaliser    & 717.681 $\pm$ 3.461 & 2.925 $\pm$ 0.000 & 3.593 $\pm$ 0.006 & 0.153 $\pm$ 0.241 & 1.605 ± 0.004 &  \\
\midrule
\multicolumn{6}{l}{\emph{Scenario C : 4 internalisers \& 2 externalisers}} & 2.901 $\pm$ 0.009 \\
Internaliser    & 729.474 $\pm$ 8.857 & 2.925 $\pm$ 0.000 & 3.638 $\pm$ 0.003 & 0.147 $\pm$ 0.354 &  3.000 ± 0.006 &  \\
\cdashline{1-7}[0.4pt/1.6pt]
Externaliser    & 717.298 $\pm$ 3.016 & 2.925 $\pm$ 0.000 & 3.590 $\pm$ 0.005 & 0.112 $\pm$ 0.291 & 2.980 ± 0.006 &  \\
\midrule
\multicolumn{6}{l}{\emph{Scenario D : 3 internalisers \& 3 externalisers}} & 2.940 $\pm$ 0.013 \\
Internaliser    & 732.590 $\pm$ 10.679 & 2.925 $\pm$ 0.000 & 3.650 $\pm$ 0.004 & 0.401 $\pm$ 0.694 & 4.490 ± 0.010 &  \\
\cdashline{1-7}[0.4pt/1.6pt]
Externaliser    & 716.432 $\pm$ 3.056 & 2.925 $\pm$ 0.000 & 3.587 $\pm$ 0.004 & 0.240 $\pm$ 0.450 & 4.394 ± 0.009 &  \\
\midrule
\multicolumn{6}{l}{\emph{Scenario E : 2 internalisers \& 4 externalisers}} & 3.016 $\pm$ 0.020 \\
Internaliser    & 736.866 $\pm$ 13.500 & 2.925 $\pm$ 0.000 & 3.675 $\pm$ 0.005 & 0.796 $\pm$ 1.270 & 6.110 ± 0.015 &  \\
\cdashline{1-7}[0.4pt/1.6pt]
Externaliser    & 716.672 $\pm$ 3.268 & 2.925 $\pm$ 0.000 & 3.586 $\pm$ 0.003 & 0.461 $\pm$ 0.750 & 5.897 ± 0.011 &  \\
\midrule
\multicolumn{6}{l}{\emph{Scenario F : 1 internaliser \& 5 externalisers}} & 3.166 $\pm$ 0.031 \\
Internaliser    & 753.301 $\pm$ 18.487 & 2.925 $\pm$ 0.000 & 3.736 $\pm$ 0.009 & 2.222 $\pm$ 2.443 & 7.981 ± 0.028 &  \\
\cdashline{1-7}[0.4pt/1.6pt]
Externaliser    & 716.301 $\pm$ 3.740 & 2.925 $\pm$ 0.000 & 3.589 $\pm$ 0.003 & 1.254 $\pm$ 1.411 & 7.549 ± 0.016 &  \\
\midrule
\multicolumn{6}{l}{\emph{Scenario G : 6 externalisers}} & 3.617 $\pm$ 0.064 \\
Externaliser    & 720.607 $\pm$ 4.840 & 2.925 $\pm$ 0.000 & 3.617 $\pm$ 0.004 & 3.961 $\pm$ 3.052 & 9.546 ± 0.034 &  \\
\bottomrule
\end{tabular*}
\end{center}
\end{table}

\subsection{The impact-agnostic scenario}

The crucial aspect of the model underlying each of these qualitative results is the transient impact. In reality, however, it may be difficult to attribute price movements to dislocations caused by impact. Moreover, the institutional and regulatory environment in FX pushes dealers toward risk-management behaviour and away from the explicit exploitation of transient price effects tied to client flow. Real-world dealers may therefore be unable to trade on price impact as a signal in the same way as in the unconstrained stochastic control problem considered in \eqref{def:objective}.

We are therefore interested in investigating the case in which dealers do not use the aggregated price distortion, and hence the externalisation order flow, as a signal when determining their strategies. More broadly, we would like to understand how market effects depend on these two different modelling assumptions, without claiming that reality exactly matches either setting. Solving such a high-dimensional stochastic game under the constraint that each dealer’s control cannot depend functionally on the transient price impact is intractable due to the nonlinearity introduced by this constraint.

In order to approximate such a solution, we modify the Nash equilibrium in  \eqref{eq:feedback} by setting the impact term $I$ in  \eqref{def:V_vector} to zero, that is: 
\begin{equation} \label{def:X_vector_tilde}
    \tilde{\mathbf{X}}_{t} = \left( 
         X^{1}_{t}, \dots, 
        X^{N}_{t}, 0
    \right)^{\top} , \quad 0\leq t \leq T. 
\end{equation}
Recalling the definitions of $\mathbf{H}^{V}, \mathbf{H}^{X}, \mathbf{c}$ in Section \ref{sec:model}, the candidate strategy profile is the variant of \eqref{eq:feedback}, plugging-in $\tilde{\mathbf{X}}$ instead of ${\mathbf{X}}$ to get, 
\begin{equation} \label{eq:candidate}
    \mathbf{V}^{}_{t} = \Big(\mathbf{H}^{V}(T-t)\Big)^{-1}\Big(\mathbf{c}-\mathbf{H}^{X}(T-t)\tilde{\mathbf{X}}_{t}\Big), \quad 0\leq t \leq T, 
\end{equation}
which is a proposed strategy profile agnostic to the price impact $I$ acting as a signal.  

The strategy profile in \eqref{eq:candidate} is an $\epsilon$-Nash equilibrium (with $\epsilon \approx 2\%$ of the P\&L) for the game defined by \eqref{def:objective}, with the same state dynamics as in \eqref{def:P}, \eqref{def:Y}, \eqref{def:S^i}, \eqref{def:Z^i}, and \eqref{def:X}, under the constraint that dealers’ externalisation rates are affine functions of their inventories. Specifically, the constrained class of admissible externalisation rates depends on dealers’ inventories, but not on the price impact:
\begin{equation} \label{adms_constrained} 
    \tilde{\altcal{A}} = \left\{q: q \text{ progressively measurable s.t. }
    \mathbb{E}\left[\int_{0}^{T}q^{2}_{t}dt\right]< \infty \quad \text{and} \quad
    q_t = \sum_{j=1}^{N}f^{j}(t)X_t^{j} + g(t)\right\}
\end{equation}
and similarly for the constrained class of admissible half-spreads. Note that the equilibrium strategies in \eqref{eq:feedback} are affine functions of both inventories and price impact. See Remark \ref{remark:optimality-test} for additional details. The following subsections are dedicated to results arising from the impact-agnostic strategy profile.  

\begin{remark}[$\epsilon$-Optimality test]
\label{remark:optimality-test}
This constrained game is difficult to solve analytically, but we can numerically run a local optimality test. We test many perturbations $\nu$ of the proposed strategy profile in $u$ from \eqref{eq:candidate} and consider $\altcal{J}^i(u^i+\nu, u^{-i})-\altcal{J}^i(u^i,u^{-i})$ 
and perform a search of the optimiser of the constrained problem around $u$. Since we show in Proposition \ref{prop:uniqueness} that our objective functional is strictly concave under Assumption \ref{assum:bounded}, any local maximiser is also a global maximiser. 
We test $100$ perturbations in the constrained admissible sets against the candidate strategy profile \eqref{eq:candidate}, using $500$ simulations for each perturbation to estimate each objective functional. A perturbed strategy profile uses draws $Z^{1}, Z^{2}, Z^{3}$ from a $\altcal{N}(0, \sigma_{\text{pert}})$ distribution and for dealer $i$ uses the strategy profile
\begin{equation}
\begin{split}
    &\delta^{a, i}_{t} = \sum_{j=1}^{N}(f^{1, j}(t)+Z^{1})X^{j}_{t} + g^{1}(t), \\
    &\delta^{b, i}_{t} = \sum_{j=1}^{N}(f^{2, j}(t)+Z^{2})X^{j}_{t} + g^{2}(t), \\
    &q^{i}_{t} = \sum_{j=1}^{N}(f^{3, j}(t)+Z^{3})X^{j}_{t} + g^{3}(t),
\end{split}
\end{equation}
where \(f^{k,j}\), for \(k\in\{1,2,3\}\) and \(j\in[N]\), are the functions corresponding to the strategy profile \eqref{eq:candidate} in non-matrix notation. We consider \(\sigma_{\text{pert}}=\{0.01,0.02,0.03\}\), which produces meaningful changes in the objective functional while remaining within the class of admissible strategies, thereby avoiding blow-up of the system. Table \ref{tab:optimality_test} shows summary statistics of these differences across the perturbation sizes, with statistically insignificant positive results removed. The brackets in Table \ref{tab:optimality_test} contain percentages of the unperturbed objective corresponding to the relevant statistic, to provide a sense of scale. We not only observe that improvements in the objective are very rare, but also that, when they do occur, they are very modest in size (approximately \(1\!-\!2\%\) of the objective) and do not increase with the perturbation size. It can also be shown that the perturbation leading to the largest improvement ceases to yield a significant improvement when applied to a second dealer, and vanishes entirely when applied to a third.
This suggests that our candidate strategy profile is very close to the constrained equilibrium. Crucially, all qualitative results presented in the following sections are robust to the above perturbations. 

\begin{table}[H]
\begin{center}
\caption{Optimality perturbation test: estimates of $\altcal{J}^i(u^i+\nu, u^{-i})-\altcal{J}^i(u^i,u^{-i})$}
\label{tab:optimality_test}
\vskip-0.2cm
\small
\begin{tabular}{llcccc}
\toprule
            &  & \textbf{Median} & \textbf{q95} & \textbf{Maximum} & \% $\mathbf{\leq 0}$ \\
\midrule
$\sigma_{\text{pert}} = 0.01$
    & Internaliser    & $-4.768$ ($-0.69\%$) & $3.106$ ($0.45\%$) & $12.783$ ($1.84\%$) & $93.8$ \\
    & Externaliser    & $-1.074$ ($-0.16\%$) & $-0.073$ ($-0.01\%$) & $-0.013$ ($-0.00\%$) & $100.0$ \\
\midrule
$\sigma_{\text{pert}} = 0.02$
    & Internaliser    & $-13.524$ ($-1.95\%$) & $-0.657$ ($-0.09\%$) & $16.843$ ($2.42\%$) & $96.5$ \\
    & Externaliser    & $-3.747$ ($-0.57\%$) & $-0.966$ ($-0.15\%$) & $-0.016$ ($-0.00\%$) & $100.0$ \\
\midrule
$\sigma_{\text{pert}} = 0.03$
    & Internaliser    & $-25.057$ ($-3.61\%$) & $-3.199$ ($-0.46\%$) & $14.236$ ($2.05\%$) & $99.2$ \\
    & Externaliser    & $-9.124$ ($-1.39\%$) & $-2.687$ ($-0.41\%$) & $-0.706$ ($-0.11\%$) & $100.0$ \\
\bottomrule
\end{tabular}
\vskip0.2cm
{\small The median, 95\textsuperscript{th} percentile, maximum, and percentage of non-positive estimates are shown. In brackets are the percentages of the unperturbed objective.}
\end{center}
\end{table}

\end{remark}

\subsection{Illustrative simulation: impact-agnostic strategy profile}  \label{sec:constrained_scenario}

We now demonstrate the impact-agnostic strategy profile \eqref{eq:candidate} in the same simulation setting as Figure \ref{fig:illustrative-unconstrained}. Specifically, the total market order flow and unaffected price movements are identical. The results are shown in Figure \ref{fig:illustrative-constrained}, which follows the same format as Figure \ref{fig:illustrative-unconstrained}. We immediately observe that the externalisation rates of the externalisers follow a much more consistent path than under the impact-aware strategy. Indeed, we effectively observe only a single path. This is because the dealers all select nearly identical trading trajectories, an effect that will be the focus of the remainder of this section. As in Figure \ref{fig:illustrative-unconstrained}, note that the blue line in the upper-middle panel cannot be seen because it is almost exactly hidden by the red line. This is again a consequence of our first result in Section \ref{sec:comparison-unconstrained}. 

Moreover, the relationship between order imbalance, externalisation rates, price impact, and deltas is now much clearer. This is most evident between 5 and 10 hours into the simulation, when the imbalance rises rapidly, driving both externalisation rates and ask-depths down, while bid-depths increase. These sharply negative trading rates then push the price impact downward, which is also reflected in the price path. However, the impact path is now much smoother, as dealers respond primarily to inventory constraints rather than engaging in short-term signal trading.

\begin{figure}[H]
  \centering

  \begin{subfigure}[b]{0.32\textwidth}
    \centering
    \includegraphics[width=\linewidth]{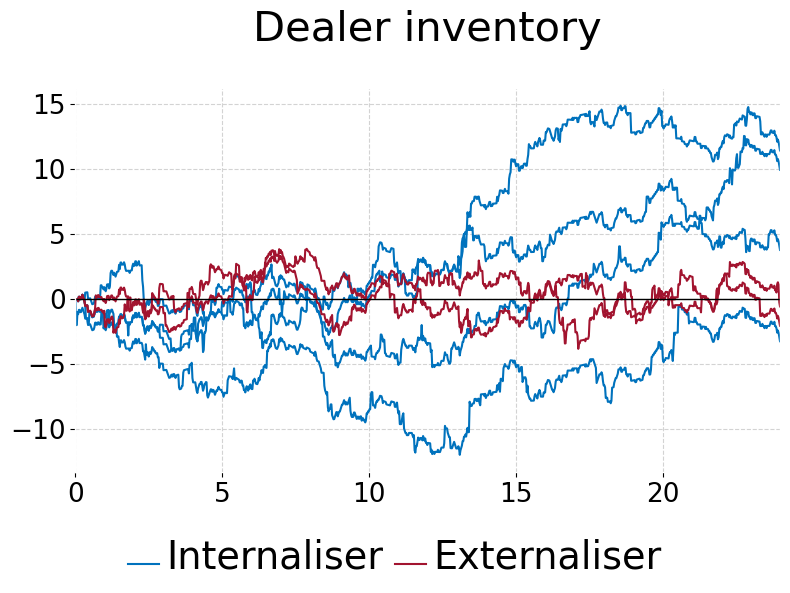}
    \label{fig:X-constrained}
  \end{subfigure}\hfill
  \begin{subfigure}[b]{0.32\textwidth}
    \centering
    \includegraphics[width=\linewidth]{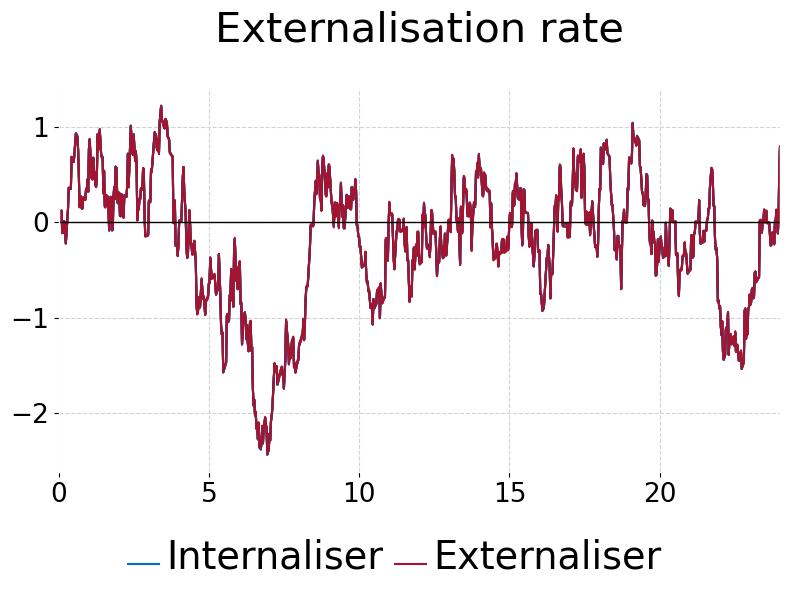}
    \label{fig:q-constrained}
  \end{subfigure}\hfill
  \begin{subfigure}[b]{0.32\textwidth}
    \centering
    \includegraphics[width=\linewidth]{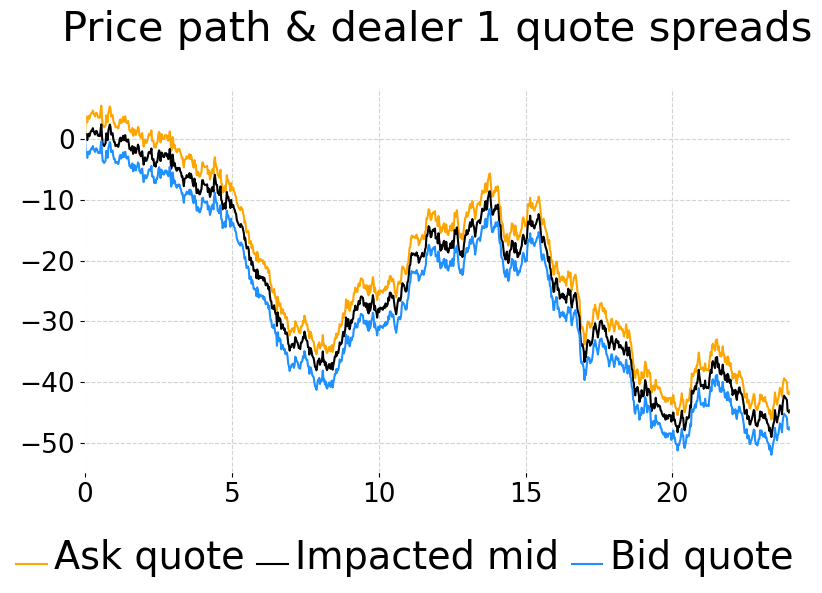}
    \label{fig:prices-constrained}
  \end{subfigure}

  \vspace{0.6em}

  \begin{subfigure}[b]{0.32\textwidth}
    \centering
    \includegraphics[width=\linewidth]{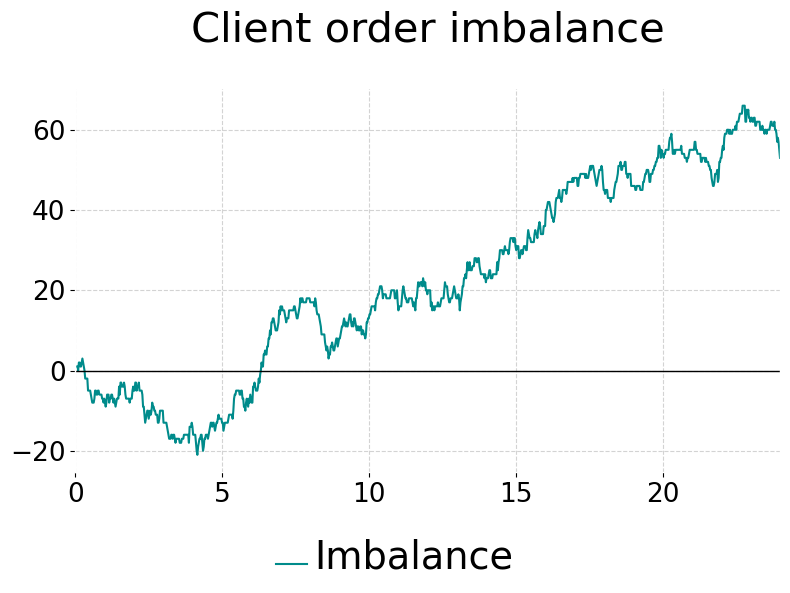}
    \label{fig:imbalance-constrained}
  \end{subfigure}\hfill
  \begin{subfigure}[b]{0.32\textwidth}
    \centering
    \includegraphics[width=\linewidth]{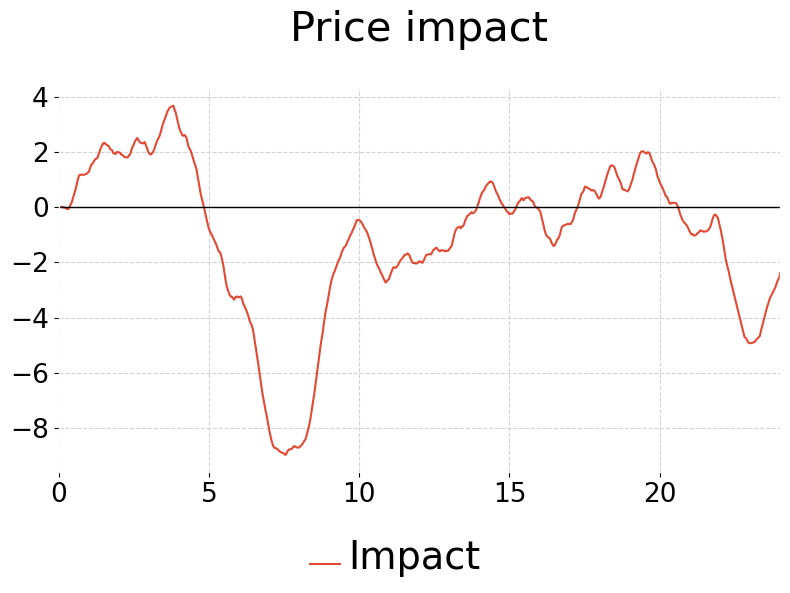}
    \label{fig:I-constrained}
  \end{subfigure}\hfill
  \begin{subfigure}[b]{0.32\textwidth}
    \centering
    \includegraphics[width=\linewidth]{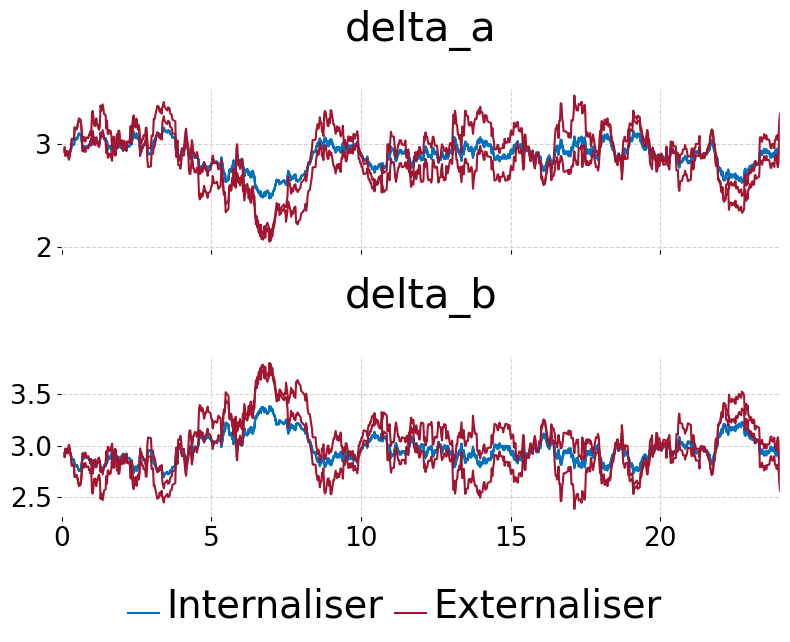}
    \label{fig:deltas-constrained}
  \end{subfigure}

  \caption{An example simulation under the impact-agnostic strategy profile with the parameters from Table \ref{tab:params}. Four dealers are internalisers and two are externalisers.}
  \label{fig:illustrative-constrained}
\end{figure}

\subsection{Comparison of dealer combinations: impact-agnostic strategy profile} \label{sec:comparison-constrained}

As in the impact-aware case, we wish to compare different combinations of dealers present in the market. The means and 95\% confidence intervals of the dealers' P\&Ls, nominal and effective spread captures, hedging costs, hedging volumes, and median client mean spread (all defined in Section \ref{sec:metrics}) are given in Table \ref{tab:keymetrics-constrained} averaged across the two dealer types. We first note that results 1 and 2 also hold under the impact-agnostic strategy profile, and we can also compare Figure \ref{fig:unconstrained-side-by-side} to Figure \ref{fig:constrained-side-by-side}. This immediately yields two further results.

\textbf{Internalisers' P\&L falls due to increased trading costs as more externalisers are present.}
When dealers are impact-agnostic, the P\&L \eqref{def:pnl} of all dealers decreases due to trading costs, as the ratio of the externalisers increases. 
The internalisers are faced with a decision between acting as externalisers (and paying increased trading costs) or continuing to act like internalisers, and having the value of their positions eroded due to price impact. They can also go further and trade the externaliser's impact as a signal, although this still carries trading costs. An example is the case where both a native internaliser and a native externaliser receive a share of bid-side order flow, hence accumulating positive inventory. The native externaliser then offloads this position  by selling, lowering the market price of the asset due to price impact. When the internalisers sit on their position and then receive ask-side order flow to offset the earlier buy, the asset price is lower than before. In such a scenario, internalisers who do not change their behaviour find that they bought high and sold low. This can be seen in Table \ref{tab:keymetrics-constrained}, which shows the mean (and 95\% confidence intervals) of P\&L. 

\textbf{Spread costs rise for the median client as more externalisers are present.}
Under the impact-agnostic strategy profile, the mean spread faced by a median client \eqref{def:median_client_best_spread} increases as more native externalisers are present in the market. Table \ref{tab:keymetrics-constrained} shows the mean (and 95\% confidence intervals) of median client's best spreads \eqref{def:median_client_best_spread} across all 1000 realisations in the different scenarios. In realisations where there are more selling clients than buying clients, dealers on average accumulate positive positions. To unwind these positions, externalisers sell, incurring negative transient price impact. This lowers the market price, so the selling clients trade at worse prices. In addition, the dealers widen spreads on the bid side in an attempt to trade less with selling clients, which also worsens the price for selling clients. The same reasoning applies for realisations where there are more buying clients than selling clients. Note that in the impact-aware case, a similar effect can be seen as the market becomes generally more dominated by externalisers (see Table \ref{tab:keymetrics-unconstrained}), but the relationship is not as strong as in the impact-agnostic case. 

\begin{figure}[H]
    \centering
    \begin{subfigure}[t]{0.48\textwidth}
        \centering
        \includegraphics[width=\linewidth]{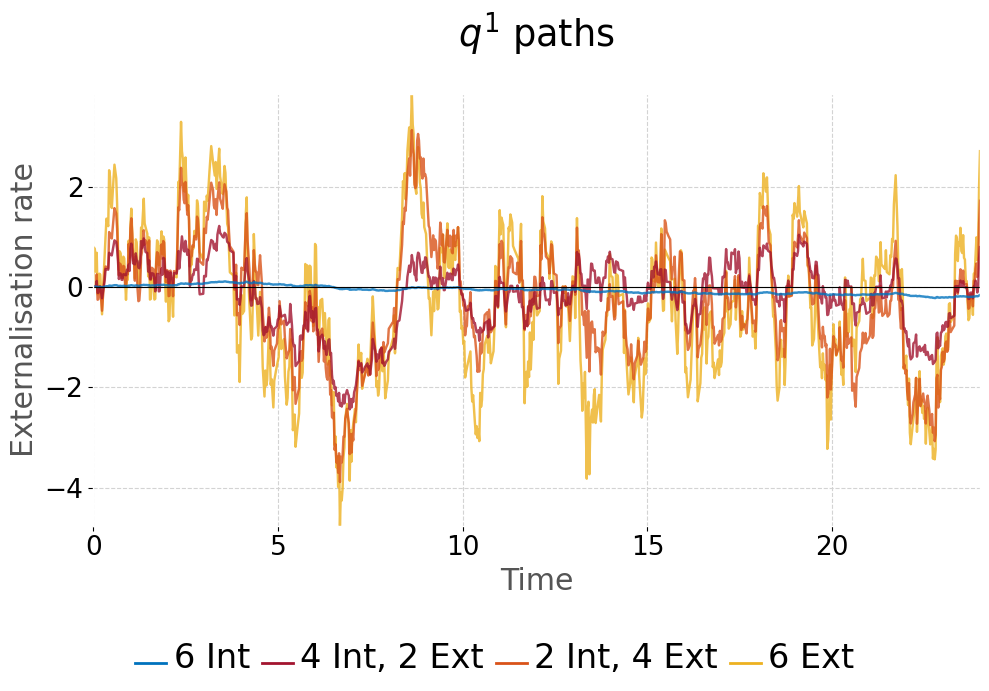}
        \caption{}
        \label{fig:q-comparison-constrained}
    \end{subfigure}
    \hfill
    \begin{subfigure}[t]{0.48\textwidth}
        \centering
        \includegraphics[width=\linewidth]{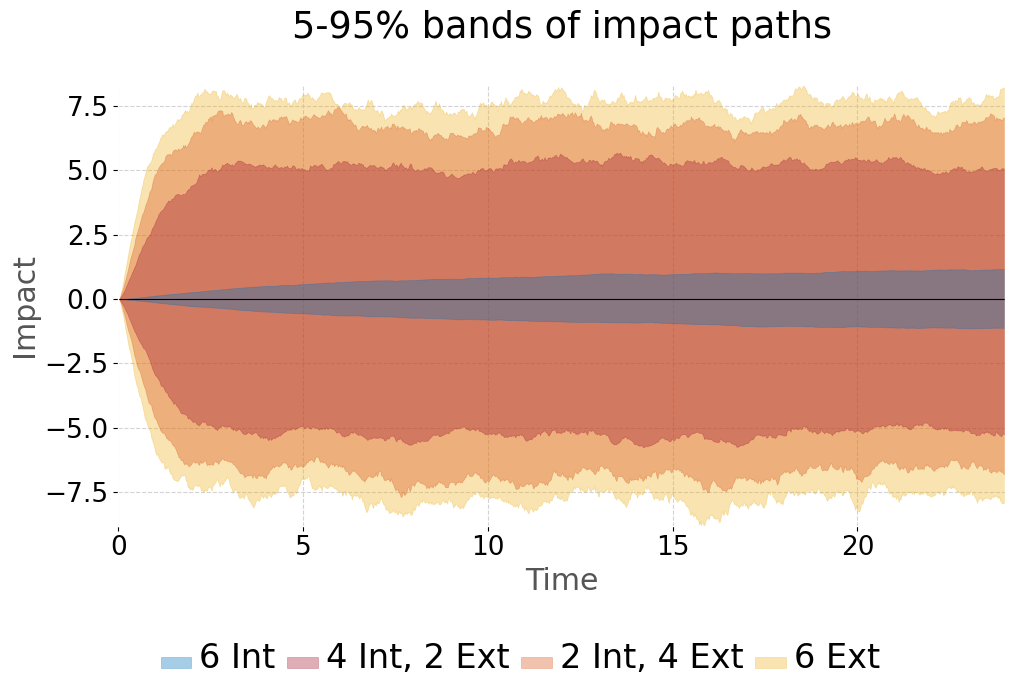}
        \caption{}
        \label{fig:I-percentiles-constrained}
    \end{subfigure}
    \caption{(a) A single path of internaliser 1's externalisation rate in Scenarios A (blue), C (red), E (orange), and G (yellow), where the scenarios are given in Table \ref{tab:dealer_types}. (b) 5--95\% percentiles for the paths of transient price impact $I$ across time in the same scenarios. Both use the impact-agnostic strategy profile.}
    \label{fig:constrained-side-by-side}
\end{figure}

\begin{table}[H]
\begin{center}
\caption{Client and dealer performance metrics: impact-agnostic strategy profile}
\label{tab:keymetrics-constrained}
\vskip-0.2cm
\small
\setlength{\tabcolsep}{4.5pt}
\fontsize{8.8pt}{11pt}\selectfont
\begin{tabular*}{\textwidth}{@{\extracolsep{\fill}}lcccccc@{}}
\toprule
            &   & \multicolumn{2}{c}{\textbf{Spread capture}} & \textbf{Hedging} & \textbf{Hedging} & \textbf{Median client} \\
\cline{3-4}
            & \textbf{P\&L} & Nominal $S$ & Effective $\altcal{S}$ & \textbf{costs} $\altcal{H}$ & \textbf{volume} $v$ & \textbf{spread} \\
\midrule
\multicolumn{6}{l}{\emph{Scenario A : 6 internalisers}} & 3.188 $\pm$ 0.021 \\
Internaliser    & 725.278 $\pm$ 5.640 & 2.925 $\pm$ 0.000 & 3.628 $\pm$ 0.003 & 1.479 $\pm$ 1.904 & 1.701 ± 0.055 &  \\
\midrule
\multicolumn{6}{l}{\emph{Scenario B : 5 internalisers \& 1 externaliser}} & 3.422 $\pm$ 0.052 \\
Internaliser    & 708.213 $\pm$ 5.549 & 2.925 $\pm$ 0.000 & 3.654 $\pm$ 0.006 & 24.516 $\pm$ 4.393 & 9.799 ± 0.104 &  \\
\cdashline{1-7}[0.4pt/1.6pt]
Externaliser    & 674.142 $\pm$ 2.809 & 2.925 $\pm$ 0.000 & 3.489 $\pm$ 0.007 & 24.778 $\pm$ 4.436 & 9.904 ± 0.105 &  \\
\midrule
\multicolumn{6}{l}{\emph{Scenario C : 4 internalisers \& 2 externalisers}} & 3.579 $\pm$ 0.060 \\
Internaliser    & 695.761 $\pm$ 5.433 & 2.925 $\pm$ 0.000 & 3.672 $\pm$ 0.008 & 40.532 $\pm$ 5.417 & 14.285 ± 0.130 &  \\
\cdashline{1-7}[0.4pt/1.6pt]
Externaliser    & 663.510 $\pm$ 2.194 & 2.925 $\pm$ 0.000 & 3.520 $\pm$ 0.007 & 40.369 $\pm$ 5.397 & 14.230 ± 0.130 &  \\
\midrule
\multicolumn{6}{l}{\emph{Scenario D : 3 internalisers \& 3 externalisers}} & 3.691 $\pm$ 0.063 \\
Internaliser    & 684.007 $\pm$ 5.663 & 2.925 $\pm$ 0.000 & 3.692 $\pm$ 0.010 & 55.911 $\pm$ 6.125 & 18.048 ± 0.141 &  \\
\cdashline{1-7}[0.4pt/1.6pt]
Externaliser    & 654.390 $\pm$ 1.919 & 2.925 $\pm$ 0.000 & 3.549 $\pm$ 0.008 & 55.398 $\pm$ 6.071 & 17.887 ± 0.140 &  \\
\midrule
\multicolumn{6}{l}{\emph{Scenario E : 2 internalisers \& 4 externalisers}} & 3.752 $\pm$ 0.066 \\
Internaliser    & 672.376 $\pm$ 6.058 & 2.925 $\pm$ 0.000 & 3.708 $\pm$ 0.012 & 68.411 $\pm$ 6.462 & 21.452 ± 0.147 &  \\
\cdashline{1-7}[0.4pt/1.6pt]
Externaliser    & 647.491 $\pm$ 1.773 & 2.925 $\pm$ 0.000 & 3.574 $\pm$ 0.008 & 67.618 $\pm$ 6.388 & 21.205 ± 0.145 &  \\
\midrule
\multicolumn{6}{l}{\emph{Scenario F : 1 internaliser \& 5 externalisers}} & 3.805 $\pm$ 0.067 \\
Internaliser    & 672.692 $\pm$ 7.750 & 2.925 $\pm$ 0.000 & 3.726 $\pm$ 0.015 & 79.492 $\pm$ 6.749 & 24.711 ± 0.155 &  \\
\cdashline{1-7}[0.4pt/1.6pt]
Externaliser    & 641.242 $\pm$ 1.664 & 2.925 $\pm$ 0.000 & 3.595 $\pm$ 0.009 & 78.431 $\pm$ 6.662 & 24.386 ± 0.153 &  \\
\midrule
\multicolumn{6}{l}{\emph{Scenario G : 6 externalisers}} & 3.843 $\pm$ 0.068 \\
Externaliser    & 636.289 $\pm$ 1.595 & 2.925 $\pm$ 0.000 & 3.616 $\pm$ 0.010 & 88.725 $\pm$ 6.951 & 27.451 ± 0.159 &  \\
\bottomrule
\end{tabular*}
\end{center}
\end{table}

\subsection{The effect of adding dealers} \label{sec:experiment2}

We investigate the effects of adding more dealers of different types to an existing pool. We begin with a pool of five internalisers and gradually increase the number of dealers. Our methodology is described as follows.
\begin{enumerate}
    \item First, we generate 1000 paths of ask-side order flow and 1000 paths of bid-side order flow.

    \item For each of these 1000 realisations, we simulate a market  with five dealers, all internalisers. The dealers use the equilibrium strategy profile from Theorem \ref{thm:solution}.

    \item We then simulate the market again for each of the same 1000 realisations for a market of six internalisers.  

    \item We repeat this process for markets of seven internalisers, eight internalisers, and so on until there are nine internalisers. Each time we use the same 1000 samples of order flow.

    \item We now return to the pool of just five internalisers, and add a sixth dealer who is an externaliser. We then rerun the 1000 realisations. 

    \item We repeat the procedure for five internalisers and two native externalisers, and so on until there are five internalisers and four native externalisers. 
\end{enumerate}
We observe in Figure \ref{fig:spread-adding-unconstrained} that under the impact-aware strategy profile \eqref{eq:feedback}, adding more dealers always reduces the client's costs \eqref{def:median_client_best_spread} due to increased competition among dealers. However, the benefit is stronger if internalisers are added as there are lower price impact costs absorbed by clients. Under the impact-agnostic strategy profile, the effect of increased impact costs dominates the effect of increased competition. This can be seen in Figure \ref{fig:spread-adding-constrained} where adding externalisers worsens client's costs substantially. Moreover, we see from the magnitudes on the vertical axes in Figure \ref{fig:Spreads Adding} that clients in general have lower spread costs when dealers are impact-aware. This happens because a market with impact-agnostic dealers has more impact than a market with impact-aware dealers, even when all dealers are internalisers, and these impact costs are passed on to the median client. 

\begin{figure}[H]
    \centering
    \begin{subfigure}[t]{0.48\textwidth}
        \centering
        \includegraphics[width=\linewidth]{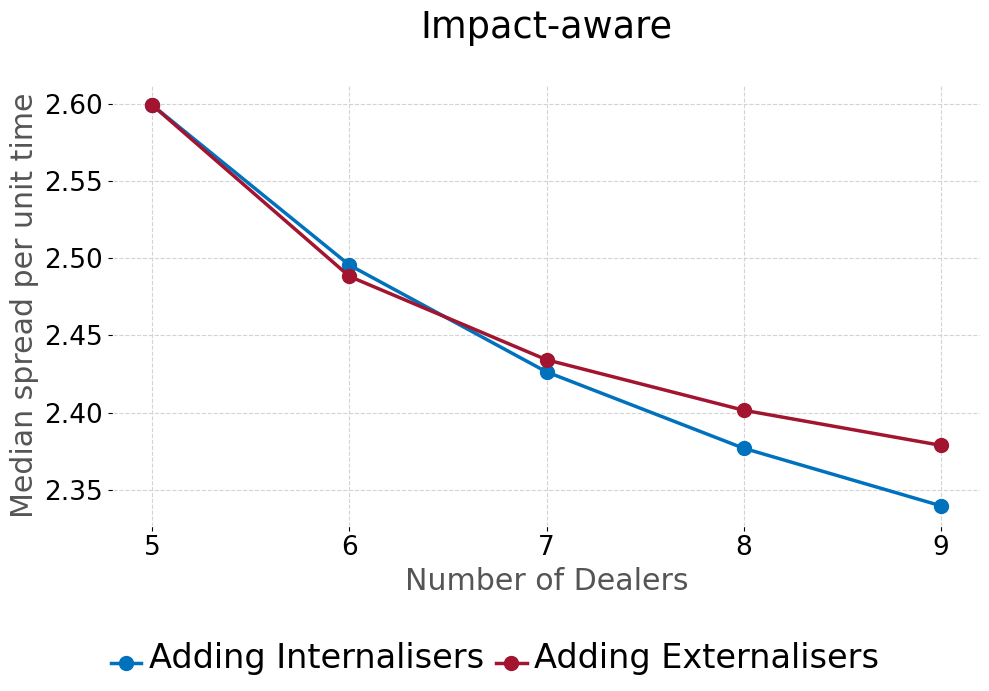}
        \caption{}
        \label{fig:spread-adding-unconstrained}
    \end{subfigure}
    \hfill
    \begin{subfigure}[t]{0.48\textwidth}
        \centering
        \includegraphics[width=\linewidth]{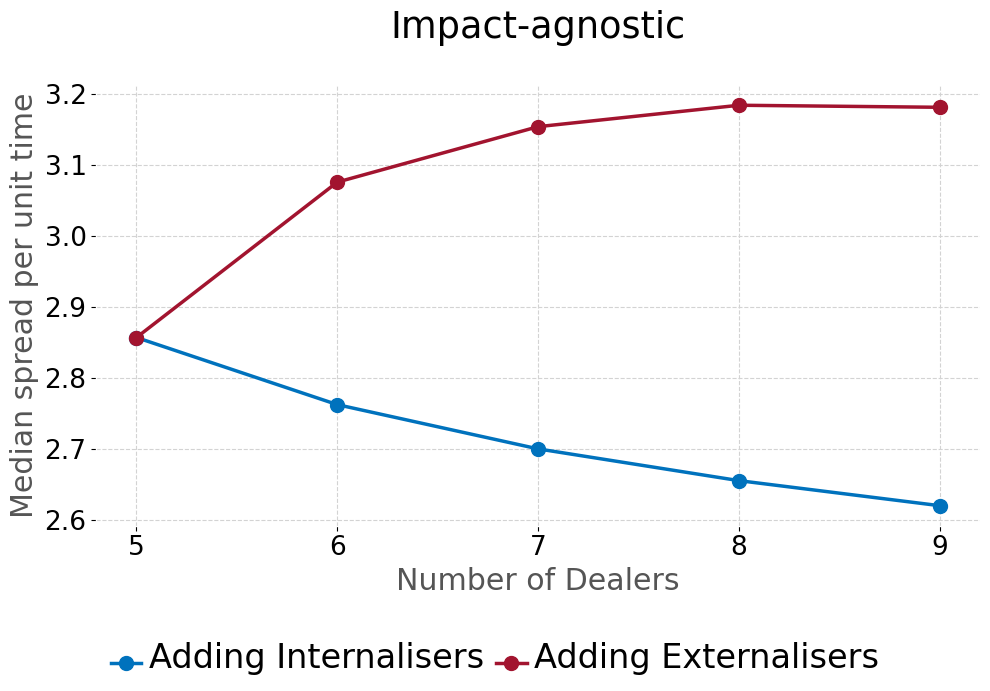}
        \caption{}
        \label{fig:spread-adding-constrained}
    \end{subfigure}
    \caption{Median client's mean spread as dealers of different types are added to an initial pool of five internalisers.}
    \label{fig:Spreads Adding}
\end{figure}

\subsection{Robustness with respect to model parameters} \label{sec:robustness}

We observed in the results of Section \ref{sec:experiment2} an effect whereby the presence of native externalisers incentivises native internalisers to trade in the same manner, thereby worsening dealers' P\&Ls and clients' spreads. Since all internalisers are incentivised to behave in this way, yet incur lower P\&Ls as a consequence, this phenomenon can be viewed as a type of prisoner's dilemma, similar to the prisoner's dilemma studied in Section 5 of \cite{oomen2017aggregator}. In this final numerical section, we investigate the dependence of this prisoner's dilemma effect on some of the model's most important parameters. Specifically, we fix the `baseline parameters' to be those in Table \ref{tab:params}, and vary \(\lambda, \beta, \phi^{I}, \phi^{X}, \alpha^{I}, \alpha^{X}\), where \(\phi^{I}, \alpha^{I}\) correspond to internalisers and \(\phi^{X}, \alpha^{X}\) correspond to externalisers. As \(\phi^{I}\) and \(\phi^{X}\) are varied, we vary \(\alpha^{I}\) and \(\alpha^{X}\), respectively, in lockstep. We only test parameter combinations satisfying Assumption \ref{assum:bounded}, with \(\phi^{I}<\phi^{X}\).

Figure \ref{fig:robustness-constrained} shows a robustness check for the impact-agnostic strategy profile \eqref{eq:candidate}, for different values of the parameters $\lambda$ and $\beta$ (left column) and $\phi^{I}$ and $\phi^{X}$ (right column). The percentage change in internaliser's P\&L \eqref{def:pnl} (top panels) and median client's spread \eqref{def:median_client_best_spread} (bottom panels) is plotted. Where P\&L is plotted, red indicates a decrease as internalisers are swapped for externalisers. Where median client's spread is plotted, red indicates an increase as internalisers are swapped for externalisers. Figure \ref{fig:robustness-unconstrained} shows the same for the impact-aware strategy profile \eqref{eq:feedback}. 

\begin{figure}[H]
    \centering
    \includegraphics[width=0.8\linewidth]{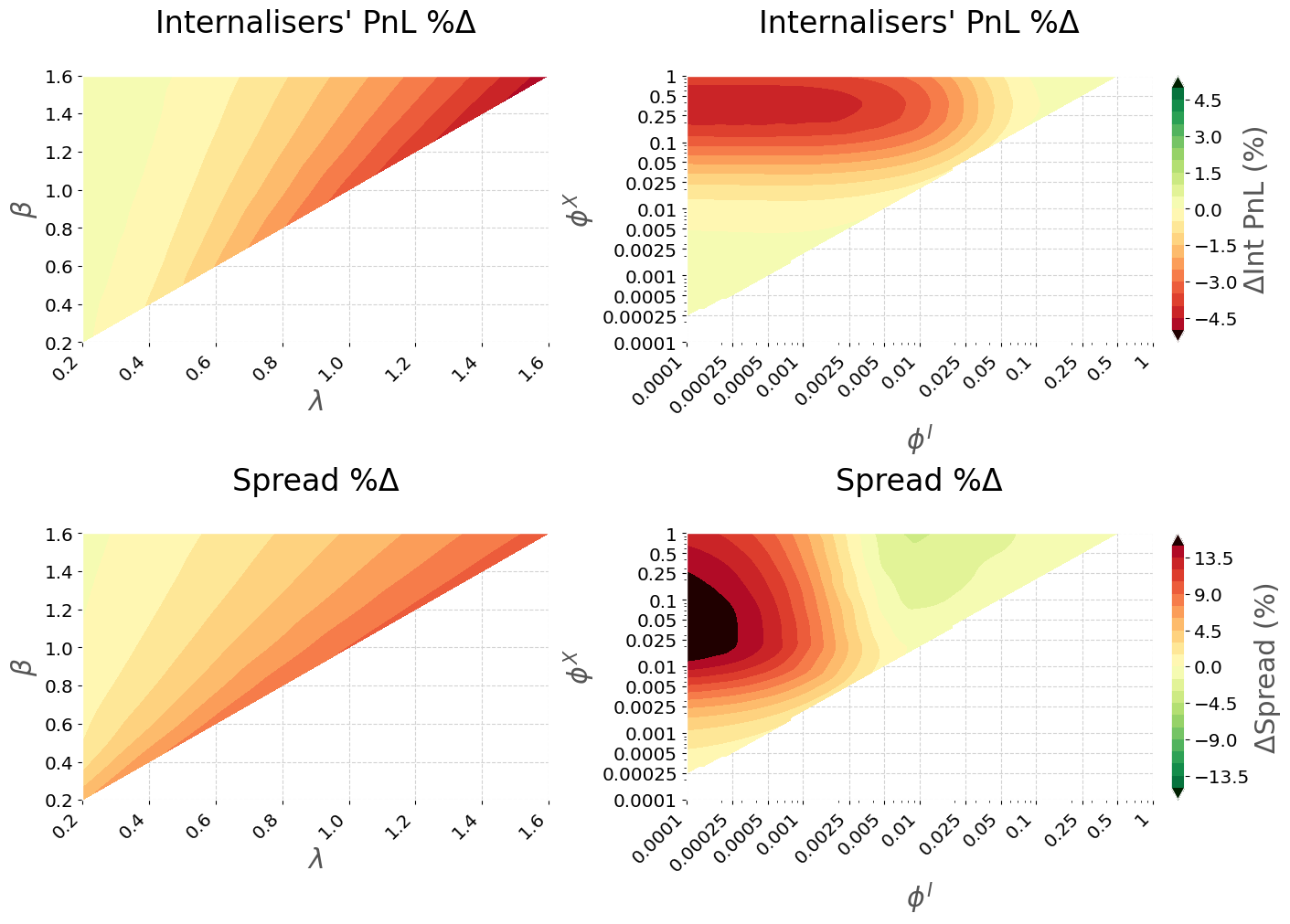}
    \caption{The effect of swapping two dealers in a pool of six internalisers for externalisers, on P\&L and median client spread, as parameters are varied under the impact-agnostic strategy profile.}
    \label{fig:robustness-constrained}
\end{figure}

\begin{figure}[H]
    \centering
    \includegraphics[width=0.8\linewidth]{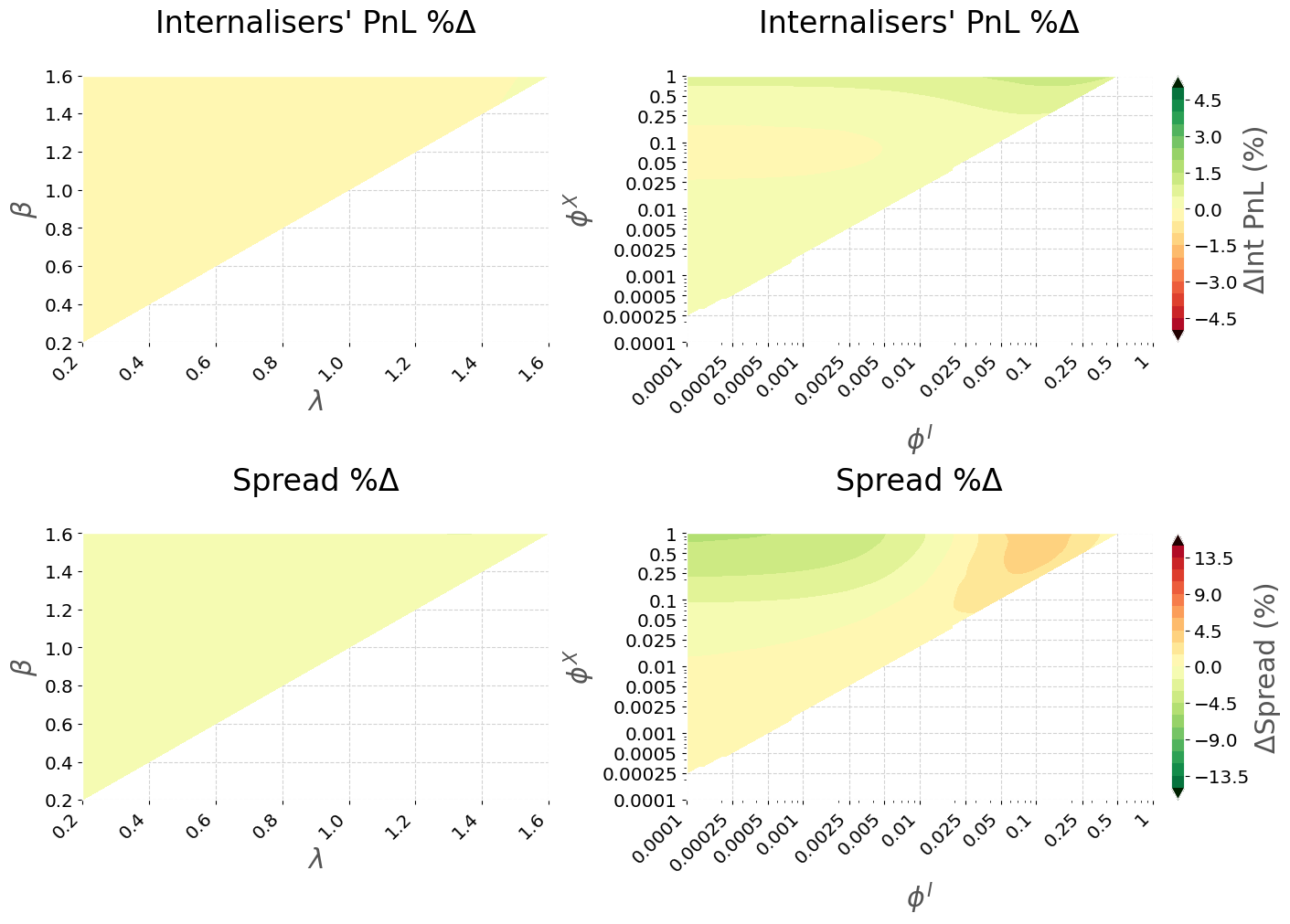}
    \caption{The effect of swapping two dealers in a pool of six internalisers for externalisers, on P\&L and median client spread, as parameters are varied under the impact-aware strategy profile.}
    \label{fig:robustness-unconstrained}
\end{figure}

This experiment yields two takeaways. First, using the impact-agnostic strategy profile, the existence of the prisoner's dilemma effect is robust to changes in the model parameters, although the magnitude of the effect varies.\footnote{One small exception is for the spread when internalisers are almost indistinguishable from externalisers. Possible explanations for this include that in this case the prisoner's dilemma effect is either too difficult to detect without a prohibitively high sample size, or that some other tiny and thus usually invisible effect takes over.} 
Second, we see again that the prisoner's dilemma effect only exists under the impact-agnostic case, at least as the first two dealers are swapped. A possible explanation for this is the following. The transient price impact acts as a common price signal, but one that is created by the dealers themselves. If dealers use this impact distortion to set their trading strategies, they trade it as a short-term price signal, and in doing so drastically shrink the impact itself as shown in \cite{neuman2023trading}. As externalisers are added, additional impact costs are therefore not larger than the improvement in P\&L from the additional effective spread capture. Thus when dealers are impact-aware, internalisers' P\&L rises (or at least does not fall) as externalisers are added. If instead dealers are impact-agnostic, the reverse is true and internalisers' P\&L falls as externalisers are added.

\begin{appendices}

\section{Proof of Theorem \ref{thm:solution}}\label{appendix:proof}

Let $N \in \mathbb{N}$. Throughout this section, we use the following notation, for stochastic processes $x^{1}, ..., x^{N}$ and non-zero constants $c^{1}, ..., c^{N}$:
\begin{equation} \label{av-not} 
\begin{aligned} 
&    \overline{x}_{t} := \frac{1}{N}\sum_{i=1}^{N}x^{i}_{t}, 
    \qquad
    \overline{\frac{1}{c}} := \frac{1}{N}\sum_{i=1}^{N}\frac{1}{c^{i}}, 
    \qquad
    \overline{\left(\frac{x}{c}\right)}_{t} := \frac{1}{N}\sum_{i=1}^{N}\frac{x^{i}_{t}}{c^{i}}.\\ 
&     \overline{x^{-i}_t} := \frac{1}{N-1}\sum_{j\neq i}x^{j}_{t},
    \qquad
    \overline{\frac{1}{c}} := \frac{1}{N-1}\sum_{j\neq i}\frac{1}{c^{j}},
    \qquad
    \overline{\left(\frac{x}{c}\right)^{-i}_t} := \frac{1}{N-1}\sum_{j\neq i}\frac{x^{j}_{t}}{c^{j}}.
    \end{aligned}
\end{equation}
In order to prove Theorem \ref{thm:solution}, we first prove the concavity property of the objective functional \eqref{def:objective}. Recall that the constants $\lambda$, $\beta$ and $\kappa$ were defined in \eqref{def:Y} and \eqref{def:Z^i}.
\begin{proposition} \label{prop:uniqueness}
Let $Z^{a}$ and $Z^{b}$ in \eqref{def:Z^i} be Poisson processes with intensities $\rho^a$ and $\rho^b$, respectively. Assume that $ \frac{\beta}{\lambda} \geq \frac{\kappa(N-1)}{4N^{2}}(\rho^a+\rho^b)$. Then the objective functional $u^{i}\rightarrow\altcal{J}^{i}(u^{i}, u^{-i})$ is strictly concave.
\end{proposition}
\begin{proof}
We will prove that for any $\eta \in (0,1)$, and for any admissible $\{u^{i}\}_{i\in [N]}$, 
\begin{equation} \label{conc-def} 
     \altcal{J}^{i}(\eta u^{i} + (1-\eta) v, u^{-i}) - \eta \altcal{J}^{i}(u^{i}, u^{-i}) - (1-\eta) \altcal{J}^{i}(v, u^{-i})>0,
 \end{equation}
 which concludes the result.  
In order to do that, we decompose the cost functional \eqref{def:objective} to the following components: 
\begin{equation} \label{objective1} 
\begin{split}
   \hat{ \altcal{J}^{i}}(u^{i}, u^{-i}) = \mathbb{E}\Bigg[\int_{0}^{T}\left(P^{\ast}_{t}+I_{t}+\delta^{a, i}_{t}\right)dZ^{a, i}_{t} -
    \int_{0}^{T}\big(P^{\ast}_{t}+I_{t}-\delta^{b, i}_{t}\big)dZ^{b, i}_{t} \Bigg], 
\end{split}
\end{equation}  
 and,  
\begin{equation} \label{objective2}
\wt{\altcal{J}^{i}}(u^{i}, u^{-i}) := \mathbb{E}\Bigg[ - \int_{0}^{T}\big(P^{\ast}_{t}+I_{t}+\frac{1}{2}\epsilon^{i} q_t^i\big)q^{i}_{t}dt 
- \phi^{i}\int_{0}^{T}\left(X^{i}_{t}\right)^{2}dt - \alpha^{i} \left(X^{i}_{T}\right)^{2} + (P^{\ast}_{T}X^{i}_{T} - px^{i}_{0}) \Bigg]. 
\end{equation}  
We first handle \eqref{objective1}. Denote by 
\begin{equation} \label{def:average_delta_a}
    \overline{\delta^{a, -i}_t} = \frac{1}{N-1}\sum_{j\in[N]\setminus \{i\}} \delta^{a, j}_{t},
\end{equation}
and 
\begin{equation} \label{def:average_delta_b}
    \overline{\delta^{
b, -i}_t } = \frac{1}{N-1}\sum_{j\in[N]\setminus \{i\}}\delta^{b, j}_{t}.
\end{equation}
We note that 
\begin{equation}
    \overline{\delta^{a}_t}= \frac{N-1}{N}\overline{\delta^{a, -i}_t}+ \frac{1}{N}\delta^{a, i}_{t} \qquad \text{and} \qquad  \overline{\delta^{b}}_{t} = \frac{N-1}{N}\overline{\delta^{b, -i}_t}+ \frac{1}{N}\delta^{b, i}_{t}, 
\end{equation}
so \eqref{def:Z^i} can be rewritten as follows: 
\begin{equation} \label{eq:Zi_rewritten}
\begin{aligned} 
    Z^{a, i}_{t} &= \frac{1}{N}\int_{0}^{t}\left(1+\frac{\kappa(N-1)}{N}\overline{\delta^{a, -i}_s} - \frac{\kappa(N-1)}{N}\delta^{a, i}_{s}\right)dZ^{a}_{s} +J^{a, i}_{t}, \\
     Z^{b, i}_{t} &= \frac{1}{N}\int_{0}^{t}\left(1+\frac{\kappa(N-1)}{N}\overline{\delta^{b, -i}_s} - \frac{\kappa(N-1)}{N}\delta^{b, i}_{s}\right)dZ^{b}_{s}+J^{b, i}_{t}.
\end{aligned} 
\end{equation}
Let $Z^{a, i}, Z^{b, i}, Q^{i}, X^{i}$ and $I$ be the controlled states by the admissible control $u^{i}=\left(\delta^{a, i}, \delta^{b, i}, q^{i}\right)$. Moreover, let $\check{Z}^{a, i}, \check{Z}^{b, i}, \check{Q}^{i}, \check{X}^{i}$ and $\check{I}$ be controlled by a different admissible control $v=\left(\gamma^{a}, \gamma^{b}, h\right)$. 
Finally, let $\eta \in (0,1)$ and denote by $\wt{Z}^{a, i}, \wt{Z}^{b, i}, \wt{Q}^{i}, \wt{X}^{i}$ and $\wt{I}$ the states controlled by $\eta u^{i}+(1-\eta) v^{i}$. Then, a lengthy but straightforward calculation using \eqref{def:Y} and \eqref{eq:Zi_rewritten} gives, 
\begin{equation} \label{eq:Za_concavity}
\begin{split}
  &\mathcal{I}^{i,(a)}(\eta u^{i} + (1-\eta) v, u^{-i}) \\
    &:= \int_{0}^{T}\left(P^{\ast}_{t}+\wt{I}_{t}+\eta\delta^{a, i}_{t} + (1-\eta)\gamma^{a}_{t}\right)d\wt{Z}^{a, i}_{t} \\
    &\qquad- \eta \int_{0}^{T}\left(P^{\ast}_{t}+I_{t}+\delta^{a, i}_{t}\right)dZ^{a, i}_{t} - (1-\eta)\int_{0}^{T}\left(P^{\ast}_{t}+\check{I}_{t}+\gamma^{a}_{t}\right)d\check{Z}^{a, i}_{t} \\
    &=\eta(1-\eta)\frac{\kappa(N-1)}{N^{2}}\left(\int_{0}^{t}\left(\delta^{a, i}_{s} - \gamma^{a}_{s}\right)^{2}dZ^{a}_{t}+
      \int_{0}^{T}\left(I_{t}-\check{I}_{t}\right) \big(\delta^{a, i}_{t}-\gamma^{a}_{t} \big) dZ^{a}_{t}\right)
\end{split}
\end{equation}
and similarly, also using again \eqref{def:Y} and \eqref{eq:Zi_rewritten},  
\begin{equation} \label{eq:Zb_concavity}
\begin{split}
  &\mathcal{I}^{i,(b)}(\eta u^{i} + (1-\eta) v, u^{-i}) \\
    &:=\int_{0}^{T}\left(P^{\ast}_{t}+\wt{I}_{t}-\eta\delta^{b, i}_{t} - (1-\eta)\gamma^{b}_{t}\right)d\wt{Z}^{b, i}_{t} \\
    &\qquad\quad- \eta \int_{0}^{T}\left(P^{\ast}_{t}+I_{t}-\delta^{b, i}_{t}\right)dZ^{b, i}_{t} - (1-\eta)\int_{0}^{T}\left(P^{\ast}_{t}+\check{I}_{t}-\gamma^{b}_{t}\right)d\check{Z}^{b, i}_{t} \\
    &=\eta(1-\eta)\frac{\kappa(N-1)}{N^{2}}\left(-\int_{0}^{t}\left(\delta^{b, i}_{s} - \gamma^{b}_{s}\right)^{2}dZ^{b}_{t}   + \int_{0}^{T}\left(I_{t}-\check{I}_{t}\right)\big(\delta^{b, i}_{t}-\gamma^{b}_{t}\big)dZ^{b}_{t}\right).
\end{split}
\end{equation}
Recall that $Z^{a}, Z^{b}$ are increasing processes. Using Young's inequality we get for $\ell \in \{a,b\}$, 
\begin{equation} \label{yng}
  \left|\int_{0}^{T}\left(I_{t}-\check{I}_{t}\right)\big(\delta^{\ell, i}_{t}-\gamma^{\ell}_{t}\big)dZ^{\ell}_{t} \right|\\
\leq  \frac{1}{2} \left(\int_{0}^{T}\left(I_{t}-\check{I}_{t}\right)^2 dZ^{\ell}_{t} + \int_{0}^{T} \big(\delta^{\ell, i}_{t}-\gamma^{\ell}_{t}\big)^2dZ^{\ell}_{t} \right) .
\end{equation}
Using again the fact that $Z^{a}, Z^{b}$ are increasing processes we have, 
\begin{equation} \label{pos111} 
 \int_{0}^{t}\left(\delta^{\ell, i}_{s} - \gamma^{\ell}_{s}\right)^{2}dZ^{\ell}_{t} > 0, \quad \ell \in \{a,b\}. 
\end{equation}
From \eqref{eq:Za_concavity}--\eqref{pos111} we get, 
\begin{equation} \label{pos11} 
\begin{aligned} 
& \mathcal{I}^{i,(a)}(\eta u^{i} + (1-\eta) v, u^{-i}) - \mathcal{I}^{i,(b)}(\eta u^{i} + (1-\eta) v, u^{-i})  \\
&= \eta(1-\eta)  \frac{\kappa(N-1)}{N^{2}} \bigg( \sum_{\ell \in\{a,b\}} \int_{0}^{t}\left(\delta^{\ell, i}_{s} - \gamma^{\ell}_{s}\right)^{2}dZ^{\ell}_{t} \\
&\qquad\qquad\qquad\qquad\quad+\frac{1}{2}\int_{0}^{T}\left(I_{t}-\check{I}_{t}\right)\big(\delta^{a, i}_{t}-\gamma^{a}_{t}\big)dZ^{a}_{t}-\frac{1}{2}\int_{0}^{T}\left(I_{t}-\check{I}_{t}\right)\big(\delta^{b, i}_{t}-\gamma^{b}_{t}\big)dZ^{b}_{t}\bigg) \\
&\geq \eta(1-\eta)  \frac{\kappa(N-1)}{N^{2}}  \sum_{\ell \in\{a,b\}}\bigg( \int_{0}^{t}\left(\delta^{\ell, i}_{s} - \gamma^{\ell}_{s}\right)^{2}dZ^{\ell}_{t}
- \frac{1}{2} \bigg|\int_{0}^{T}\left(I_{t}-\check{I}_{t}\right)\big(\delta^{\ell, i}_{t}-\gamma^{\ell}_{t}\big)dZ^{\ell}_{t}\bigg| \bigg)   \\
&\geq  \eta(1-\eta)  \frac{\kappa(N-1)}{4N^{2}}  \sum_{\ell \in\{a,b\}}\left( \int_{0}^{t}\left(\delta^{\ell, i}_{s} - \gamma^{\ell}_{s}\right)^{2}dZ^{\ell}_{t}-  \int_{0}^{T}\left(I_{t}-\check{I}_{t}\right)^2 dZ^{\ell}_{t} \right).
\end{aligned} 
\end{equation}
Together with \eqref{objective1} it follows that, 
 \begin{equation} \label{pos} 
\begin{aligned} 
 &  \hat{\altcal{J}^{i}}(\eta u^{i} + (1-\eta) v, u^{-i}) - \eta \hat{\altcal{J}^{i}}(u^{i}, u^{-i}) - (1-\eta) \hat{\altcal{J}^{i}}(v, u^{-i}) \\
&=\E\big[\mathcal{I}^{i,(a)}(\eta u^{i} + (1-\eta) v, u^{-i}) - \mathcal{I}^{i,(b)}(\eta u^{i} + (1-\eta) v, u^{-i}) \big]
 \\
&>-\eta(1-\eta) \frac{\kappa(N-1)}{4N^{2}}  \sum_{\ell \in\{a,b\}} \E\left[\int_{0}^{T}\left(I_{t}-\check{I}_{t}\right)^2 dZ^{\ell}_{t}\right]. 
\end{aligned} 
\end{equation}
Next we derive a lower bound on the convex difference of \eqref{objective2}. 
Note that for the first term on the right-hand side of \eqref{objective2} we have 
\begin{equation} \label{eq:PY_concavity}
\begin{split}
    &-\int_{0}^{T}\left(P^{\ast}_{t} + \wt{I}_{t} + \frac{1}{2}\epsilon^{i} \left(\eta q^{i}_{t} + (1-\eta) h_{t}\right)\right)\left(\eta q^{i}_{t} + (1-\eta) h_{t}\right)dt \\
    &+\eta \int_{0}^{T}\left(P^{\ast}_{t} + I_{t} + \frac{1}{2}\epsilon^{i} q^{i}_{t}\right)q^{i}_{t}dt +(1-\eta) \int_{0}^{T}\left(P^{\ast}_{t} + \check{I}_{t} + \frac{1}{2}\epsilon^{i} h_{t}\right)h_{t}dt \\
    &=\eta(1-\eta)\int_{0}^{T}\big( (I_{t}-\check{I}_{t} ) (q^{i}_{t}-h_{t} ) + \frac{1}{2}\epsilon^{i} (q^{i}_{t}-h_{t} )^{2} \big)dt.
\end{split}
\end{equation}
 From \eqref{def:Y} we get, 
 \begin{equation} \label{eq:Yi_dynamics}
    d(I_{t}-\check{I}_{t}) = \left(-\beta (I_{t}-\check{I}_{t})+\lambda (q^{i}_{t}- h_t)\right)dt, 
\end{equation}
which since $\lambda\neq0$, is used together with integration by parts to get, 
\begin{equation} \label{eq:Yi2_concavity}
\begin{aligned} 
    \int_{0}^{T} (I_{t} - \check{I}_{t} ) (q^{i}_{t}-h_{t} )dt &=\frac{1}{\lam} \int_0^T (I_{t} - \check{I}_{t} ) d (I_{t} - \check{I}_{t} )  +\frac{\beta}{\lambda}\int_{0}^{T}  (I_{t}-\check{I}_{t} )^{2} dt   \\
    &= \frac{1}{2\lambda} (I_{t}-\check{I}_{T} )^{2} + \frac{\beta}{\lambda}\int_{0}^{T}  (I_{t}^{i}-\check{I}^{i}_{t} )^{2} dt.
 \end{aligned} 
\end{equation}
Note also that $\wt{X}^{i}_{t} = \eta X^{i}_{t} + (1-\eta)\check{X}^{i}_{t}$ and hence
\begin{equation} \label{eq:X_concavity}
    P^{\ast}_{t}\wt{X}^{i}_{t} - \eta P^{\ast}_{t}X^{i}_{t} - (1-\eta)P^{\ast}_{t}\check{X}^{i}_{t} = 0,
\end{equation}
and
\begin{equation} \label{eq:X2_concavity}
   - \left(\wt{X}^{i}_{t}\right)^{2} + \eta \left(X^{i}_{t}\right)^{2} + (1-\eta)\left(\check{X}^{i}_{t}\right)^{2} = \eta(1-\eta)\left(X^{i}_{t} - \check{X}^{i}_{t}\right)^{2} \geq 0.
\end{equation}
From \eqref{eq:PY_concavity}--\eqref{eq:X2_concavity} and \eqref{objective2} we conclude that 
\begin{equation} \label{pos2} 
\begin{aligned} 
 &\wt{\altcal{J}^{i}}(\eta u^{i} + (1-\eta) v, u^{-i}) - \eta \wt{\altcal{J}^{i}}(u^{i}, u^{-i}) - (1-\eta) \wt{\altcal{J}^{i}}(v, u^{-i})\\\
 &\geq \eta(1-\eta)\frac{\beta}{\lambda} \E\left[\int_{0}^{T}  (I_{t}-\check{I}_{t} )^{2} dt\right]. 
\end{aligned} 
\end{equation}
From \eqref{objective1}, \eqref{objective2},  \eqref{pos}, \eqref{pos2} and \eqref{def:objective} we obtain

\begin{equation} \label{conc-1}
\begin{aligned}  
  &   \altcal{J}^{i}(\eta u^{i} + (1-\eta) v, u^{-i}) - \eta \altcal{J}^{i}(u^{i}, u^{-i}) - (1-\eta) \altcal{J}^{i}(v, u^{-i}) \\ 
  & > \eta(1-\eta)\left(\frac{\beta}{\lambda} \E\left[\int_{0}^{T}  (I_{t}-\check{I}_{t} )^{2} dt\right]-  \frac{\kappa(N-1)}{4N^{2}}    \sum_{\ell \in\{a,b\}} \E\left[\int_{0}^{T}\left(I_{t}-\check{I}_{t}\right)^2 dZ^{\ell}_{t}\right]\right)
 \end{aligned} 
 \end{equation}
 Using the assumption that $Z^{a}$ and $Z^{b}$ in \eqref{def:Z^i} are Poisson processes with intensities $\rho^a$ and $\rho^b$, respectively, it then follows that \eqref{conc-def} holds if 
 $$
 \frac{\beta}{\lambda} \geq \frac{\kappa(N-1)}{4N^{2}}(\rho^a+\rho^b). 
 $$
 \end{proof} 
 
The proof of Theorem \ref{thm:solution} uses a variational approach as in \cite{neuman2023trading}. We first derive the Gateaux derivative of $\altcal{J}^{i}$ in Lemma \ref{lemma:gateaux}, then from the resulting first-order condition we obtain a system of FBSDEs  (see Proposition \ref{prop:FBSDE}). We then explicitly solve the system and, combined with Proposition \ref{prop:uniqueness}, the result of Theorem \ref{thm:solution} follows. 

By Proposition \ref{prop:uniqueness} and Proposition 2.1 from Chapter 2 of \cite{ekeland1999convex}, $\altcal{J}^{i}(\cdot)$ admits a maximiser $u \in \mathcal A$ if the Gateaux derivative satisfies
\begin{equation} \label{def:Gateaux}
    \langle(\altcal{J}^{i})'(u, u^{-i}), v\rangle := \underset{\eta\rightarrow0}{\lim}\left(\frac{\altcal{J}^{i}(u+\eta v, u^{-i})-\altcal{J}^{i}(u, u^{-i})}{\eta}\right) =0, \quad \textrm{for all } v \in\altcal{A}. 
\end{equation}
We derive an explicit expression for the Gateaux derivative of $\altcal{J}$ in Lemma \ref{lemma:gateaux} below. 
\begin{lemma} \label{lemma:gateaux}
    Let $\altcal{J}$ be the objective functional \eqref{def:objective}. Then, for any admissible strategy \newline$v=(\gamma^{a}, \gamma^{b}, h) \in  \altcal{D} \times  \altcal{D}  \times \altcal{A}$ (see \eqref{adms}, \eqref{adms2})  we have,  
    \begin{equation} \label{eq:gateaux}
    \begin{split}
        & \langle(\altcal{J}^{i})'(u^i, u^{-i}), v\rangle  \\ 
        &=\mathbb{E}\Bigg[
        \int_{0}^{T}\gamma^{a}_{s}\Bigg(\rho^{a}\bigg(\frac{1}{N}\left(1-\kappa(\delta^{a, i}_{s}-\overline{\delta^{a}}_{s})\right) - \kappa\frac{N-1}{N^{2}}\left(P^{\ast}_{s}+I_{s}+\delta^{a, i}_{s}\right) \\ 
        &\qquad\quad-2\phi^{i}\kappa\frac{N-1}{N^{2}}\int_{s}^{T}X^{i}_{t}dt + \kappa\frac{N-1}{N^{2}}\left(P^{\ast}_{s}-2\alpha^{i}X^{i}_{T}\right)\bigg) + \mu\Bigg)ds \\
        &\qquad+\int_{0}^{T}\gamma^{b}\Bigg(\rho^{b}\bigg(\frac{1}{N}\left(1-\kappa(\delta^{b,i }_{s}-\overline{\delta^{b}}_{s})\right) + \kappa\frac{N-1}{N^{2}}\left(P^{\ast}_{s} + I_{s} - \delta^{b, i}_{s}\right) \\
        &\qquad\quad+2\phi^{i}\kappa\frac{N-1}{N^{2}}\int_{s}^{T}X^{i}_{t}dt - \kappa\frac{N-1}{N^{2}}\left(P^{\ast}_{s}-2\alpha^{i}X^{i}_{T}\right)\bigg) + \mu\Bigg)ds \\
        &\qquad+\int_{0}^{T}h_{s}\bigg(\frac{\lambda}{N}\int_{s}^{T}e^{-\beta(t-s)}\left(1-\kappa(\delta^{a, i}_{t}-\overline{\delta^{a}}_{t})\right)dZ^{a}_{t} \\
        &\qquad\quad-\frac{\lambda}{N}\int_{s}^{T}e^{-\beta(t-s)}\left(1-\kappa(\delta^{b, i}_{t}-\overline{\delta^{b}}_{t})\right)dZ^{b}_{t} \\
        &\qquad\quad-\lambda \int_{s}^{T}e^{-\beta(t-s)}q^{i}_{t}dt - P^{\ast}_{s} - I_{s} - \epsilon^{i} q^{i}_{s} \\
        &\qquad\quad-2\phi^{i}\int_{s}^{T}X^{i}_{t}dt + (P^{\ast}_{t}-2\alpha^{i}X^{i}_{T})\bigg)ds\Bigg] ,\quad i=1,...,N. 
    \end{split}
    \end{equation}
\end{lemma}
\begin{proof}
    Let $(Z^{a, i}, Z^{b, i}, Q^{i}, X^{i},I)$ be the state processes controlled by a set of admissible controls $u^{i}=\left(\delta^{i, a}, \delta^{i, b}, q^{i}\right)$ and assume that the admissible controls $u^{-i}$ are fixed. For any admissible set of controls $v^{i}=\left(\gamma^{a}, \gamma^{b}, h\right)$, denote by $(\wt{Z}^{a, i}, \wt{Z}^{b, i}, \wt{Q}^{i}, \wt{X}^{i},\wt{I})$ the corresponding set of state processes controlled instead by $u^{i}+\eta v^{i}$ instead of $u^i$ (where $u^{-i}$ is as before). From \eqref{def:Y} and  \eqref{def:Z^i} it follows that for any scalar $\eta>0$ we have, 
    \begin{equation}\label{eq:varied_Y_Z}
    \begin{split}
        &\wt{I}_{t} = I_{t} + \eta \lambda \int_{0}^{t}e^{-\beta(t-s)}h_{s}ds, \\
        &\wt{Z}^{a, i}_{t} = Z^{a, i}_{t} - \eta \kappa \frac{N-1}{N^{2}}\int_{0}^{t} \gamma^{a}_{s}dZ^{a}_{s}, \\
        &\wt{Z}^{b, i}_{t} = Z^{b, i}_{t} - \eta \kappa \frac{N-1}{N^{2}}\int_{0}^{t} \gamma^{b}_{s}dZ^{b}_{s}. 
    \end{split}
    \end{equation}
    Together with \eqref{def:X} we get
    \begin{equation} \label{eq:varied_X}
        \wt{X}^{i}_{t} = X^{i}_{t} + \eta \Bigg(\kappa \frac{N-1}{N^{2}}\Bigg(\int_{0}^{t}\gamma^{a}_{s}dZ^{a}_{s} - \int_{0}^{t}\gamma^{b}_{s}dZ^{b}_{s}\Bigg) + \int_{0}^{t} h_{s}ds\Bigg).
    \end{equation}
    From \eqref{def:objective}, \eqref{eq:varied_Y_Z} and \eqref{eq:varied_X} it follows that
        \begin{equation}
    \begin{split}
        &\altcal{J}(u^{i}+\eta v, u^{-i}) = \\ &\mathbb{E}\Bigg[ 
        \int_{0}^{T}\left(P^{\ast}_{t}+I_{t} + \eta \lambda \int_{0}^{t}e^{-\beta(t-s)}h_{s}ds+\delta^{a, i}_{t}+\eta\gamma^{a}_{t}\right)\left(dZ^{a, i}_{t}-\eta\kappa\frac{N-1}{N^{2}}\gamma^{a}_{t}dZ^{a}_{t}\right) \\
        &\quad-\int_{0}^{T}\left(P^{\ast}_{t} + I_{t} + \eta \lambda \int_{0}^{t}e^{-\beta(t-s)}h_{s}ds - \delta^{b, i}_{t}-\eta\gamma^{b}_{t}\right)\left(dZ^{b, i}_{t}-\eta\kappa\frac{N-1}{N^{2}}\gamma^{b}_{t}dZ^{b}_{t}\right) \\
        &\quad-\int_{0}^{T}\left(P^{\ast}_{t}+I_{t} + \eta \lambda  \int_{0}^{t}e^{-\beta(t-s)}h_{s}ds + \frac{1}{2}\epsilon^{i}q^{i}_{t} + \frac{1}{2}\epsilon^{i}\eta h_{t}\right)\left(q^{i}_{t} + \eta h_{t}\right)dt \\
        &\quad-\phi^{i}\int_{0}^{T}\Bigg((X^{i}_{t})^{2}+2\eta X^{i}_{t}\left(\kappa\frac{N-1}{N^{2}}\left(\int_{0}^{t}\gamma^{a}_{s}dZ^{a}_{s}-\int_{0}^{t}\gamma^{b}_{s}dZ^{b}_{s}\right)+\int_{0}^{t}h_{s}ds\right) \\
        &\qquad\qquad+ \eta^{2}\left(\kappa\frac{N-1}{N^{2}}\left(\int_{0}^{t}\gamma^{a}_{s}dZ^{a}_{s}-\int_{0}^{t}\gamma^{b}_{s}dZ^{b}_{s}\right)+\int_{0}^{t}h_{s}ds\right)^{2}\Bigg)dt \\
        &\quad-\alpha^{i}\Bigg((X^{i}_{T})^{2} + 2\eta X^{i}_{T}\left(\kappa\frac{N-1}{N^{2}}\left(\int_{0}^{T}\gamma^{a}dZ^{a}_{t}-\int_{0}^{T}\gamma^{b}_{t}dZ^{b}_{t}\right)+\int_{0}^{T}h_{t}dt\right) \\
        &\qquad\qquad+\eta^{2}\left(\kappa\frac{N-1}{N^{2}}\left(\int_{0}^{T}\gamma^{a}dZ^{a}_{t} - \int_{0}^{T}\gamma^{b}_{t}dZ^{b}_{t}\right) + \int_{0}^{T}h_{t}dt\right)^{2}\Bigg) \\
        &\qquad\qquad+P^{\ast}_{t}\left(X^{i}_{T} + \eta\left(\kappa\frac{N-1}{N^{2}}\left(\int_{0}^{T}\gamma^{a}_{t}dZ^{a}_{t} - \int_{0}^{T}\gamma^{b}dZ^{b}_{t}\right) + \int_{0}^{T}h_{t}dt\right)\right)
        \Bigg].
     \end{split}
    \end{equation}
Note that we can rewrite the right-hand side of \eqref{eq:varied_Y_Z} as follows, 
    \begin{equation} \label{comp1} 
    \begin{split}
        &\altcal{J}(u^{i}+\eta v, u^{-i}) = \\ &\mathbb{E}\Bigg[ 
        \int_{0}^{T}\left(P^{\ast}_{t}+I_{t}+\delta^{a, i}_{t}\right)dZ^{a, i}_{t} - \int_{0}^{T}\left(P^{\ast}_{t}+I_{t}-\delta^{b, i}_{t}\right)dZ^{b, i}_{t} \\
        &\quad+\eta\Bigg(
        \int_{0}^{T}\left(\lambda\int_{0}^{t}e^{-\beta(t-s)}h_{s}ds + \gamma^{a}_{t}\right)dZ^{a, i}_{t} 
        - \kappa\frac{N-1}{N^{2}}\int_{0}^{T}\left(P^{\ast}_{t}+I_{t}+\delta^{a, i}_{t}\right)\gamma^{a}_{t}dZ^{a}_{t} \\
        &\qquad\quad- \int_{0}^{T}\left(\lambda\int_{0}^{t}e^{-\beta(t-s)}h_{s}ds-\gamma^{b}_{t}\right)dZ^{b, i}_{t}
        + \kappa\frac{N-1}{N^{2}}\int_{0}^{T}\left(P^{\ast}_{t}+I_{t}-\delta^{b, i}_{t}\right)\gamma^{b}_{t}dZ^{b}_{t}
        \Bigg)\\
        &\quad+\eta^{2}\Bigg(\kappa^{2}\left(\frac{N-1}{N^{2}}\right)^{2}\int_{0}^{T}\left(\lambda\int_{0}^{t}e^{-\beta(t-s)}h_{s}ds+\gamma^{a}_{t}\right)\gamma^{a}_{t}dZ^{a}_{t} \\
        &\qquad\quad- \kappa^{2}\left(\frac{N-1}{N^{2}}\right)^{2}\int_{0}^{T}\left(\lambda\int_{0}^{t}e^{-\beta(t-s)}h_{s}ds-\gamma^{b}_{t}\right)\gamma^{b}_{t}dZ^{b}_{t}\Bigg) \\
        &\quad-\int_{0}^{T}\left(P^{\ast}_{t}+I_{t}+\frac{1}{2}\epsilon q^{i}_{t}\right)q^{i}_{t}dt \\
        &\quad-\eta\Bigg(\int_{0}^{T}\left(\lambda\int_{0}^{t}e^{-\beta(t-s)}h_{s}ds + \frac{1}{2}\epsilon^{i}h_{t}\right)q^{i}_{t}dt + \int_{0}^{T}\left(P^{\ast}_{t}+I_{t}+\frac{1}{2}\epsilon^{i}q^{i}_{t}\right)h_{t}dt\Bigg) \\
        &\quad-\eta^{2}\int_{0}^{T}\left(\lambda\int_{0}^{t}e^{-\beta(t-s)}h_{s}ds+\frac{1}{2}\epsilon^{i}h_{t}\right)h_{t}dt \\
        &\quad-\phi^{i}\int_{0}^{T}(X^{i}_{t})^{2}dt-\alpha^{i}(X^{i}_{T})^{2} \\
        &-\eta\Bigg(2\phi^{i}\int_{0}^{T}X^{i}_{t}\left(\kappa\frac{N-1}{N^{2}}\left(\int_{0}^{t}\gamma^{a}_{s}dZ^{a}_{s}-\gamma^{b}_{s}dZ^{b}_{s}\right) + \int_{0}^{t}h_{s}ds\right) \\
        &\qquad\quad+2\alpha^{i}X^{i}_{T}\left(\kappa\frac{N-1}{N^{2}}\left(\int_{0}^{T}\gamma^{a}_{t}dZ^{a}_{t}-\int_{0}^{T}\gamma^{b}_{t}dZ^{b}_{t}\right) + \int_{0}^{T}h_{t}dt\right)\Bigg) \\
        &\quad+\eta^{2}\Bigg(\phi^{i}\int_{0}^{T}\left(\kappa\frac{N-1}{N^{2}}\left(\int_{0}^{t}\gamma^{a}_{s}dZ^{a}_{t}-\int_{0}^{t}\gamma^{b}_{s}dZ^{b}_{s}\right) + \int_{0}^{t}h_{s}\right)^{2}dt \\
        &\qquad\quad+2\alpha^{i}\left(\kappa\frac{N-1}{N^{2}}\left(\int_{0}^{T}\gamma^{a}_{t}dZ^{a}_{t}-\int_{0}^{T}\gamma^{b}_{t}dZ^{b}_{t}\right) + \int_{0}^{T}h_{t}dt\right)^{2} \Bigg) \\
        &\quad+P^{\ast}_{t}X^{i}_{T} + \eta P^{\ast}_{t}X^{i}_{T}\left(\kappa\frac{N-1}{N^{2}}\left(\int_{0}^{T}\gamma^{a}_{t}dZ^{a}_{t} - \int_{0}^{T}\gamma^{b}_{t}dZ^{b}_{t}\right) + \int_{0}^{T}h_{t}dt\right)
        \Bigg].
    \end{split}
    \end{equation}
    By subtracting $\altcal{J}(u^{i}, u^{-i})$ in \eqref{def:objective} from \eqref{comp1}, using \eqref{def:Z^i}, and then dividing by $\eta$, and taking the limit as $\eta\rightarrow0$, while using similar arguments as in the proof of Lemma 5.2 in \cite{neuman2022optimal} in order to bound higher order terms, we get,  
    \begin{equation} 
    \begin{split}
        &\langle(\altcal{J}^{i})'(u^{i}, u^{-i})), v\rangle \\
        &= \mathbb{E}\Bigg[
        \frac{1}{N}\int_{0}^{T} \left(\lambda \int_{0}^{t}e^{-\beta(t-s)}h_{s}ds + \gamma^{a}_{t}\right)\left(1-\kappa\left(\delta^{a, i}_{t}-\overline{\delta^{a}}_{t}\right)\right)dZ^{a}_{t} \\
        &\qquad + \int_{0}^{T}\left(\lambda \int_{0}^{t}e^{-\beta(t-s)}h_{s}ds + \gamma^{a}_{t}\right)dJ^{a, i}_{t} - \frac{\kappa(N-1)}{N^{2}}\int_{0}^{T}\gamma^{a}_{t}\left(P^{\ast}_{t}+I_{t}+\delta^{a, i}_{t}\right)dZ^{a}_{t}  \\
        &\qquad-\frac{1}{N}\int_{0}^{T}\left(\lambda \int_{0}^{t}e^{-\beta(t-s)}h_{s}ds - \gamma^{b}_{t}\right)\left(1-\kappa\left(\delta^{b, i}_{t}-\overline{\delta^{b}}_{t}\right)\right)dZ^{b}_{t} \\
        &\qquad - \int_{0}^{T}\left(\lambda \int_{0}^{t}e^{-\beta(t-s)}h_{s}ds - \gamma^{b}_{t}\right)dJ^{b, i}_{t} + \frac{\kappa(N-1)}{N^{2}}\int_{0}^{T}\gamma^{b}_{t}\left(P^{\ast}_{t}+I_{t}-\delta^{b, i}_{t}\right)dZ^{b}_{t} \\
        &\qquad-\int_{0}^{T}\left(P^{\ast}_{t}h_{t} + \lambda  q^{i}_{t}\int_{0}^{t}e^{-\beta(t-s)}h_{s}ds + I_{t}h_{t} + \epsilon^{i} q^{i}_{t}h_{t}\right) dt \\
        &\qquad-2\phi^{i}\int_{0}^{T} X_{t}^{i}\left(\frac{\kappa(N-1)}{N^{2}}\left(\int_{0}^{t}\gamma^{a}_{s}dZ^{a}_{t}-\int_{0}^{t}\gamma^{b}_{s}dZ^{b}_{t}\right) + \int_{0}^{t}h_{s}ds\right)dt \\
        &\qquad+\left(P^{\ast}_{t}-2\alpha^{i} X_{T}^{i}\right)\left(\frac{\kappa (N-1)}{N^{2}}\left(\int_{0}^{T} \gamma^{a}_{t}dZ^{a}_{t}-
        \int_{0}^{T}\gamma^{b}_{t}dZ^{b}_{t}\right) + \int_{0}^{T}h_{t}ds\right)
        \Bigg].
    \end{split}
    \end{equation}
    By compensating the first six integrals, then using linearity of expectations and the martingale property, we get
        \begin{equation} \label{eq:gateaux_pre_fubini}
    \begin{split}
        &\langle(\altcal{J}^{i})'(u^{i}, u^{-i})), v\rangle \\
        &= \mathbb{E}\Bigg[
        \frac{\rho^{a}}{N}\int_{0}^{T} \left(\lambda \int_{0}^{t}e^{-\beta(t-s)}h_{s}ds + \gamma^{a}_{t}\right)\left(1-\kappa\left(\delta^{a, i}_{t}-\overline{\delta^{a}}_{t}\right)\right)dt \\
        &\qquad + \mu\int_{0}^{T}\left(\lambda \int_{0}^{t}e^{-\beta(t-s)}h_{s}ds + \gamma^{a}_{t}\right)dt - \frac{\rho^{a}\kappa(N-1)}{N^{2}}\int_{0}^{T}\gamma^{a}_{t}\left(P^{\ast}_{t}+I_{t}+\delta^{a, i}_{t}\right)dt  \\
        &\qquad-\frac{\rho^{b}}{N}\int_{0}^{T}\left(\lambda \int_{0}^{t}e^{-\beta(t-s)}h_{s}ds - \gamma^{b}_{t}\right)\left(1-\kappa\left(\delta^{b, i}_{t}-\overline{\delta^{b}}_{t}\right)\right)dt\\
        &\qquad- \mu\int_{0}^{T}\left(\lambda \int_{0}^{t}e^{-\beta(t-s)}h_{s}ds - \gamma^{b}_{t}\right)dt + \frac{\rho^{b}\kappa(N-1)}{N^{2}}\int_{0}^{T}\gamma^{b}_{t}\left(P^{\ast}_{t}+I_{t}-\delta^{b, i}_{t}\right)dt \\
        &\qquad-\int_{0}^{T}\left(P^{\ast}_{t}h_{t} + \lambda  q^{i}_{t}\int_{0}^{t}e^{-\beta(t-s)}h_{s}ds + I_{t}h_{t} + \epsilon^{i} q^{i}_{t}h_{t}\right) dt \\
        &\qquad-2\phi^{i}\int_{0}^{T} X_{t}^{i}\left(\frac{\kappa(N-1)}{N^{2}}\left(\int_{0}^{t}\gamma^{a}_{s}dZ^{a}_{t}-\int_{0}^{t}\gamma^{b}_{s}dZ^{b}_{t}\right) + \int_{0}^{t}h_{s}ds\right)dt \\
        &\qquad+\left(P^{\ast}_{t}-2\alpha^{i} X_{T}^{i}\right)\left(\frac{\kappa (N-1)}{N^{2}}\left(\int_{0}^{T} \gamma^{a}_{t}dZ^{a}_{t}-
        \int_{0}^{T}\gamma^{b}_{t}dZ^{b}_{t}\right) + \int_{0}^{T}h_{t}ds\right)
        \Bigg].
    \end{split}
    \end{equation}
    By applying Fubini's theorem to \eqref{eq:gateaux_pre_fubini}, we obtain \eqref{eq:gateaux}. 
\end{proof}

Before we state Proposition \ref{prop:FBSDE}, which characterises the optimal controls as the solution to a forward-backward-stochastic-differential-equation, we introduce the following square-integrable martingales
A maximising $(\delta^{a, i}, \delta^{b, i}, q^{i})$ to \eqref{def:objective}, must therefore satisfy \eqref{eq:delta_a_first_order}, \eqref{eq:delta_b_first_order} and \eqref{eq:q_first_order}.
Define the square-integrable martingale
\begin{equation} \label{def:tildeM}
    M^{i}_{t} := \mathbb{E}\left[P^{\ast}_{t}-2\alpha^{i} X^{i}_{T}-2\phi^{i}\int_{0}^{T}X^{i}_{r}dr\Big\vert\mathcal{F}_{t}\right], \quad 0\leq t \leq T, 
\end{equation}
and recall the definitions of $\Gamma^{i}$ and $\wt{N}^i$ in \eqref{def:Gamma} and  \eqref{def:tildeN}, respectively. 
We further define the square-integrable martingales
\begin{equation} \label{N-L-mart} 
    N^{i}_{t} := \lambda\int_{0}^{t}e^{\beta s}d\wt{N}^{i}_{s}, \quad L^{i}_{t} := -\int_{0}^{t}e^{\beta s}d\wt{N}^{i}_{s} + M^{i}_{t},    \quad 0\leq t \leq T.
\end{equation}
In the following proposition we characterize the Nash equilibrium of the game as a system of FBSDEs. 
\begin{proposition} \label{prop:FBSDE}
    Assume that $Z^{a}$ and $Z^{b}$ are Poisson processes with intensities $\rho^a$ and $\rho^b$, respectively, initiating at zero. A set of admissible controls $(u^{i})_{i\in [N]}$ , where $u^{i}=(\delta^{a, i}, \delta^{b, i}, q^{i})$, is a Nash equilibrium in the sense of Definition \ref{def:nash}, if and only if $(X^{i}, \delta^{a, i}, \delta^{b, i}, q^{i}, \Gamma^{i})_{i\in [N]}$ and $Y$ satisfy the following coupled system of FBSDEs:      \begin{equation} \label{eq:FBSDE_aggregated}
    \begin{dcases}
        &dX^{i}_{t} = \frac{1}{N}\Big(1-\kappa\Big(\delta^{b, i}_{t} -\overline{\delta^{b}}_{t}\Big)\Big)dZ^{b}_{t} - \frac{1}{N}\Big(1-\kappa\Big(\delta^{a, i}_{t} - \overline{\delta^{a}}_{t}\Big)\Big)dZ^{a}_{t} + q^{i}_{t}dt, \\
                &dI_{t} = -\beta I_{t}dt + \lambda N \overline{q}_{t}dt, \\
        &d\delta^{a, i}_{t} = -dP^{\ast}_{t} + \beta I_{t}dt - \lambda N \overline{q}_{t}dt + 2\phi^{i}X^{i}_{t}dt + \frac{2N}{2N-1}\overline{\phi X}_{t}dt + \\
        &\qquad\qquad dM^{i}_{t} + \frac{N}{2N-1}d\overline{M_{t}^{-i}}, \\
        &d\delta^{b, i}_{t} = dP^{\ast}_{t} - \beta I_{t}dt + \lambda N \overline{q}_{t}dt - 2\phi^{i}X^{i}_{t}dt - \frac{2N}{2N-1}\overline{\phi X}_{t}dt \\
        &\qquad\qquad - dM^{i}_{t} - \frac{N}{2N-1}d\overline{M^{-i}_t}, \\
        &dq^{i}_{t} = \frac{1}{\epsilon^{i}}\Bigg(-dP^{\ast}_{t} + \beta I_{t}dt - \lambda N \overline{q}_{t}dt + 2\phi^{i}X^{i}_{t}dt + \beta \Gamma^{i}_{t}dt +\lambda q^{i}_{t}dt + dL^{i}_{t} \\
        &\qquad\qquad- \frac{\lambda}{N}\left(1-\kappa\left(\delta^{a, i}_{t}-\overline{\delta^{a}}_{t}\right)\right)dZ^{a}_{t} + \frac{\lambda}{N}\left(1-\kappa\left(\delta^{b, i}_{t}-\overline{\delta^{b}}_{t}\right)\right)dZ^{b}_{t} \Bigg), \\
        &d\Gamma^{i}_{t} = \beta\Gamma^{i}_{t}dt -\frac{\lambda}{N}\left(1-\kappa\left(\delta^{a, i}_{t} - \overline{\delta^{a}}_{t}\right)\right)dZ^{a}_{t} + \frac{\lambda}{N}\left(1-\kappa\left(\delta^{b, i}_{t} - \overline{\delta^{b}}_{t}\right)\right)dZ^{b}_{t} \\
        &\qquad\qquad + \lambda q^{i}_{t}dt + dN^{i}_{t}, \\
    \end{dcases}
    \end{equation}
    with the  initial and terminal conditions: 
    \be
 \begin{aligned} 
  X^{i}_{0}=x^{i}_{0},  \quad I_{0} = 0, \quad  \Gamma^{i}_{T} = 0, \quad  q^{i}_{T} = -\frac{1}{\epsilon^{i}}\left(I_{T}+2\alpha^{i} X^{i}_{T}\right), 
\end{aligned} 
\ee
\be
\delta^{a, i}_{T} = \frac{N}{\kappa(2N-1)} +\frac{N\mu}{\rho^{a}}+ \frac{N}{2N-1}\overline{\delta^{a}_{T}} - \frac{N-1}{2N-1}I_{t} - \frac{2N-2}{2N-1}\alpha^{i} X^{i}_{T}, \\ 
\ee
and 
\be
  \delta^{b, i}_{T} = \frac{N}{\kappa(2N-1)}+\frac{N\mu}{\rho^{b}} + \frac{N}{2N-1}\overline{\delta^{b}_{T}} + \frac{N-1}{2N-1}I_{t} + \frac{2N-2}{2N-1}\alpha^{i} X^{i}_{T}. 
\ee
\end{proposition}

\begin{proof}
From Proposition \ref{prop:uniqueness}, Lemma \ref{lemma:gateaux}, and Proposition 2.1 in Chapter 2 of \cite{ekeland1999convex}, we derive a first-order condition for maximiser $u^{i}=(\delta^{a, i}, \delta^{b, i}, q^{i})$ for the objective \eqref{def:objective}, for every $i\in [N]$.

\textit{Necessity:} We first prove that a profile $(u_i)_{i\in N}$ which satisfies \eqref{def:Gateaux}, also satisfies \eqref{eq:FBSDE_aggregated}. 
By applying optional projection on the right-hand side of \eqref{eq:gateaux} and using \eqref{def:Gateaux} we get,   \begin{equation} \label{eq:opt_proj}
\begin{split}
    \mathbb{E}\Bigg[&
    \int_{0}^{T}\gamma^{a}_{s}\Bigg(\rho^{a}\bigg(\frac{1}{N}\left(1-\kappa(\delta^{a, i}_{s}-\overline{\delta^{a}}_{s})\right) - \kappa\frac{N-1}{N^{2}}\left(P^{\ast}_{s}+I_{s}+\delta^{a, i}_{s}\right) \\ 
    &\qquad\quad+\kappa\frac{N-1}{N^{2}}\mathbb{E}\bigg[\left(P^{\ast}_{t}-2\alpha^{i}X^{i}_{T}\right) -2\phi^{i}\int_{s}^{T}X^{i}_{t}dt\bigg\vert\mathcal{F}_{s}\bigg]\bigg)+ \mu\Bigg)ds \\
    &+\int_{0}^{T}\gamma^{b}_s\Bigg(\rho^{b}\bigg(\frac{1}{N}\left(1-\kappa(\delta^{b,i }_{s}-\overline{\delta^{b}}_{s})\right) + \kappa\frac{N-1}{N^{2}}\left(P^{\ast}_{s} + I_{s} - \delta^{b, i}_{s}\right) \\
    &\qquad\quad-\kappa\frac{N-1}{N^{2}}\mathbb{E}\bigg[\left(P^{\ast}_{t}-2\alpha^{i}X^{i}_{T}\right) -2\phi^{i}\int_{s}^{T}X^{i}_{t}dt\bigg\vert\mathcal{F}_{s}\bigg]\bigg) + \mu\Bigg)ds \\
    &+\int_{0}^{T}h_{s}\bigg(\frac{\lambda}{N}\mathbb{E}\bigg[\int_{s}^{T}e^{-\beta(t-s)}\left(1-\kappa(\delta^{a, i}_{t}-\overline{\delta^{a}}_{t})\right)dZ^{a}_{t}\bigg\vert\mathcal{F}_{s}\bigg] \\
    &\qquad\quad-\frac{\lambda}{N}\mathbb{E}\bigg[\int_{s}^{T}e^{-\beta(t-s)}\left(1-\kappa(\delta^{b, i}_{t}-\overline{\delta^{b}}_{t})\right)dZ^{b}_{t}\bigg\vert\mathcal{F}_{s}\bigg] \\
    &\qquad\quad-\lambda \mathbb{E}\bigg[\int_{s}^{T}e^{-\beta(t-s)}q^{i}_{t}dt\bigg\vert\mathcal{F}_{s}\bigg] - P^{\ast}_{s} - I_{s} - \epsilon^{i} q^{i}_{s} \\
    &\qquad\quad+\mathbb{E}\bigg[\left(P^{\ast}_{t}-2\alpha^{i}X^{i}_{T}\right) -2\phi^{i}\int_{s}^{T}X^{i}_{t}dt\bigg\vert\mathcal{F}_{s}\bigg]\bigg)ds\Bigg] = 0,
\end{split}
\end{equation}
for any admissible $\gamma^{a}, \gamma^{b}, h$.  It follows that the conditions given below are necessary for \eqref{eq:opt_proj} to hold: 
\begin{equation} \label{eq:delta_a_first_order}
    1+\frac{N\mu}{\rho^{a}} - \kappa \delta^{a, i}_{s} + \kappa\overline{\delta^{a}_{s}} - \frac{\kappa(N-1)}{N}\left(P^{\ast}_{s}+I_{s}+\delta^{a, i}_{s} - \mathbb{E}\left[P^{\ast}_{t}-2\alpha^{i} X^{i}_{T}-2\phi^{i}\int_{s}^{T}X^{i}_{t}dt\Big\vert\mathcal{F}_{s}\right]\right) = 0 
\end{equation}
and 
\begin{equation}\label{eq:delta_b_first_order}
    1+\frac{N\mu}{\rho^{b}} - \kappa \delta^{b, i}_{s} + \kappa\overline{\delta^{b}_{s}} + \frac{\kappa(N-1)}{N}\left(P^{\ast}_{s}+I_{s}-\delta^{b, i}_{s} - \mathbb{E}\left[P^{\ast}_{t}-2\alpha^{i} X^{i}_{T}-2\phi^{i}\int_{s}^{T}X^{i}_{t}dt\Big\vert\mathcal{F}_{s}\right]\right) = 0 
\end{equation}
and
\begin{equation}\label{eq:q_first_order}
\begin{split}
    &\frac{\lambda}{N}\mathbb{E}\bigg[\int_{s}^{T}e^{-\beta(t-s)}\left(1-\kappa(\delta^{a, i}_{t}-\overline{\delta^{a}}_{t})\right)dZ^{a}_{t}\bigg\vert\mathcal{F}_{s}\bigg] \\
    &-\frac{\lambda}{N}\mathbb{E}\bigg[\int_{s}^{T}e^{-\beta(t-s)}\left(1-\kappa(\delta^{b, i}_{t}-\overline{\delta^{b}}_{t})\right)dZ^{b}_{t}\bigg\vert\mathcal{F}_{s}\bigg] \\
    &-\lambda \mathbb{E}\bigg[\int_{s}^{T}e^{-\beta(t-s)}q^{i}_{t}dt\bigg\vert\mathcal{F}_{s}\bigg] - P^{\ast}_{s} - I_{s} - \epsilon^{i} q^{i}_{s} +\mathbb{E}\bigg[\left(P^{\ast}_{t}-2\alpha^{i}X^{i}_{T}\right) -2\phi^{i}\int_{s}^{T}X^{i}_{t}dt\bigg\vert\mathcal{F}_{s}\bigg] = 0.
\end{split}
\end{equation}
That is, the maximiser $(\delta^{a, i}, \delta^{b, i}, q^{i})$ to \eqref{def:objective}, must satisfy \eqref{eq:delta_a_first_order}--\eqref{eq:q_first_order}.
From \eqref{def:tildeM} and \eqref{eq:delta_a_first_order} we get that,  
\begin{equation} \label{eq:delta_a_pre_dynamics}
    d\delta^{a, i}_{s} = \frac{N}{2N-1}d\overline{\delta^{a}}_{s} - \frac{N-1}{2N-1}\bigg(dP^{\ast}_{s} -\beta I_{s}ds + \lambda N \overline{q}_{s}ds - dM^{i}_{s} - 2\phi^{i}X^{i}_{s}ds\bigg).
\end{equation}
From \eqref{def:tildeM} and \eqref{eq:delta_b_first_order} it follows that
\begin{equation} \label{eq:delta_b_pre_dynamics}
    d\delta^{b, i}_{s} = \frac{N}{2N-1}d\overline{\delta^{b}}_{s} + \frac{N-1}{2N-1}\bigg(dP^{\ast}_{s} -\beta I_{s} ds + \lambda N \overline{q}_{s} ds - dM^{i}_{s} - 2\phi^{i}X^{i}_{s}ds\bigg).
\end{equation}
Note that from \eqref{def:Gamma} and \eqref{def:Z^i} it follows that $\Gamma^{i}$ satisfies, 
\begin{equation} \label{eq:Gamma_dynamcs}
\begin{split}
    d\Gamma^{i}_{s} =& \beta\Gamma^{i}_{s}ds -\frac{\lambda}{N}\left(1-\kappa\left(\delta^{a, i}_{s} - \overline{\delta^{a}}_{s}\right)\right)dZ^{a}_{s} \\
    &+ \frac{\lambda}{N}\left(1-\kappa\left(\delta^{b, i}_{s} - \overline{\delta^{b}}_{s}\right)\right)dZ^{b}_{s} +\lambda q^{i}_{s}ds + \lambda e^{\beta s}d\wt{N}^{i}_{s}.
\end{split}
\end{equation}
Finally from \eqref{def:Gamma} and \eqref{eq:q_first_order} we have, 
\begin{equation} \label{eq:optimal_q}
    q^{i}_{s} = \frac{1}{\epsilon^{i}}\bigg(-P^{\ast}_{s} - I_{s} + M^{i}_{s} + 2\phi^{i}\int_{0}^{s}X^{i}_{t}dt + \Gamma^{i}_{s} \bigg). 
\end{equation}
Together with \eqref{eq:Gamma_dynamcs} and \eqref{def:Y} it follows that the dynamics of $q^{i}$ is given by
\begin{equation} \label{eq:q_dynamics} 
\begin{split}
    dq^{i}_{s} = \frac{1}{\epsilon^{i}}\bigg(
    &-dP^{\ast}_{s} + \beta I_{s}ds - \lambda N\overline{q}_{s}ds + dM^{i}_{s} + 2\phi^{i}X^{i}_{s}ds + \beta \Gamma^{i}_{s}ds \\
    &- \frac{\lambda}{N}\left(1-\kappa\left(\delta^{a, i}_{s}-\overline{\delta^{a}}_{s}\right)\right)dZ^{a}_{s} + \frac{\lambda}{N}\left(1-\kappa\left(\delta^{b, i}_{s}-\overline{\delta^{b}}_{s}\right)\right)dZ^{b}_{s} \\
    &+\lambda q^{i}_{s}ds - \lambda e^{-\beta s}d\wt{N}^{i}_{s} \bigg).
\end{split}
\end{equation}

Summing over the index $i$ on both sides of \eqref{eq:delta_a_pre_dynamics} and \eqref{eq:delta_b_pre_dynamics}, dividing by $N$, and using the notation in \eqref{av-not} we get, 
\begin{equation} \label{eq:delta_a_bar_dynamics}
     d\overline{\delta^{a}_s}= -dP^{\ast}_{s} + \beta I_{s} ds - \lambda N \overline{q}_{s}ds + d\overline{M}_{s} + 2\overline{\phi X}_{s} ds,
 \end{equation}
and
\begin{equation} \label{eq:delta_b_bar_dynamics}
    d\overline{\delta^{b}_s}= dP^{\ast}_{s} - \beta I_{s} ds + \lambda N \overline{q}_{s}ds - d\overline{M}_{s} - 2\overline{\phi X}_{s} ds.
\end{equation}
 Using \eqref{eq:delta_a_pre_dynamics} and \eqref{eq:delta_a_bar_dynamics}, and the notation introduced in  \eqref{av-not} for $\overline{M^{-i}}$ below, we get that  \begin{equation} \label{eq:delta_a_dynamics}
    d\delta^{a, i}_{s} = -dP^{\ast}_{s} + \beta I_{s}ds - \lambda N \overline{q}_{s}ds + 2\phi^{i}X^{i}_{s}ds + \frac{2N}{2N-1}\overline{\phi X}_{s}ds + dM^{i}_{s} + \frac{N}{2N-1}d\overline{M^{-i}_s}
\end{equation}
and using \eqref{eq:delta_b_pre_dynamics} and \eqref{eq:delta_b_bar_dynamics}, we get that
\begin{equation} \label{eq:delta_b_dynamics}
    d\delta^{b, i}_{s} = dP^{\ast}_{s} - \beta I_{s}ds + \lambda N \overline{q}_{s}ds - 2\phi^{i}X^{i}_{s}ds - \frac{2N}{2N-1}\overline{\phi X}_{s}ds - dM^{i}_{s} - \frac{N}{2N-1}d\overline{M^{-i}_s}.
\end{equation}
Using \eqref{eq:Gamma_dynamcs} and \eqref{av-not} we get,  
\begin{equation} \label{eq:Gamma_bar_dynamics}
    d\overline{\Gamma}_{s} = \beta\overline{\Gamma}_{s} ds + \lambda \overline{q}_{s}ds + \lambda \left(dZ^{b}_{s} - dZ^{a}_{s}\right) + d\overline{N}_{s}.
\end{equation}
Furthermore, using \eqref{eq:optimal_q} and \eqref{av-not} we note that 
\begin{equation} \label{eq:q_bar}
    \overline{q}_{s} = -\overline{\frac{1}{\epsilon}}P^{\ast}_{s} - \overline{\frac{1}{\epsilon}}I_{s} + \overline{\left(\frac{M}{\epsilon}\right)}_{s} + 2\int_{0}^{s}\overline{\left(\frac{\phi X}{\epsilon}\right)}_{t}dt + \overline{\left(\frac{\Gamma}{\epsilon}\right)}_{s}.
\end{equation}
From \eqref{N-L-mart}, \eqref{eq:q_dynamics} and \eqref{av-not} it follows that the dynamics of $\overline{q}$ is given by  
\begin{equation} \label{eq:q_bar_dynamics}
\begin{split}
    d\overline{q}_{s} =& -\overline{\frac{1}{\epsilon}}dP^{\ast}_{s} + \overline{\frac{1}{\epsilon}}\beta I_{s}ds - \overline{\frac{1}{\epsilon}}\lambda N \overline{q}_{s}ds + 2\overline{\left(\frac{\phi X}{\epsilon}\right)}_{s}ds + \beta\overline{\left(\frac{\Gamma}{\epsilon}\right)}_{s}ds \\
    &+ \frac{\kappa\lambda}{N^{2}}\overline{\left(\frac{\delta^{a}}{\epsilon}\right)}_{s}dZ^{a}_{s} - \frac{\kappa\lambda}{N^{2}}\overline{\frac{1}{\epsilon}}\overline{\delta^{a}}_{s}dZ^{a}_{s} - \frac{\kappa\lambda}{N^{2}}\overline{\left(\frac{\delta^{b}}{\epsilon}\right)}_{s}dZ^{b}_{s} + \frac{\kappa\lambda}{N^{2}}\overline{\frac{1}{\epsilon}}\overline{\delta^{b}}_{s}dZ^{b}_{s} \\
    &+ \overline{\frac{1}{\epsilon}}\frac{\lambda}{N}\left(dZ^{b}_{s}-dZ^{a}_{s}\right) + \lambda \overline{\left(\frac{q}{\epsilon}\right)}_{s}ds + d\overline{\left(\frac{L}{\epsilon}\right)}_{s}.
\end{split}
\end{equation}
Finally, it is left to verify the terminal conditions in \eqref{eq:FBSDE_aggregated}. The terminal condition $\Gamma^{i}_{T}=0$ follows from the definitions of $\Gamma^{i}$ and $\tilde{N}^{i}$, given in \eqref{def:Gamma} and \eqref{def:tildeN} respectively. The terminal condition for $q^{i}$ follows from \eqref{eq:optimal_q} and the terminal condition of $\Gamma^{i}$, where we recall the definition of $\tilde{M}^{i}$ from \eqref{def:tildeM}.  Note that the unaffected price term and the running inventory penalty terms in \eqref{eq:optimal_q} cancel due to \eqref{def:tildeM}. Finally, the terminal conditions of $\delta^{a, i}$ and $\delta^{b, i}$ follow from \eqref{eq:delta_a_first_order} and \eqref{eq:delta_b_first_order} respectively.

By combining \eqref{def:Y}, \eqref{def:Z^i}, \eqref{def:X}, \eqref{eq:delta_a_bar_dynamics}--\eqref{eq:q_bar_dynamics} we get \eqref{eq:FBSDE_aggregated}.  
\newline
 \textit{Sufficiency:} Let $(u^i,X^{i}, \delta^{a, i}, \delta^{b, i}, q^{i}, \Gamma^{i})_{i\in [N]}$ and $Y$ be a solution to \eqref{eq:FBSDE_aggregated}. We will prove that \eqref{def:Gateaux} is satisfied. 
Reversing the steps leading to \eqref{eq:delta_a_dynamics} and \eqref{eq:delta_b_dynamics}, we get \eqref{eq:delta_a_first_order} and \eqref{eq:delta_b_first_order}. Substituting \eqref{eq:delta_a_first_order} and \eqref{def:tildeM} into the first expected integral on the left-hand side of \eqref{eq:opt_proj}, we get 
\begin{equation} \label{eq:ask_sufficiency_1}
\begin{aligned}
    &\mathbb{E}\Bigg[\int_{0}^{T}\gamma^{a}_{s}\rho^{a}\Bigg(
    \frac{1}{N}\left(1-\kappa\left(\delta^{a, i}_{s}-\overline{\delta^{a}_s}\right)\right)-\kappa\frac{N-1}{N^{2}}\left(P^{\ast}_{s}+I_{s}+\delta^{a, i}_{s}\right) \\
    &\qquad\qquad-2\phi^{i}\kappa\frac{N-1}{N^{2}}\int_{s}^{T}X^{i}_{t}dt + \kappa\frac{N-1}{N^{2}}\left(P^{\ast}_{t}-2\alpha^{i}X^{i}_{T}\right)\Bigg)ds\Bigg] \\
    &=\mathbb{E}\Bigg[\int_{0}^{T}\gamma^{a}_{s}\rho^{a}\Bigg(
    -\kappa\frac{N-1}{N^{2}}M^{i}_{s}+\kappa\frac{N-1}{N^{2}}\left(P^{\ast}_{t}-2\alpha^{i}X^{i}_{T}-2\phi^{i}\int_{0}^{T}X^{i}_{t}dt\right)
    \Bigg)ds\Bigg] \\
    &= \kappa\frac{N-1}{N^{2}}\mathbb{E}\Bigg[\int_{0}^{T}\gamma^{a}_{s}\rho^{a}
   \left(M^{i}_{T}-M^{i}_{s}\right)
    ds\Bigg].
\end{aligned}
\end{equation}
An application of the tower property, Fubini's theorem, and the martingale property of $M^{i}$ then gives us 
\begin{equation} \label{eq:ask_sufficiency_2}
    \mathbb{E}\Bigg[\int_{0}^{T}\gamma^{a}_{s}\rho^{a}\left(M^{i}_{T}-M^{i}_{s}\right) 
   ds\Bigg] =\mathbb{E}\Bigg[\int_{0}^{T}\gamma^{a}_{s}\rho^{a}\Bigg(
  \left(\mathbb{E}\left[M^{i}_{T}|\mathcal{F}_{s}\right]-M^{i}_{s}\right)
    \Bigg)ds\Bigg] = 0.
\end{equation}
Substituting \eqref{eq:delta_b_first_order} and \eqref{def:tildeM} into the second expected integral on the left-hand side of \eqref{eq:opt_proj} and repeating the same steps we get,   
\begin{equation} \label{eq:bid_sufficiency_1}
\begin{split}
    &\mathbb{E}\Bigg[\gamma^{b}_{s}\rho^{b}\Bigg(
    \frac{1}{N}\left(1-\kappa\left(\delta^{b, i}_{s}-\overline{\delta^{b}}_{s}\right)\right) - \kappa\frac{N-1}{N^{2}}\left(P^{\ast}_{s}+I_{s}-\delta^{b, i}_{s}\right) \\
     &\qquad\quad+2\phi^{i}\kappa\frac{N-1}{N^{2}}\int_{s}^{T}X^{i}_{t}dt - \kappa\frac{N-1}{N^{2}}\left(P^{\ast}_{t}-2\alpha^{i}X^{i}_{T}\right)
    \Bigg)ds\Bigg] =0.
 \end{split}
\end{equation}

From \eqref{eq:Gamma_dynamcs} and \eqref{eq:q_dynamics} we get \eqref{eq:optimal_q}. Substituting this into the third expected integral on the left-hand side of \eqref{eq:opt_proj}, we get
\begin{equation} \label{eq:q_sufficiency_1}
\begin{split}
    &\mathbb{E}\Bigg[\int_{0}^{T}h_{s}\Bigg(\lambda\int_{s}^{T}e^{-\beta(t-s)}dZ^{a, i}_{t} - \lambda\int_{s}^{T}e^{-\beta(t-s)}dZ^{b, i}_{t} - \lambda\int_{s}^{T}e^{-\beta(t-s)}q^{i}_{t}dt \\
    &\qquad\qquad- P^{\ast}_{s} - I_{s} - \epsilon^{i}q^{i}_{s} - 2\phi^{i}\int_{s}^{T}X^{i}_{t}dt + \left(P^{\ast}_{t}-2\alpha^{i}X^{i}_{T}\right)\Bigg)ds
    \Bigg] \\
    &=\mathbb{E}\Bigg[\int_{0}^{T}h_{s}\Bigg(\lambda\int_{s}^{T}e^{-\beta(t-s)}dZ^{a, i}_{t} - \lambda\int_{s}^{T}e^{-\beta(t-s)}dZ^{b, i}_{t} - \lambda\int_{s}^{T}e^{-\beta(t-s)}q^{i}dt \\
    &\qquad\qquad-M^{i}_{s}-2\phi^{i}\int_{0}^{s}X^{i}_{t}dt - \Gamma^{i}_{s} - 2\phi^{i}\int_{s}^{T}X^{i}_{t}dt + \left(P^{\ast}_{t}-2\alpha^{i}X^{i}_{T}\right)\Bigg)ds
    \Bigg] \\
     &=\mathbb{E}\Bigg[\int_{0}^{T}h_{s}\Bigg(
    e^{\beta s}\wt{N}^{i}_{T} - e^{\beta s}\wt{N}^{i}_{s} + M^{i}_{T} - M^{i}_{s}
    \Bigg)ds\Bigg] \\
    &=\mathbb{E}\Bigg[\int_{0}^{T}h_{s}\Bigg(
    e^{\beta s}\mathbb{E}\left[\wt{N}^{i}_{T}|\mathcal{F}_{s}\right] - e^{\beta s}\wt{N}^{i}_{s} + \mathbb{E}\left[M^{i}_{T}|\mathcal{F}_{s}\right] - M^{i}_{s}
    \Bigg)ds\Bigg] = 0.
\end{split}
\end{equation}
where we have used \eqref{def:Gamma}, \eqref{def:tildeN} and \eqref{def:tildeM} in the second equality and the tower property, Fubini theorem and the martingale property of $\wt{N}^{i}$ and $M^{i}$ in the last two equalities. 
Putting together \eqref{eq:ask_sufficiency_1}--\eqref{eq:q_sufficiency_1} into \eqref{eq:opt_proj} we verify \eqref{def:Gateaux}. 
\end{proof}

Now we are ready to prove Theorem \ref{thm:solution}. 
\begin{proof}[Proof of Theorem \ref{thm:solution}.] 
From Proposition \ref{prop:uniqueness}, Lemma \ref{lemma:gateaux} and Proposition \ref{prop:FBSDE}, it follows that the unique Nash equilibrium is the solution of \eqref{eq:FBSDE_aggregated}. We will derive a closed form solution to this equilibrium and prove its admissibility in the following steps. 

\textbf{Step 1: Matrix form of \eqref{eq:FBSDE_aggregated}.}
Combining each of the $N$-FBSDE systems in \eqref{eq:FBSDE_aggregated} gives us one $5N+2$ dimensional FBSDE system. To write this in a compact form, we introduce some notation and recall the definitions of $\mathbf{V}$ and $\mathbf{X}$ which were given in \eqref{def:V_vector} and \eqref{def:X_vector}, respectively. 
Using the notation convention in \eqref{av-not}, we additionally define
\begin{equation}
    S^{i}_{t} = 
    \begin{pmatrix}
        -P^{\ast}_{t} + \frac{N}{2N-1}dM^{i}_{t} + \frac{N-1}{2N-1}\overline{M^{-i}}_{t}  \\
        P^{\ast}_{t} - \frac{N}{2N-1}dM^{i}_{t} - \frac{N-1}{2N-1}\overline{M^{-i}}_{t} \\
        \frac{1}{\epsilon}\Big(-P^{\ast}_{t} + L^{i}_{t}\Big) \\
        N^{i}_{t}
    \end{pmatrix} ,
\end{equation}
\begin{equation}
    \check{S} = \bigg(
    -P^{\ast}_{t} + \overline{M}_{t} ,
    P^{\ast}_{t} - \overline{M}_{t} ,
    -\overline{\left(\frac{1}{\epsilon}\right)}P^{\ast}_{t} + \overline{\left(\frac{L}{\epsilon}\right)}_{t} ,
    \overline{N}_{t} ,
    0 
    \bigg)^\top,
\end{equation}
and
\begin{equation}
    \wt{\mathbf{S}}_{t} = 
    (
        S^{1}_{t} , ... , S^{N}_{t} , \check{S}_{t} , \mathbf{0}_{N+1}^{\top}. 
   )^{\top}
\end{equation}
Finally, define 
\begin{equation} \label{def:Y_vector}
    \mathbf{Y}_{t} = \left(
        (\mathbf{V}_{t})^{\top}, (\mathbf{X}_{t})^{\top}
    \right)^{\top}, \quad 0\leq t \leq T. 
\end{equation}
 Then, using the matrices $\wt{\mathbf{K}}, \mathbf{Q}^{a}, \mathbf{Q}^{b}$ which are introduced in Appendix \ref{appendix:matrices2}, we can write the overall $5N+2$ dimensional FBSDE in \eqref{eq:FBSDE_aggregated} as follows,  
\begin{equation} \label{eq:matrix_FBSDE_Z}
    d\mathbf{Y}_{t} = \mathbf{Y}_{t} \big(\wt{\mathbf{K}}dt + \mathbf{Q}^{a} dZ^{a}_{t} + \mathbf{Q}^{b} dZ^{b}_{t}\big) + d\wt{\mathbf{S}}_{t}, \quad 0\leq t\leq T, 
\end{equation}
with $N+1$ initial conditions given by $\mathbf{Y}^{4N+1+i}_{0} = x^{i}$ for $i\in [N]$, $\mathbf{Y}^{5N+2}_{0} = 0$ and $4N+1$ terminal conditions given by  
\begin{equation} \label{eq:TY=c}
    \mathbf{T}\mathbf{Y}_{T} = \mathbf{c} := 
    \begin{pmatrix}
        0 & ... & 0 & 1
    \end{pmatrix}^{\top}\in\mathbb{R}^{4N+1}, 
\end{equation}
which were introduced below \eqref{eq:FBSDE_aggregated}. Recall that the matrix $\mathbf{T}$ is defined in \eqref{def:T}. 
Recall that $Z^{a}$ and $Z^{b}$ are Poisson processes with intensities $\rho^a$ and $\rho^b$, respectively. 
We can therefore rewrite \eqref{eq:matrix_FBSDE_Z} as follows, 
\begin{equation} \label{eq:matrix_FBSDE_rho}
    d\mathbf{Y}_{t} = \mathbf{Y}_{t}(\wt{\mathbf{K}} + \mathbf{Q}^{a}\rho^{a}_{t} + \mathbf{Q}^{b}\rho^{b}_{t})dt + d\mathbf{S}_{t}, \quad 0\leq t \leq T. 
\end{equation}
where 
\begin{equation}
    \mathbf{S}_{t}=\wt{\mathbf{S}}_{t} + \mathbf{Q}^{a} \int_{0}^{t}\mathbf{Y}_{s}(dZ^{a}_{s}-\rho^{a}ds) +\mathbf{Q}^{b} \int_{0}^{t} \mathbf{Y}_{s}(dZ^{b}_{s}-\rho^{b}ds), \quad 0\leq t \leq T. 
\end{equation}

\textbf{Step 2: Closed form solution for \eqref{eq:FBSDE_aggregated}.}
Define
\begin{equation}
    \mathbf{K} = \wt{\mathbf{K}} + \mathbf{Q}^{a}\rho^{a} + \mathbf{Q}^{b}\rho^{b}, \quad   \mathbf{G}(t) = \exp\left(\mathbf{K}t\right), \quad 0\leq t \leq T. 
\end{equation}
Then, the solution to \eqref{eq:matrix_FBSDE_rho} is given by
\begin{equation} \label{eq:matrix sde solution}
    \mathbf{Y}_{T} = \mathbf{G}(T-t)\mathbf{Y}_{t} + \int_{t}^{T}\mathbf{G}(T-s)d\mathbf{S}_{s}.
\end{equation}
Let $\mathbf{H}(t) := \mathbf{T}\mathbf{G}(t)\in\mathbb{R}^{(4N+1)\times(5N+2)}$. 
Then, multiplying both sides of \eqref{eq:matrix sde solution} by $\mathbf{T}$ and using \eqref{eq:TY=c}, we get,
      \begin{equation} \label{eq:matrix sde solution2}
    c = \mathbf{H}(T-t)\mathbf{Y}_{t} + \int_{t}^{T}\mathbf{H}(T-s)d\mathbf{S}_{s}.
\end{equation} 
Denote by $\mathbf{H}^{V}$ the first $4N+1$ columns of $\mathbf{H}$ which correspond to the processes in $\mathbf{V}$ and by $\mathbf{H}^{X}$ the last $N+1$ columns which correspond to the processes in $\mathbf{X}$. Therefore from \eqref{def:Y_vector} and \eqref{eq:matrix sde solution2} we get,  
\begin{equation}
  \mathbf{c} = \mathbf{H}^{V}(T-t)\mathbf{V}_{t} + \mathbf{H}^{X}(T-t)\mathbf{X}_{t} + \int_{t}^{T} \mathbf{H}(T-s) d\mathbf{S}_{s}.
\end{equation}
By Assumption \ref{assum:bounded}(i), $\mathbf{H}^{V}(t)$ is invertible for any $t\in [0,T]$ so we have, 
 \begin{equation} \label{eq:feedback-pf}
    \mathbf{V}^{}_{t} = \left(\mathbf{H}^{V}(T-t)\Big)^{-1}\Big(\mathbf{c}-\mathbf{H}^{X}(T-t)\mathbf{X}_{t}-\int_{t}^{T} \mathbf{H}(T-s) d\mathbf{S}_{s}\right), \quad 0\leq t \leq T.  
\end{equation}
Taking conditional expectation with respect to $\mathcal F_t$ on both sides of \eqref{eq:feedback-pf} and using the fact that $\mathbf{S}$ is a martingale we get \eqref{eq:feedback}. 

\textbf{Step 3: Admissibility of  \eqref{eq:feedback}.}
From \eqref{eq:feedback} and the triangle inequality we get,
\begin{equation}
\begin{split}
    \mathbb{E}\left[\norm*{\mathbf{V}_{t}}^{2}\right] 
    &=\mathbb{E}\left[\norm*{\left(\mathbf{H}^{V}(T-t)\right)^{-1}\mathbf{c} - \left(\mathbf{H}^{V}(T-t)\right)^{-1}\mathbf{H}^{X}(T-t)\mathbf{X}_{t}}^{2}\right] \\
    &\leq \mathbb{E}\left[\norm*{\left(\mathbf{H}^{V}(T-t)\right)^{-1}\mathbf{c}}^{2}\right] +\E\left[ \norm*{\left(\mathbf{H}^{V}(T-t)\right)^{-1}\mathbf{H}^{X}(T-t)\mathbf{X}_{t}}^{2}\right].
\end{split}
\end{equation}
Then, using Assumption \ref{assum:bounded}(i),  \eqref{def:Y}, \eqref{def:X}, \eqref{def:X_vector}, Cauchy-Schwarz inequality and H\"older inequality we get,    
\begin{equation} \label{gr1} 
\begin{split}
    \mathbb{E}\left[\norm*{\mathbf{V}_{t}}^{2}\right]
    &\leq C_{1} + C_{2}\mathbb{E}\left[\norm*{\mathbf{X}_{t}}^{2}\right] \\
     &\leq C_{1} + \wt C_{2} \big(\sum_{i=1}^{N}\mathbb{E}\left[ ({X}^i_{t})^{2}\right] + \mathbb{E}\left[ {I}_{t}^{2}\right]  \big)\\
     &\leq \wt C_{1} + \wt C_{2} \sum_{i=1}^{N}\int_0^t\mathbb{E}\left[ (q^i_{s})^{2}\right]ds   \\
            &\leq \wt C_{1} + \wt C_{2} \int_0^t\mathbb{E}\left[ \norm*{\mathbf{V}_{s}}^{2}\right] ds, 
     \end{split}
\end{equation}
for some constants $\wt C_{1}, \wt C_{2}>0$. Note that we have used \eqref{v-i-def} and \eqref{def:V_vector} in the last inequality. 
From \eqref{gr1} we get 
\begin{equation}  
\begin{split}
\sup_{s\in [0,t]}   \mathbb{E}\left[\norm*{\mathbf{V}_{s}}^{2}\right] \leq \wt C_{1} + \wt C_{2} \int_0^t \sup_{r\in [0,s]}\mathbb{E}\left[ \norm*{\mathbf{V}_{r}}^{2}\right] ds, \quad \textrm{for all } 0\leq t\leq T. 
     \end{split}
\end{equation}
Hence by applying Gronwall's Lemma we get, 
\begin{equation}
    \underset{0\leq t\leq T}{\sup} \mathbb{E}\left[\norm*{\mathbf{V}_{t}}^{2}\right] < \infty,
\end{equation}
which in turn implies by \eqref{v-i-def} and \eqref{def:V_vector} that 
\begin{equation}
      \sup_{0\leq t\leq T}  \left( \mathbb{E}\big[\big(q^{i}_{t}\big)^{2}\big]  +  \mathbb{E}\big[\big(\delta^{a, i}_{t}\big)^{2}\big] + \mathbb{E}\big[\big(\delta^{b, i}_{t} \big)^{2}\big] \right)< \infty, \quad \textrm{for all } i\in [N]. 
 \end{equation}
 We conclude that the solution to \eqref{eq:FBSDE_aggregated} is admissible by \eqref{adms} and \eqref{adms2}. 
\end{proof}

\section{\texorpdfstring{The Matrices $K$ and $T$ from Section \ref{sec:model}}{The Matrices K and T from Section \ref{sec:model}}}\label{appendix:matrices}

\subsection{\texorpdfstring{Definition of $\mathbf{K}\in\mathbb{R}^{(5N+2)\times(5N+2)}$}{Definition of K}}
In this section, we define the matrices $\mathbf{K}$ and $\mathbf{T}$. These are constructed out of several other matrices. Define $K(i, j)\in\mathbb{R}^{4\times4}$ for $i, j\in[N]$ by
\begingroup
\begin{equation}
\begin{split}
    &K(i, j) = 
    \begin{pmatrix}
        0 & 0 & -\lambda  & 0 \\
        0 & 0 & \lambda  & 0 \\
        \frac{\rho^{a}\lambda\kappa(N-1)}{\epsilon^{i}N^{2}} & -\frac{\rho^{b}\lambda\kappa(N-1)}{\epsilon^{i}N^{2}} & 0 & \frac{\beta}{\epsilon^{i}} \\
        \frac{\rho^{a}\lambda\kappa(N-1)}{N^{2}} & -\frac{\rho^{b}\lambda\kappa(N-1)}{N^{2}} & \lambda  & \beta \\
    \end{pmatrix}
    \qquad \text{if} \quad i=j, \\
    &K(i, j) =
    \begin{pmatrix}
       0 & 0 & -\lambda  & 0 \\
        0 & 0 & \lambda  & 0 \\
        -\frac{\rho^{a}\lambda\kappa}{\epsilon^{i}N^{2}} & \frac{\rho^{b}\lambda\kappa}{\epsilon^{i}N^{2}} & -\frac{\lambda}{\epsilon^{i}} & 0 \\
        -\frac{\rho^{a}\lambda\kappa}{N^{2}} & \frac{\rho^{b}\lambda\kappa}{N^{2}} & 0 & 0
    \end{pmatrix}
    \qquad \text{if} \quad i\neq j, \\
\end{split}
\end{equation}
\endgroup
\begingroup
\begin{equation}
    \bm{\varphi}(i, j) =
    \begin{dcases}
        \begin{pmatrix}
            2\phi^{j} \\
            -2\phi^{j} \\
            \frac{2\phi^{j}}{\epsilon^{j}} \\
            0
        \end{pmatrix} 
        \qquad \text{if} \quad i = j \\
        \begin{pmatrix}
            \frac{2N}{2N-1}\phi^{j} \\
            -\frac{2N}{2N-1}\phi^{j} \\
            0 \\
            0
        \end{pmatrix}
        \qquad \text{if} \quad i \neq j,
    \end{dcases}
\end{equation}
\endgroup
\begingroup
\begin{equation}
    \check{\bm{\varphi}}(j) = 
    \begin{pmatrix}
        \beta \\
        -\beta \\
        \frac{\beta}{\epsilon^{j}} \\
        0
    \end{pmatrix}
\end{equation}
\begin{equation} \label{def:Phi_i}
    \Phi^{i} = \big(
         \bm{\varphi}(i, 1) , \bm{\varphi}(i, 2) , \bm{\varphi}(i, 3) ,
        ... ,
        \bm{\varphi}(i, N-1) , 
        \bm{\varphi}(i, N) ,
        \check{\bm{\varphi}}(i) \big)^\top
     \in\mathbb{R}^{4\times(N+1)}
\end{equation}
\endgroup
\begingroup
\begin{equation}
    \bm{\vartheta}(i, j) =
    \begin{dcases}
            \big( \frac{\rho^{a}\kappa(N-1)}{N^{2}},  -\frac{\rho^{b}\kappa(N-1)}{N^{2}} , 1 , 0 \big)^\top
         \qquad \text{if} \quad i=j,\\ 
             (-\frac{\rho^{a}\kappa}{N^{2}}, \frac{\rho^{b}\kappa}{N^{2}} , 0 , 0)^\top
        \qquad \text{if} \quad i \neq j,
    \end{dcases}
\end{equation}
\endgroup
and
\begingroup
 \begin{equation} \label{def:vartheta_tilde}
    \tilde{\bm{\vartheta}} = 
(        0 , 0 , \lambda , 0 ),
\end{equation}
\endgroup
\begingroup
 \begin{equation}
    \Theta^{j} = 
    \begin{pmatrix}
        \bm{\vartheta}(1, j) \\
        \bm{\vartheta}(2, j) \\
        \bm{\vartheta}(3, j) \\
        \vdots \\
        \bm{\vartheta}(N-1, j)  \\ 
        \bm{\vartheta}(N, j) \\
        \tilde{\bm{\vartheta}}
    \end{pmatrix}
    \in\mathbb{R}^{(N+1)\times4},
\end{equation}
\endgroup
\begingroup
 \begin{equation}
    \check{K}(i) = \Big(
         0 , 0 , \frac{\lambda}{\epsilon^{i}N}(\rho^{b}-\rho^{a}) , \frac{\lambda}{N}(\rho^{b}-\rho^{a})
  \Big)^{\top}
    \in\mathbb{R}^{4},
\end{equation}
\endgroup
\begingroup
\begin{equation}
    \underline{K} = 
    \begin{pmatrix}
        0 & 0 & 0 & \dots & 0 & 0 & 0 \\
        0 & 0 & 0 & \dots & 0 & 0 & 0 \\
        0 & 0 & 0 & \dots & 0 & 0 & 0 \\
        \vdots & \vdots & \vdots & \ddots & \vdots & \vdots & \vdots \\
        0 & 0 & 0 & \dots & 0 & 0 & 0 \\
        0 & 0 & 0 & \dots & 0 & 0 & 0 \\
        0 & 0 & 0 & \dots & 0 & 0 & -\beta \\
    \end{pmatrix}
    \in\mathbb{R}^{(N+1)\times(N+1)},
\end{equation}
\endgroup
and
\begingroup
\begin{equation}
    \check{\Theta} = 
    \begin{pmatrix}
        \frac{\rho^{b}-\rho^{a}}{N} & \frac{\rho^{b}-\rho^{a}}{N} & \frac{\rho^{b}-\rho^{a}}{N} & \dots & \frac{\rho^{b}-\rho^{a}}{N} & \frac{\rho^{b}-\rho^{a}}{N} & 0
    \end{pmatrix}^{\top}
    \in\mathbb{R}^{N+1}.
\end{equation}
\endgroup  
Then $ \mathbf{K}$ is defined as follows: 
\begingroup
\begin{equation} \label{def:K}
    \mathbf{K} = 
    \begin{pmatrix}
        K(1, 1) & K(1, 2) & \dots & K(1, N) & \check{K}(1) & \Phi^{1} \\
        K(2, 1) & K(2, 2) &  \dots & K(2, N) & \check{K}(2) &\Phi^{2} \\
        \vdots & \vdots & \ddots & \vdots & \vdots & \vdots \\
        \vdots & \vdots &  \ddots & \vdots & \vdots & \vdots \\
        K(N-1, 1) & K(N-1, 2) & \dots & K(N-1, N) & \check{K}(N-1) & \Phi^{N-1} \\
        K(N, 1) & K(N, 2) & \dots &  K(N, N) & \check{K}(N) & \Phi^{N} \\ 
        \mathbf{0}_{4}^{\top} & \mathbf{0}_{4}^{\top} & \dots &  \mathbf{0}_{4}^{\top} & 0 & \mathbf{0}_{N+1}^{\top}\\
        \Theta^{1} & \Theta^{2}& \dots &  \Theta^{N} & \check{\Theta} & \underline{K} \\
    \end{pmatrix}.
\end{equation}
\endgroup

\subsection{\texorpdfstring{Definition of $\mathbf{T}\in\mathbb{R}^{(4N+1)\times(5N+2)}$}{Definition of T}}
In order to define $\mathbf{T}$ we need a few preliminary definitions. 
Define $T(i, j)\in\mathbb{R}^{4\times4}$ for $i, j\in[N]$ by
\begingroup
\begin{equation}
\begin{split}
    &T(i, j) = 
    \begin{pmatrix}
        -1 + \frac{1}{2N-1} & 0 & 0 & 0 \\
        0 & -1 + \frac{1}{2N-1} & 0 & 0 \\
        0 & 0 & -1 & 0 \\
        0 & 0 & 0 & -1 \\
    \end{pmatrix}
    \qquad \text{if} \quad i=j, \\
    &T(i, j) =
    \begin{pmatrix}
        \frac{1}{2N-1} & 0 & 0 & 0 \\
        0 & \frac{1}{2N-1} & 0 & 0 \\
        0 & 0 & 0 & 0 \\
        0 & 0 & 0 & 0 \\
    \end{pmatrix}
    \qquad \text{if} \quad i\neq j, \\
\end{split}
\end{equation}
\endgroup
\begingroup
\begin{equation}
    \check{T} = 
        \Big( \frac{N}{\kappa(2N-1)} +\frac{N\mu}{\rho^{a}}, \frac{N}{\kappa(2N-1)} +\frac{N\mu}{\rho^{b}}, 0 , 0
 \Big)^{\top}
    \in\mathbb{R}^{4}, 
\end{equation}
\endgroup

\begingroup
\begin{equation}
    \bm{\psi}(i, j) = 
    \begin{dcases}
        \begin{pmatrix}
            -2\alpha^{j}\frac{N-1}{2N-1} \\
            2\alpha^{j}\frac{N-1}{2N-1} \\
            -\frac{2\alpha^{j}}{\epsilon^{j}} \\
            0
        \end{pmatrix}
        \qquad \text{if} \quad i=j \\
        \mathbf{0}_{4} 
        \qquad \text{if} \quad i \neq j, 
    \end{dcases}
\end{equation}
\endgroup
\begingroup
\begin{equation}
    \tilde{\bm{\psi}}(j) = 
    \begin{pmatrix}
        -\frac{N-1}{2N-1} \\
        \frac{N-1}{2N-1} \\
        -\frac{1}{\epsilon^{j}} \\
        0
    \end{pmatrix}
    \in\mathbb{R}^{4}, 
\end{equation}
\endgroup
and
\begingroup
\begin{equation}
    \Psi^{i} = \big(
         \bm{\psi}(i, 1) , \bm{\psi}(i, 2) , \bm{\psi}(i, 3) , ...
        , ...
        , \bm{\psi}(i, N-1) , \bm{\psi}(i, N) , \tilde{\bm{\psi}}(i) \big)^\top
     \in\mathbb{R}^{4\times(N+1)}. 
\end{equation}
\endgroup
Then $\mathbf{T}$ is given by, 
\begingroup
\begin{equation} \label{def:T}
    \mathbf{T} = 
    \begin{pmatrix}
        T(1, 1) & T(1, 2) & \dots & \dots & T(1, N-1) & T(1, N) & \check{T} & \Psi^{1} \\
        T(2, 1) & T(2, 2) & \dots & \dots & T(2, N-1) & T(2, N) & \check{T} & \Psi^{2} \\
        \vdots & \vdots & \ddots & \ddots & \vdots & \vdots & \vdots & \vdots \\
        \vdots & \vdots & \ddots & \ddots & \vdots & \vdots & \vdots & \vdots \\ 
        T(N-1, 1) & T(N-1, 2) & \dots & \dots & T(N-1, N-1) & T(N-1, N) & \check{T} & \Psi^{N-1} \\    
        T(N, 1) & T(N, 2) & \dots & \dots & T(N, N-1) & T(N, N) & \check{T} & \Psi^{N} \\
        \mathbf{0}_{4}^{\top} & \mathbf{0}_{4}^{\top} & \dots & \dots & \mathbf{0}_{4}^{\top} & \mathbf{0}_{4}^{\top} & 1 & \mathbf{0}_{N+1}^{\top}. 
    \end{pmatrix}
\end{equation}
\endgroup

\section{Additional matrix definitions }\label{appendix:matrices2}
 In this section, we define the matrices $\tilde{\mathbf{K}}, \mathbb{Q}^{a}$ and $\mathbb{Q}^{b}$, which are introduced in \eqref{eq:matrix_FBSDE_Z}. These matrices are constructed out of several other matrices.

\subsection{\texorpdfstring{Definition of $\wt{\mathbf{K}}\in\mathbb{R}^{(5N+2)\times(5N+2)}$}{Definition of K}}

Define $\tilde{K}(i, j)\in\mathbb{R}^{4\times4}$ for $i, j\in[N]$ by
\begingroup
\begin{equation}
\begin{split}
    &\tilde{K}(i, j) = 
    \begin{pmatrix}
        0 & 0 & -\lambda  & 0 \\
        0 & 0 & \lambda  & 0 \\
        0 & 0 & 0 & \frac{\beta}{\epsilon^{i}} \\
       0 & 0 & \lambda  & \beta \\
    \end{pmatrix}
    \qquad \text{if} \quad i=j, \\
    &\tilde{K}(i, j) =
    \begin{pmatrix}
       0 & 0 & -\lambda  & 0 \\
        0 & 0 & \lambda  & 0 \\
        0 & 0 & -\frac{\lambda}{\epsilon^{i}} & 0 \\
        0 & 0 & 0 & 0
    \end{pmatrix}
    \qquad \text{if} \quad i\neq j.
\end{split}
\end{equation}
\endgroup
 
Recall that $\Phi^{i}\in\mathbb{R}^{4\times(N+1)}$ was defined in \eqref{def:Phi_i} and that  $\tilde{\bm{\vartheta}}\in\mathbb{R}^{4}$ was defined in \eqref{def:vartheta_tilde}. Define  
\begingroup
\begin{equation}
    \tilde{\bm{\vartheta}}'(i, j) =
    \begin{dcases}
            ( 0, 0 , 1 , 0)
         \qquad \text{if} \quad i=j,\\ 
             (0 , 0 , 0 , 0)
        \qquad \text{if} \quad i \neq j,
    \end{dcases}
\end{equation}
\endgroup
 and 
\begingroup
 \begin{equation}
    \tilde{\Theta}^{j} = 
    \begin{pmatrix}
        \tilde{\bm{\vartheta}}'(1, j) \\
        \tilde{\bm{\vartheta}}'(2, j) \\
        \tilde{\bm{\vartheta}}'(3, j) \\
        \vdots \\
        \tilde{\bm{\vartheta}}'(N-1, j)  \\ 
        \tilde{\bm{\vartheta}}'(N, j) \\
        \tilde{\bm{\vartheta}}
    \end{pmatrix}
    \in\mathbb{R}^{(N+1)\times4}.
\end{equation}
\endgroup
Let 
\begingroup
\begin{equation}
    \underline{K} = 
    \begin{pmatrix}
        0 & 0 & 0 & \dots & 0 & 0 & 0 \\
        0 & 0 & 0 & \dots & 0 & 0 & 0 \\
        0 & 0 & 0 & \dots & 0 & 0 & 0 \\
        \vdots & \vdots & \vdots & \ddots & \vdots & \vdots & \vdots \\
        0 & 0 & 0 & \dots & 0 & 0 & 0 \\
        0 & 0 & 0 & \dots & 0 & 0 & 0 \\
        0 & 0 & 0 & \dots & 0 & 0 & -\beta \\
    \end{pmatrix}
    \in\mathbb{R}^{(N+1)\times(N+1)}. 
\end{equation}
\endgroup
Then we define, \begingroup
\begin{equation} 
    \tilde{\mathbf{K}} = 
    \begin{pmatrix}
        \tilde{K}(1, 1) & \tilde{K}(1, 2) & \dots & \tilde{K}(1, N) & \mathbf{0}_{4} & \Phi^{1} \\
        \tilde{K}(2, 1) & \tilde{K}(2, 2) & \dots & \tilde{K}(2, N) & \mathbf{0}_{4} &\Phi^{2} \\
        \vdots & \vdots & \ddots & \vdots & \vdots & \vdots \\
        \vdots & \vdots & \ddots & \vdots & \vdots & \vdots \\
        \tilde{K}(N-1, 1) & \tilde{K}(N-1, 2) & \dots & \tilde{K}(N-1, N) & \mathbf{0}_{4} & \Phi^{N-1} \\
        \tilde{K}(N, 1) & \tilde{K}(N, 2) & \dots & \tilde{K}(N, N) & \mathbf{0}_{4} & \Phi^{N} \\ 
        \mathbf{0}_{4}^{\top} & \mathbf{0}_{4}^{\top} & \dots & \mathbf{0}_{4}^{\top} & 0 & \mathbf{0}_{N+1}^{\top}\\
        \tilde{\Theta}^{1} & \tilde{\Theta}^{2} & \dots & \tilde{\Theta}^{N} & \mathbf{0}_{N+1} & \underline{K} \\
    \end{pmatrix}. 
\end{equation}
\endgroup

\subsection{\texorpdfstring{Definition of $\mathbf{Q}^{a}\in\mathbb{R}^{(5N+2)\times(5N+2)}$}{Definition of Qa}}

Define $Q^{a}(i, j)\in\mathbb{R}^{4\times4}$ for $i, j\in[N]$ by
\begingroup
\begin{equation}
\begin{split}
    &Q^{a}(i, j) = 
    \begin{pmatrix}
        0 & 0 & 0  & 0 \\
        0 & 0 & 0  & 0 \\
        \frac{\lambda\kappa(N-1)}{\epsilon^{i}N^{2}} & 0 & 0 & 0 \\
        \frac{\lambda\kappa(N-1)}{N^{2}} & 0 & 0 & 0 \\
    \end{pmatrix}
    \qquad \text{if} \quad i=j, \\
    &Q^{a}(i, j) =
    \begin{pmatrix}
       0 & 0 & 0  & 0 \\
        0 & 0 & 0  & 0 \\
        -\frac{\lambda\kappa}{\epsilon^{i}N^{2}} & 0 & 0 & 0 \\
        -\frac{\lambda\kappa}{N^{2}} & 0 & 0 & 0
    \end{pmatrix}
    \qquad \text{if} \quad i\neq j. \\
\end{split}
\end{equation}
\endgroup
Let 
\begingroup
\begin{equation}
    \bm{\vartheta}^{a}(i, j) =
    \begin{dcases}
        \begin{pmatrix}
            \frac{\kappa(N-1)}{N^{2}} & 0 & 0 & 0
        \end{pmatrix}
        \qquad \text{if} \quad i=j,\\ 
        \begin{pmatrix}
            -\frac{\kappa}{N^{2}} & 0 & 0 & 0
        \end{pmatrix}
        \qquad \text{if} \quad i \neq j, 
    \end{dcases}
\end{equation}
\endgroup
\begingroup
 \begin{equation}
    \Theta^{a, j} = 
    \begin{pmatrix}
        \bm{\vartheta}^{a}(1, j) \\
        \bm{\vartheta}^{a}(2, j) \\
        \bm{\vartheta}^{a}(3, j) \\
        \vdots \\
        \bm{\vartheta}^{a}(N-1, j)  \\ 
        \bm{\vartheta}^{a}(N, j) \\
        \mathbf(0)_{4}
    \end{pmatrix}
    \in\mathbb{R}^{(N+1)\times4}, 
\end{equation}
\endgroup
\begingroup
 \begin{equation}
    \check{Q}^{a}(i) = \big(        0 , 0 , -\frac{\lambda}{\epsilon^{i}N} , -\frac{\lambda}{N}
\big)^{\top}
    \in\mathbb{R}^{4}, 
\end{equation}
\endgroup
and
\begingroup
\begin{equation}
    \check{\Theta}^{a} = 
 \big(        -\frac{1}{N} , -\frac{1}{N} , -\frac{1}{N} , \dots , -\frac{1}{N} , -\frac{1}{N} ,0
    \big)^{\top}
    \in\mathbb{R}^{N+1}. 
\end{equation}
\endgroup
Then we define, 
\begingroup
\begin{equation}
    \mathbf{Q}^{a} = 
    \begin{pmatrix}
        Q^{a}(1, 1) & Q^{a}(1, 2) & \dots & Q^{a}(1, N) & \check{Q}^{a}(1) & \mathbf{0}_{4\times(N+1)} \\
        Q^{a}(2, 1) & Q^{a}(2, 2) & \dots & Q^{a}(2, N) & \check{Q}^{a}(2) & \mathbf{0}_{4\times(N+1)} \\
        \vdots & \vdots & \ddots & \vdots & \vdots & \vdots \\
        \vdots & \vdots & \ddots & \vdots & \vdots & \vdots \\
        Q^{a}(N-1, 1) & Q^{a}(N-1, 2) & \dots & Q^{a}(N-1, N) & \check{Q}^{a}(N-1) & \mathbf{0}_{4\times(N+1)} \\
        Q^{a}(N, 1) & Q^{a}(N, 2) & \dots & Q^{a}(N, N) & \check{Q}^{a}(N) & \mathbf{0}_{4\times(N+1)} \\ 
        \mathbf{0}_{4}^{\top} & \mathbf{0}_{4}^{\top} & \dots & \mathbf{0}_{4}^{\top} & 0 & \mathbf{0}_{N+1}^{\top}\\
        \Theta^{a, 1} & \Theta^{a, 2} & \dots & \Theta^{a, N} & \check{\Theta}^{a} & \mathbf{0}_{(N+1)\times(N+1)} \\
    \end{pmatrix}. 
\end{equation}
\endgroup

\subsection{\texorpdfstring{Definition of $\mathbf{Q}^{b}\in\mathbb{R}^{(5N+2)\times(5N+2)}$}{Definition of Qb}}

Define $Q^{b}(i, j)\in\mathbb{R}^{4\times4}$ for $i, j\in[N]$ by
\begingroup
\begin{equation}
\begin{split}
    &Q^{b}(i, j) = 
    \begin{pmatrix}
        0 & 0 & 0  & 0 \\
        0 & 0 & 0  & 0 \\
        0 & -\frac{\lambda\kappa(N-1)}{N^{2}} & 0 & 0 \\
        0 & -\frac{\lambda\kappa(N-1)}{N^{2}} & 0  & 0 \\
    \end{pmatrix}
    \qquad \text{if} \quad i=j, \\
    &Q^{b}(i, j) =
    \begin{pmatrix}
       0 & 0 & 0  & 0 \\
        0 & 0 & 0  & 0 \\
       0 & \frac{\lambda\kappa}{\epsilon^{i}N^{2}} & 0 & 0 \\
       0 & \frac{\lambda\kappa}{N^{2}} & 0 & 0
    \end{pmatrix}
    \qquad \text{if} \quad i\neq j.  
\end{split}
\end{equation}
\endgroup
Let, 
\begingroup
\begin{equation}
    \bm{\vartheta}^{b}(i, j) =
    \begin{dcases}
        \begin{pmatrix}
            0 & -\frac{\kappa(N-1)}{N^{2}} & 1 & 0
        \end{pmatrix}
        \qquad \text{if} \quad i=j,\\ 
        \begin{pmatrix}
            0 & \frac{\kappa}{N^{2}} & 0 & 0
        \end{pmatrix}
        \qquad \text{if} \quad i \neq j, 
    \end{dcases}
\end{equation}
\endgroup
\begingroup
 \begin{equation}
    \Theta^{b, j} = 
    \begin{pmatrix}
        \bm{\vartheta}^{b}(1, j) \\
        \bm{\vartheta}^{b}(2, j) \\
        \bm{\vartheta}^{b}(3, j) \\
        \vdots \\
        \bm{\vartheta}^{b}(N-1, j)  \\ 
        \bm{\vartheta}^{b}(N, j) \\
        \mathbf{0}_{4}
    \end{pmatrix}
    \in\mathbb{R}^{(N+1)\times4}, 
\end{equation}
\endgroup
\begingroup
 \begin{equation}
    \check{Q}^{b}(i) = 
 \big(        0 , 0 ,\frac{\lambda}{\epsilon^{i}N} , \frac{\lambda}{N}
    \big)^{\top}
    \in\mathbb{R}^{4}, 
\end{equation}
\endgroup
and
\begingroup
\begin{equation}
    \check{\Theta}^{b} = 
 \big(
        \frac{1}{N} ,\frac{1}{N} , \frac{1}{N} , \dots , \frac{1}{N} , \frac{1}{N} , 0
 \big)^{\top}
    \in\mathbb{R}^{N+1}. 
\end{equation}
\endgroup
Then finally we define
\begingroup
\begin{equation} 
    \mathbf{Q}^{b} = 
    \begin{pmatrix}
        Q^{b}(1, 1) & Q^{b}(1, 2) & \dots & Q^{b}(1, N) & \check{Q}^{b}(1) & \mathbf{0}_{4\times(N+1)} \\
        Q^{b}(2, 1) & Q^{b}(2, 2) & \dots & Q^{b}(2, N) & \check{Q}^{b}(2) & \mathbf{0}_{4\times(N+1)} \\
        \vdots & \vdots & \ddots & \vdots & \vdots & \vdots \\
        \vdots & \vdots & \ddots & \vdots & \vdots & \vdots \\
        Q^{b}(N-1, 1) & Q^{b}(N-1, 2) & \dots & Q^{b}(N-1, N) & \check{Q}^{b}(N-1) & \mathbf{0}_{4\times(N+1)} \\
        Q^{b}(N, 1) & Q^{b}(N, 2) & \dots & Q^{b}(N, N) & \check{Q}^{b}(N) & \mathbf{0}_{4\times(N+1)} \\ 
        \mathbf{0}_{4}^{\top} & \mathbf{0}_{4}^{\top} & \dots & \mathbf{0}_{4}^{\top} & 0 & \mathbf{0}_{N+1}^{\top}\\
        \Theta^{b, 1} & \Theta^{b, 2} & \dots & \Theta^{b, N} & \check{\Theta}^{b} & \mathbf{0}_{(N+1)\times(N+1)}
    \end{pmatrix}.
\end{equation}
\endgroup

\end{appendices}

\pagebreak

\section*{Acknowledgements}

We are especially grateful to Roel Oomen for his contributions to the development of this work. His insights, feedback, and support were invaluable in shaping the final manuscript. We additionally thank Philipp Plank for helpful discussion regarding the $\epsilon$-optimality test in Remark \ref{remark:optimality-test}.

\bibliographystyle{plainnat}
\bibliography{prisoners}

\end{document}